%% file: ccuda-pldi.tex
\newif\ifcomments
\newif\ifextended
\commentstrue{}

\documentclass[acmsmall,screen,nonacm]{acmart}

\usepackage[frozencache,cachedir=minted-cache]{minted}
\usepackage[many]{tcolorbox}
\tcbuselibrary{minted,skins,breakable}
\tcbset{listing engine=minted}
\usepackage{proof}
\usepackage{tabularx}
\usepackage{style}
\usepackage{enumitem}
\usepackage{listings}
\usepackage{xcolor}
\usepackage{minted, caption}
\usepackage{adjustbox}
\usepackage{wrapfig}
\usepackage{multicol}
\usepackage[subtle]{savetrees}

\makeatletter
\let\@authorsaddresses\@empty
\makeatother

\AtBeginDocument{%
  \providecommand\BibTeX{{%
\normalfont B\kern-0.5em{\scshape i\kern-0.25em b}\kern-0.8em\TeX}}}

\newminted[ccudaCode]{python}{
  linenos,
  autogobble,
  breaklines,
  baselinestretch=1.0,
  fontsize=\tiny,
  numbersep=3pt,
  frame=single,
}

\newtcblisting{ccudaBox}{
  minted language=python,
  minted options={
    linenos,
    autogobble,
    breaklines,
    baselinestretch=1.0,
    fontsize=\tiny,
    numbersep=3pt,
    frame=single,
  },
  before upper={\dontdofcolorbox},
  colback=white,     
  colframe=white,    
  boxrule=0pt,       
  left=0pt, right=0pt, top=0pt, bottom=0pt, 
  listing only,
  enhanced,
  sharp corners,
}

\newtcblisting{cudaBox}{
  minted language=CUDA,
  minted options={
    linenos,
    autogobble,
    breaklines,
    baselinestretch=1.0,
    fontsize=\tiny,
    numbersep=3pt,
    frame=single,
  },
  colback=white,     
  colframe=white,    
  boxrule=0pt,       
  left=0pt, right=0pt, top=0pt, bottom=0pt, 
  listing only,
  enhanced,
  sharp corners,
}

\makeatletter
\AtBeginEnvironment{minted}{\dontdofcolorbox}
\def\dontdofcolorbox{\renewcommand\fcolorbox[4][]{##4}}
\makeatother

\makeatletter
\AtBeginEnvironment{minted}{\dontdofcolorbox}
\AtBeginEnvironment{tcb@listing@minted}{\dontdofcolorbox}
\def\dontdofcolorbox{\renewcommand\fcolorbox[4][]{##4}}
\makeatother

\makeatletter
\renewcommand\fcolorbox[4][]{#4}
\makeatother

\newcommand{\unit}{perspective\xspace}
\newcommand{\Unit}{Perspective\xspace}
\newcommand{\name}{Prism\xspace}
\newcommand{\calcname}{Bundl\xspace}

\begin{document}
\title{Modular GPU Programming with Typed \Unit{}s}

\newcommand{\mb}[1]{\ifcomments\textcolor{red}{[\textbf{MB:}~#1]}\fi}
\newcommand{\dhs}[1]{\ifcomments\textcolor{orange}{[\textbf{DHS:}~#1]}\fi}
\newcommand{\jwc}[1]{\ifcomments\textcolor{green}{[\textbf{JWC:}~#1]}\fi}
\newcommand{\refline}[1]{\textcolor{black}{{#1}}}
\newcommand{\cn}{\ifcomments\textcolor{red}{\textbf{CITATION NEEDED}}\xspace\fi}
\newcommand{\saman}[1]{\ifcomments\textcolor{magenta}{[\textbf{Saman:}~#1]}\fi}

\newcommand{\tempunit}{\ifcomments\textcolor{brown}{\fi\unit{}\xspace\ifcomments}\fi}
\newcommand{\TempUnit}{\ifcomments\textcolor{brown}{\fi\Unit{}\xspace\ifcomments}\fi}

\newcommand{\code}[1]{\textbf{\texttt{#1}}}
\newcommand{\partition}{partition}
\newcommand{\ownpartition}{claim}

\author{\href{https://manya-bansal.github.io/}{Manya Bansal}}
\email{manya227@mit.edu}
\affiliation{
  \institution{Massachusetts Institute of Technology}
  \city{Cambridge}
  \state{Massachusetts}
  \country{United States of America}
}

\author{\href{https://www.sainati.pl/}{Daniel Sainati}}
\email{sainati@seas.upenn.edu}
\affiliation{
  \institution{University of Pennsylvania}
  \city{Philadelphia}
  \state{Pennsylvania}
  \country{United States of America}
}

\author{\href{https://www.cis.upenn.edu/\%7e\%6a\%77\%63/}{Joseph W. Cutler}}
\email{jwc@seas.upenn.edu}
\affiliation{
  \institution{University of Pennsylvania}
  \city{Philadelphia}
  \state{Pennsylvania}
  \country{United States of America}
}

\author{\href{https://people.csail.mit.edu/\%73\%61\%6D\%61\%6E/}{Saman Amarasinghe}}
\affiliation{
  \institution{Massachusetts Institute of Technology}
  \city{Cambridge}
  \state{Massachusetts}
  \country{United States of America}
}

\author{\href{https://people.csail.mit.edu/\%6A\%72\%6B/}{Jonathan Ragan-Kelley}}
\affiliation{
  \institution{Massachusetts Institute of Technology}
  \city{Cambridge}
  \state{Massachusetts}
  \country{United States of America}
}

\begin{abstract}
  \input{0-abstract.tex}
\end{abstract}

\maketitle
\input{1-intro-v2.tex}
\input{2-background-motivation.tex}
\input{3-langauge.tex}
\input{4-core-calculus.tex}
\input{5-compilation.tex}
\input{6-evaluation.tex}

\input{8-related-work.tex}
\input{7-limitations.tex}

\input{10-ack.tex}

\bibliographystyle{ACM-Reference-Format}
\bibliography{references}

\newpage

\appendix
\input{appendix-proof.tex}
\input{eval/code-appendix.tex}

\end{document}

%% file: 0-abstract.tex
To achieve peak performance on modern GPUs, one must balance two frames of
mind: issuing instructions to individual threads to control their behavior,
while simultaneously tracking the convergence of many threads acting in concert
to perform \emph{collective operations} like Tensor Core instructions. The tension between
these two mindsets makes modular programming error prone. Functions that
encapsulate collective operations, despite being called per-thread, must
be executed cooperatively by groups of threads. 

In this work, we introduce Prism, a new GPU language that restores
modularity while still giving programmers the low-level control over collective
operations necessary for high performance.  Our core idea is \emph{typed
perspectives}, which materialize, at the type level, the granularity at which the
programmer is controlling the behavior of threads. We describe the design of
Prism, implement a compiler for it, and lay its theoretical foundations in a
core calculus called Bundl. We implement state-of-the-art GPU kernels in
Prism and find that it offers programmers the safety guarantees needed to
confidently write modular code without sacrificing performance.



%% file: 1-intro-v2.tex
\vspace*{-0.5em}
\section{Introduction}\label{sec:intro}


CUDA~\cite{nvidia-cuda-2025} is a low-level, imperative programming language
for NVIDIA GPUs. These GPUs are organized into a hierarchy of compute resources.
\emph{Threads} are the basic unit of sequential execution, \emph{blocks} are
groups of threads that can cooperate through shared scratchpad memory, and the
\emph{grid} is the full collection of blocks launched for a computation. A
GPU program executes in parallel across this hierarchy, but is written \emph{from the
perspective of a single thread}.

While operations are specified per-thread, some are only
valid when executed \textit{collectively} by a group of threads. The
\code{\_\_syncthreads()} intrinsic, for instance, synchronizes all threads
within a block, and it will cause a deadlock if executed by only some of those threads. There are many such collective intrinsics, including Tensor
Core instructions ~\cite{NVIDIAPTXISAWarpLevelMatrix,
NVIDIAPTXISAAsyncWarpGroup, NVIDIAPTXISATCGen05}, warp shuffle operations \cite{warp-shuffle},
and other kinds of synchronization primitives \cite{sync-functions}. These operations
require programmers to carefully marshal compute resources to
coordinate which threads execute which lines of code. As a result, such
collective operations break the illusion of threads executing
independently. 

The conceptual clash between regular statements (executed by a single thread)
and collective operations (executed by a group of threads) impacts not only
hardware intrinsics, but also \emph{user-defined} collective operations, like functions.  
While functions are invoked individually by threads, they may encapsulate collective
behavior, creating a contradiction between the per-thread syntax of their
invocation and the cooperative semantics of their execution.

This contradiction puts modularity at odds with correctness and is apparent even in widely
used CUDA libraries that package common functionality via function interfaces 
\cite{nvidia-cub-2025, Thakkar_CUTLASS_2023,cublasdx}. Consider the
following snippet of documentation taken directly
from CUB ~\cite{nvidia-cub-2025}, a library of
parallel primitives. It describes the
\code{BlockReduce} function \cite{blockreduce}, in which all threads in a block
collaboratively apply a reduction operation, such as a maximum or prefix sum,
over an array:

\begin{figure}[h]
  \vspace*{-1em}
  \raggedright
  \colorbox{gray!20}{%
  \parbox{\dimexpr\linewidth-2\fboxsep}{
      \footnotesize
      Computes a \underline{block-wide} reduction for thread$_0$ using the specified binary reduction
      functor.
      \begin{itemize}
        \item The return value is undefined in threads other than thread$_0$.
        \item A subsequent \code{\_\_syncthreads()} threadblock barrier should be
          invoked after calling this method if the collective's temporary storage (e.g.,
          \code{temp\_storage}) is to be reused or repurposed.
      \end{itemize}
  }%
  }
  \vspace*{-1em}
  \end{figure}

The snippet above attempts to convey several
assumptions implicit in \code{BlockReduce}'s implementation.
First, the function is only well defined when invoked by all threads in a block.
Second, because it accesses memory shared among threads, any
subsequent reuse of that storage requires synchronization to ensure all
threads have completed their accesses.
In effect, CUB attempts to retrofit CUDA with information about the compute and
memory requirements of collective operations.  By prefixing functions
with identifiers such as \code{Block}, CUB creates ad-hoc ``namespaces'' for
different functions that assume similar invariants.  However, without a type
system to statically enforce these invariants, correct usage of this
function---and collective operations in general---depends on carefully reading and interpreting the
documentation.





In this work, we ask: can we provide low-level control over
collective operations while statically guaranteeing that they execute with
the necessary compute resources?
By reifying configurations of compute resources as type-level
\emph{\unit{}s}---so named
because they describe the view of GPU resources against which each statement is written---we find that we can.
Our key insight is that GPU programmers naturally map computations onto different
compute resources, and, unlike CUDA, which obscures this mapping, we can track
it with a type system. This tracking enforces that collective operations are
executed with the necessary resources, while still allowing users to access low-level intrinsics. 
Further, it allows users to define and compose their own collective operations 
by specifying the \unit with which their code must be run. 

Our approach departs from previous work, which attempts to resolve the
mismatch between per-thread syntax and collective execution by restricting
access to low-level operations.
Tile-based languages like Triton~\cite{triton},
Helion~\cite{PyTorch_Helion2025}, and
Tilus~\cite{ding2025tilustilelevelgpgpuprogramming}, limit the user's
perspective to the block level, which prevents them from writing the
highest-performance kernels. A variety of functional languages,
\cite{henriksen-futhark-2017,accelerate-grover,Elliott-Vertigo} provide
compile-time guarantees through their type systems but do not expose low-level
control over hardware. Some other efforts, like Descend
\cite{kopcke-descend-2024}, aim to provide a low-level, memory-safe GPU systems
programming language, but lack support for modern GPU features (like Tensor
Cores or asynchrony) entirely. As a result, despite its flaws, CUDA remains the de-facto standard for
writing high-performance kernels on modern GPUs \cite{deepseek}.

We introduce \name{}, a new low-level GPU programming language that enforces
that operations are only executed with the correct view of hardware resources, enabling
fearless composition.
Inspired by dependency calculi~\cite{abadi-core-1999}, \name{} statically tracks the configurations of compute
and memory resources with type-level \unit{}s. We also develop \calcname, a core
calculus underpinning \name{}, provide formal rules for manipulating the \unit{}s of
both code and data, and prove \emph{type-and-\unit{} safety}. This ensures that
\calcname is sound, and that code always has the right \unit{} to execute
operations at run time. A parallel goal, alongside safe composition, is performance. We
incorporate modern GPU features such as Tensor Cores and asynchronous data
movement into \name{}, and demonstrate that \name{} can achieve the same
performance as hand-written, highly optimized code on an H100 and a 4070 Ti
Super. Our contributions are:
\begin{enumerate}
  \item \name{}, a low-level GPU language that tracks the logical grouping of compute and memory resources with type-level \unit{}s,
  empowering users to write modular code (\Cref{sec:language});
  \item \calcname{}, a formal model of \name{} that tracks \unit{}s in its type system and operational semantics, along with a soundness theorem guaranteeing that programs execute operations only when they have been statically proven to possess the necessary \unit{} (\Cref{sec:core}); and
  \item An implementation of \name{} (Section \ref{sec:compilation}) that demonstrates that it can support modern 
  GPU features like Tensor Cores and asynchronous data movement, achieving performance comparable to hand-optimized CUDA implementations (\Cref{sec:eval}).
\end{enumerate}

Section \ref{sec:related-work} discusses related work, and 
Section \ref{sec:limitations} discusses the limitations of our approach and outlines future work.  

%% file: 2-background-motivation.tex
\vspace*{-1em}
\section{Background \& Motivation}
\label{sec:background-motiv}



Before diving into the design of \name, we begin with an overview of the GPU's
compute and memory hierarchies and outline the challenges posed by reasoning
about them collectively. We also discuss how \name{} can help solve these
problems.

\subsection{Compute Hierarchy}\label{sec:compute-hierarchy}

\input{figs/2-panel}

In CUDA, programmers launch computations that run on thousands of threads. These
threads are organized into a \textit{compute hierarchy} that defines how work is
distributed and scheduled on the GPU. At the top of this hierarchy is the
\textit{grid}, representing all threads launched as part of a single
computation. The grid is divided into \textit{blocks}, each containing a
user-specified number of threads. \textit{Threads} are the finest unit
of execution and the machine's basic unit of sequential control.  

Users describe a computation by writing a single program which
is replicated across all threads. This program controls the behavior of
individual threads by reading built-in identifiers like \code{threadIdx.x}---to
determine a thread's position within a block---and \code{blockIdx.x}---to
determine the block's position within the grid---at runtime. By making control
flow decisions based on these two variables, users can assign each thread to its
share of the full computation.

A natural and tempting way to interpret these built-in variables is to think of
them as the indices of implicit "parallel-for loops" surrounding the
program, where each iteration executes simultaneously. While this view  
is sufficient to understand CUDA's programming model when threads do independent
work, it quickly breaks down in the presence of collective operations.

Unlike most instructions, which are executed by a single thread, collective
operations must be executed collaboratively by a group of threads. For example,
on line \refline{14} of ~\Cref{fig:cuda-simple-mma}, we perform a
Tensor Core operation~\cite{NVIDIAPTXISAWarpLevelMatrix}, which is only
meaningful when invoked by a collection of 32 threads---a
\textit{warp}---acting together. For the call, each thread sets up its portion
of the operands, and the operation is performed \textit{once} for the entire
warp, with the results scattered across participating threads. In this example, the first 32
threads in a block are tasked with this operation; programmers mentally group
these 32 threads into a collective unit, and then reflect that grouping in the
program via the \code{if}-statement on line \refline{2}. If the condition on
line \refline{2} were instead \code{tid >= 0 \&\& tid < 30}, making fewer than
32 threads reach line \refline{14}, the result would be undefined. To make
matters worse, the restriction is not only on the {number} of threads executing
the operation, but also on their {alignment}. For threads to form a warp, the starting ID of
the group must be aligned to a multiple of 32. So,
if the condition on line \refline{2} were instead \code{tid >= 1 \&\& tid < 33},
the result would still be invalid, even though 32 threads would execute it.

Collective operations make reasoning about CUDA programs challenging because
they force programmers to track the \emph{convergence behavior} of threads.
Specifically, programmers must reason both about how many threads reach
each in the program and how those threads are arranged, accounting
for alignment. This difficulty is further amplified by two factors. First,
programs may have multiple points of convergence, requiring programmers
to mentally track the relative ID of a thread within a logical group as that
group evolves over the course of a program's execution. Second, threads may be
participating in multiple levels of convergence within the same program.
In our example, we only considered the warp-level Tensor Core operation, but
there are other collective operations that require convergence at different
granularities like \textit{warpgroup}-level Tensor Core operations---which
must be issued collectively by four warps---or the block-level
\code{\_\_syncthreads()} primitive---which must be executed by
all threads within a block. 

Reasoning about collective operations is already error prone within the context
of a single function, but becomes even more difficult when reasoning
across functions. A callee may assume a certain configuration
of threads and may structure its computation around that assumption. However, such assumptions are not visible in
the callee's function interface, which only exposes the input and output types.
To invoke the function correctly, users must read its documentation---or worse,
its implementation---to understand its assumptions, breaking modular
reasoning.

In \name{}, we materialize the programmer's mental grouping of 
the compute hierarchy explicitly in the program's source. Consider the example in
~\Cref{fig:ccuda-simple-mma}, which shows an equivalent rewrite of the CUDA
program in ~\Cref{fig:cuda-simple-mma} using \name{}. The call to the \code{mma} is valid only
because it is executed within a \code{group} of 32 threads, made explicit on
line \refline{2}. 
This intrinstic exposes its invariant via its function signature, which we will discuss in detail in ~\Cref{sec:collective}.
Since \name{}'s \code{group} construct enforces both the size
and alignment of participating threads, the validity of the \code{mma} operation
is guaranteed by construction. 

However, it is not sufficient to only consider the compute hierarchy
when reasoning about convergence; we must consider the memory hierarchy as well. Whenever threads see different data,
they can use it to induce divergent behavior. We've
already seen an example of this: the \code{tid} variable in
~\Cref{fig:cuda-simple-mma}. In general, the possibility of divergence
lurks anywhere threads branch on data.

\subsection{Data and the Memory Hierarchy}\label{sec:mem-hierarchy}

Data on the GPU is organized into a \textit{memory hierarchy}, mirroring the
compute hierarchy. All threads in a grid can read and write from \textit{global
memory}, where operands reside at the start of a computation and where
results are eventually written. Each block has access to a limited amount of
\textit{shared memory}, a programmer-managed scratchpad typically used to stage
repeatedly-used data. Finally, each thread maintains its own state in
\textit{registers} and \textit{local memory}. These are private to each thread,
while shared and global memory are accessible by multiple threads.

\paragraph{Thread-Local Data}

Registers and local memory are owned by individual threads, 
so variables with the same name can have different values on different threads at run
time. In general, CUDA offers no support for dealing with the divergence that arises
as a consequence.

A key challenge for \name{} is managing this divergence, ensuring that all
threads organized by other language constructs like \code{group} remain
logically unified, even in the face of data-dependent control flow. To do this,
\name{} tracks the \emph{frequency} with which values vary in space.  The
\code{@} syntax attached to each declaration denotes this frequency. For example,
\code{thread[1]} variables can vary for every thread, while \code{block[1]} variables can
vary for every block, but not for threads within that block. Intuitively, the
\code{@} construct ``colors'' a variable across threads, and any threads that
share a color are required to agree on that variable's value.
~\Cref{fig:freq-color} illustrates valid and invalid colorings of a variable.

Using this information, \name{} enforces rules for reading from and writing to
variables.  So, for example, if we introduce a condition based on
\code{tid} {inside} the \code{group}'s scope in
~\Cref{fig:ccuda-simple-mma}, \name{} will reject that program;
code that diverges at low frequency cannot read from variables that change at
higher frequency.  We explain the rules of \name{} programs
in detail in ~\Cref{sec:language}.

\begin{wrapfigure}{r}{0.35\textwidth}
  \vspace{-2em}
  \begin{cudaBox}
  int x = threadIdx.x; 
  if (x >= 0 && x< 32){
    float A[4], B[2];
    float C[4] = { 0 };
          
    // k is a_mem's stride.
    A[0] = a_mem[(x/4)*k+(x
    A[1] = a_mem[(x/4)*k+(x
    A[2] = a_mem[(x/4+8)*k+(x
    A[3] = a_mem[(x/4+8)*k+(x
          
    // n is b_mem's stride.
    B[0] = b_mem[(x
    B[1] = b_mem[(x
          
    asm("mma.sync...");
          
    // Write back into c_mem.
    // n is c_mem's stride.
    c_mem[(x/4)*n+2*(x
    c_mem[(x/4)*n+2*(x
    c_mem[(x/4+8)*n+2*(x
    c_mem[(x/4+8)*n+2*(x
  
    __syncwarp();}
  \end{cudaBox}
  \vspace{-2em}
  \caption{Invoking a warp-level Tensor Core instruction in CUDA after loading data from memory.}
  \label{fig:cuda-simple-mma-mem}
    \vspace{-0.5em}
  \begin{ccudaBox}
  @prism("device")
  @requires(thread[32])
  def simple_mma(
      a: ptr(const(float)) @ thread[32], 
      b: ptr(const(float)) @ thread[32],
      c: ptr(float) @ thread[32]):
    x : int @ thread[1] = id()
    with group(thread[32]):
      A : float[4] @ thread[1]
      B : float[2] @ thread[1]
      C : float[4] @ thread[1]
      # Reads do not need to be lowered. 
      A[0] = a_mem[(x/4)*k+(x
      A[1] = a_mem[(x/4)*k+(x
      A[2] = a_mem[(x/4+8)*k+(x
      A[3] = a_mem[(x/4+8)*k+(x
      B[0] = b_mem[(x
      B[1] = b_mem[(x
      # Skipping initialize C to 0... 
      intrinsic.mma(
        A[0], A[1], A[2], A[3],  
        B[0], B[1],
        C[0], C[1], C[2], C[3],
        out=[C[0], C[1], C[2], C[3]])
      # Must be at thread[1] to write.
      idx = \
        lambda ro, co: (x/4+ro)*n+2*(x
      with partition(c_mem,p=thread[32],f=idx) as c_thrd:
        c_thrd[0, 0] = C[0]
        c_thrd[0, 1] = C[1]
        c_thrd[8, 0] = C[2]
        c_thrd[8, 1)] = C[3]
      # --- Sync point ---
  \end{ccudaBox}
  \vspace{-2em}
  \caption{Invoking a warp-level Tensor Core instruction in \name{} after loading data from memory.}
  \label{fig:ccuda-simple-mma-mem}
  \vspace{-1em}
\end{wrapfigure}

\paragraph{Thread-Shared Data}

Shared and global memory, on the other hand, is not replicated per-thread;
instead, all threads in a block have the same view into shared memory, while
all threads on the grid have the same view into global memory.

CUDA does not explicitly model shared and global memory spaces, nor does it
track whether allocations in a given memory space exceed device limits. \name{},
on the other hand, distinguishes these spaces and restricts shared memory to
static allocations, throwing an error if an allocation exceeds device limits. We
will discuss these aspects in detail in ~\Cref{sec:data-unit}. For now, however, we
focus on how views of memory converge and diverge in tandem with the compute hierarchy.

As an example, when launching a kernel on the GPU, a global memory pointer
initially belongs to the grid because every thread sees the same pointer value at
the start of the computation. To write to that memory, each thread computes an
offset from the original base address. In doing so, the pointer effectively
diverges across threads so each
can write independently. After these writes, the view on the memory must
reconverge, ensuring that all threads have completed their updates
before it can return to its original logical owner.

Let us reconsider the Tensor Core example from ~\Cref{fig:cuda-simple-mma}, this
time initializing the operands of the Tensor Core operation from pointers into
global memory. In ~\Cref{fig:cuda-simple-mma-mem}, \code{A} and \code{B} are now
populated from \code{a\_mem} and \code{b\_mem}. After the
Tensor Core operation completes, the result is written back to \code{c\_mem},
also in global memory. This variable \code{c\_mem} is logically accessed from two different levels
of the compute hierarchy. At the beginning, each thread in a group of 32
sees the same value for \code{c\_mem}. Then, they
each locate an offset within \code{c\_mem} that they write to (lines
\refline{20-23}). To restore \code{c\_mem} back to its 32 thread-level
ownership, all 32 threads must synchronize (line \refline{25}) to ensure all per-thread writes
have completed.  This is similar to the requirement we saw documented in
CUB in ~\Cref{sec:intro}.

Similarly to the compute hierarchy, CUDA
programmers are responsible for tracking the logical owner of a view of memory
as it evolves over the course of a program, and for ensuring that appropriate
synchronization occurs.

\name{}, by contrast, makes the evolution of memory's view
explicit in the syntax and automatically inserts the synchronization
required to restore that view to its original owner. In
~\Cref{fig:ccuda-simple-mma-mem}, we show how memory is lowered through the
compute hierarchy for the same Tensor Core operation from
~\Cref{fig:cuda-simple-mma-mem}. The lowering of \code{c\_mem} is required in
this program as \name{} only permits writes from the view of a single
thread. To lower \code{c\_mem}, we use \name{}'s \code{partition} operation on
line \refline{28}. The operation takes a pointer to partition,
\code{c\_mem}, and lowers it to a single thread, assigning it a new name,
\code{c\_thrd}; \name{} will not allow use of the old variable \code{c\_mem}
within the \code{partition}'s scope.  The \code{partition} also takes an
indexing function, and each time \code{c\_thrd} is accessed, this
function is implicitly applied. Once the \code{partition}’s scope ends, \name{}
inserts a synchronization barrier before the next use of the original variable, so that
all per-thread writes have completed. In this way, at the end of the
\code{partition}, the original variable represents a convergent
view of the data once again. 

%% file: figs/2-panel.tex
\begin{figure}
\centering
\begin{minipage}[t]{0.29\linewidth}
\centering
    \begin{cudaBox}
    int tid = threadIdx.x; 
    if (tid >= 0 && tid < 32){
      float A[4], B[2];
      float C[4] = { 0 };
      // Populate A with unique values.
      for (int i = 0; i < 4; i++) 
          A[i] = tid * 4 + 1;
      // Populate B with unique values.
      for (int i = 0; i < 2; i++) 
          B[i] = tid * 4 + 1;
      // Issue a warp-level Tensor Core
      // operation: D = A * B + C
      // (eliding some typecasts).
      asm("mma.sync.aligned.m16n8k8..." 
        "{
        "{
        "{
        "{
        : "=r"(C[0]),"=r"(C[1]), ... 
        :  "r"(A[0]), "r"(A[1]), ... 
           "r"(B[0]), "r"(B[1]),
           "r"(C[0]), "r"(C[1]), ...);}
    \end{cudaBox}
    \vspace{-2em}
    \caption{Warp-level Tensor Core instruction in CUDA.}
    \label{fig:cuda-simple-mma}
\end{minipage}
\hfill
\begin{minipage}[t]{0.28\linewidth}
\centering
\begin{ccudaBox}
tid : int @ thread[1] = id(); 
with group(thread[32]):
  A : float[4] @ thread[1]
  B : float[2] @ thread[1]
  C : float[4] @ thread[1]
    
  # Populate A with unique values
  for i in range(0, 4, 1): 
    A[i] = tid * 4 + 1
    C[i] = 0
  # Populate B with unique values
  for i in range(0, 2, 1): 
    B[i] = tid * 4 + 1

  # Issue Tensor Core op.
   intrinsic.mma(
     A[0], A[1], A[2], A[3],  
     B[0], B[1],
     C[0], C[1], C[2], C[3],
     out=[A[0], A[1], A[2], A[3]])

\end{ccudaBox}
\vspace{-2em}
\caption{Warp-level Tensor Core instruction in \name{}.}
\label{fig:ccuda-simple-mma}
\end{minipage}
\hfill
\begin{minipage}[t]{0.38\linewidth}
\centering
\vspace{-17em}
\includegraphics[width=0.8\linewidth]{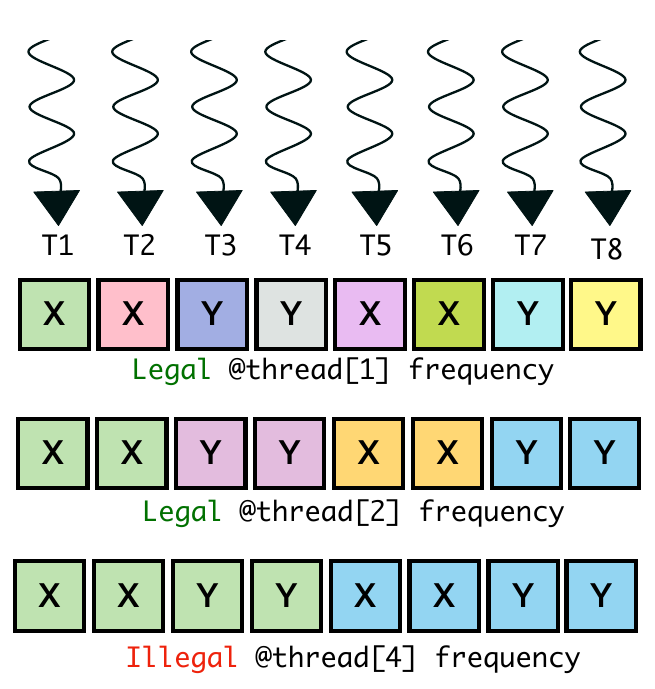}
\vspace{-1.35em}
\caption{The \code{@} syntax represents the frequency at which a variable varies
across threads \code{T1-T8}. Each \code{@} denotes a coloring of the variable;
threads that ``view'' the variable with the same color must observe the same
value.}
\label{fig:freq-color}
\end{minipage}
\vspace*{-1.5em}
\end{figure}

%% file: 3-langauge.tex
\begin{figure}
\begin{minipage}[t]{\linewidth}
  \begin{minipage}{0.49\linewidth}
    \input{figs/3-mm-1}
  \end{minipage}
  \hfill
  \begin{minipage}{0.47\linewidth}
     \input{figs/3-mm-2}
  \end{minipage}
  \vspace{-1em}
  \caption{Naive tensor float 32 matrix multiplication in \name{} (full program can be found in Appendix ~\ref{sec:tf32-func}).}
  \vspace{-1em}
  \label{fig:global-mm}
\end{minipage}
\vspace{-0.5em}
\end{figure}

\enlargethispage{\baselineskip}

\section{The \name{} Language}\label{sec:language}

\name{} is an imperative, low-level language designed at a level of
abstraction comparable to that of CUDA. Unlike CUDA, however,
\name{}'s syntax materializes the mapping of computations onto the compute and
memory hierarchy explicitly in the program source. Using this information,
\name{} enforces that programs only execute collective operations with
sufficient resources.



To guide our discussion about \name{}'s design, we use the program in
~\Cref{fig:global-mm} as a running example. It computes a matrix multiplication
between two float arrays, \code{A} and \code{B}, to produce an output matrix
\code{C}. In this program, each block computes an independent $16 \times 8 $
tile of the output. To do so, blocks first locate the tile index assigned to them
(line \refline{19-20}). Next, threads in each block load corresponding rows and
columns from \code{A} and \code{B} (line \refline{41-53}) into shared memory.
Finally, the program invokes a warp-level, Tensor Core instruction to compute
the output (line \refline{62}), requiring threads in a warp to
converge. We encapsulate this Tensor Core instruction in a
function, demonstrating how function composition works in \name{}. This function
is the same as the one in ~\Cref{fig:ccuda-simple-mma-mem}.

\subsection{Levels}

\name{} models the machine's compute hierarchy through \textit{levels}. There are three levels in \name{}---\code{grid}, \code{block} and \code{thread}---which are organized as
expected: a \code{grid} consists of multiple \code{block}s,
each of which consist of multiple \code{thread}s. Levels are ordered, with \code{thread} $<$ \code{block} $<$ \code{grid}.

There are two key differences between \name{}'s levels and those in CUDA. First, 
\name{} does not model CUDA's three-dimensional grid or block
structure. Second, there are two commonly-used ``levels'', warp and warpgroup, absent from our
hierarchy. On the hardware, the units of each level are arranged in a linear order, and the three-dimensional structures of CUDA are simply interpretations of this
ordering, not distinct hardware resources. Similarly, warps and warpgroups are
organizational constructs defined in terms of existing levels. Namely, a warp is a
group of 32 threads whose first thread ID is aligned to 32. A warpgroup, which
was introduced with the release of the Hopper architecture~\cite{hopper},
consists of 4 consecutive warps.


Rather than baking these interpretations into \name{} by adding new dimensions
and levels, we let users express multi-dimensional structures and define
groupings of custom sizes.

\vspace*{-0.5em}
\subsection{\Unit{}s}\label{sec:comp-scope}

\Unit{}s are the central concept of \name{}, representing the view of the
hierarchy from which a given statement is defining the machine's behavior.
Perspectives allow \name{} to determine which compute resources the programmer
is controlling at every point in the program, whether they are available in the program's context, and if
those resources are sufficient for a given operation.

A \unit{} is a level---\code{grid}, \code{block}, or
\code{thread}---paired with a static constant \code{n}, specifying the
number of units at that level. For example, \code{thread[2]} denotes a \unit{} of
two threads, \code{block[4]} denotes a \unit{} of four blocks, and so on.
\Unit{}s also carry \textit{alignment information}: a \unit{} of size \code{n} is aligned to
\code{n}. In this way, a warp is simply a
desugaring of \code{thread[32]}, and a warpgroup is a desugaring of
\code{thread[128]}.

Finally, \unit{}s are partially ordered. We say that
\code{level}\textsubscript{2}[\code{n}\textsubscript{2}] is \textit{broader} than \code{level}\textsubscript{1}[\code{n}\textsubscript{1}], or that 
\code{level}\textsubscript{1}[\code{n}\textsubscript{1}] is \textit{narrower} than  \code{level}\textsubscript{2}[\code{n}\textsubscript{2}],   
if either:

\begin{enumerate}
    \item \code{level}\textsubscript{1} $<$ \code{level}\textsubscript{2}; or,
    \item \code{level}\textsubscript{1} $=$ \code{level}\textsubscript{2} and \code{n}\textsubscript{1} divides \code{n}\textsubscript{2}.
\end{enumerate}

\vspace*{-0.5em}
\subsection{\Unit{}s on Code} 
\label{sec:resource-sets-for-code}

Code is associated with a set of \unit{}s, called a \textit{\unit{} bound},
which corresponds to the compute resources whose behavior it defines. At every point in the
program, the \unit{} bound for that point indicates which layer of the hierarchy
is being programmed, and how that layer can be destructed into narrower
\unit{}s.  For example, a line of code with a \code{\{block[1], thread[4]\}}
\unit{} bound tells \name{} that the current line of code is being programmed at
the  \code{block[1]} \unit{} and that the \code{block[1]} \unit{} can be
destructed into a some number of \code{thread[4]} \unit{}s. As a
shorthand, when we refer to a code's \unit{}, we mean the broadest \unit{}
available in its \unit{} bound (in this example, \code{block[1]}).

Functions begin with a top-level \unit{} bound.
In ~\Cref{fig:global-mm}, the \unit{} bound is defined on line \refline{2},
using the notation
\code{@requires(grid[n\textsubscript{1}],block[n\textsubscript{2}],thread[n\textsubscript{3}])}.
Programmers can then shape the program's
current \unit{} using two constructs: \code{group} and \code{split}.

\paragraph{\textbf{Group}}
The \code{group} construct lets programmers shift from a broader \unit{} to some
number of narrower \unit{}s contained in it. Operationally, the \code{group}
construct does this by replicating code written from the
narrower \unit{} across the broader one. In effect, \code{group} forks
many copies of a narrower \unit{}.

\begin{wrapfigure}{r}{0.30\textwidth}
  \vspace{-3em}
  \begin{ccudaBox}
# Example 1
with group(thread[2]):
    # Illegal because block > thread.
    with group(block[1]):
        pass
# Example 2
with group(block[6]):
    # Illegal because 6 
    with group([block[5]]):
        pass
  \end{ccudaBox}
   \vspace{-2em}
  \caption{Illegal uses of \code{group}.}
  \label{fig:illegal-groups}
 
   \vspace{-3em}
\end{wrapfigure}

Let's consider this in context of our example. In ~\Cref{fig:global-mm},
execution begins at \code{grid[1]} on line \refline{12}. At that point, a
programmer controls the whole grid's behavior. To produce different output
tiles, the programmer shifts their \unit{} to \code{block[1]} on line
\refline{25}. The code within the \code{group(block[1])} defines the
behavior of a single block, and is replicated across all blocks in the grid.

Not all uses of \code{group} are valid:
the examples in ~\Cref{fig:illegal-groups} have no meaning on
the hardware. On line \refline{4}, the program tries to broaden its \unit{} to
\code{block[1]} from \code{thread[2]}, which is illegal. On
the other hand, in the second example, the program tries to narrow its \unit{}
from \code{block[6]} to \code{block[5]}. While 5 is indeed less than 6, \name{}
cannot evenly replicate \code{block[5]} across a \code{block[6]}
\unit{}, and so rejects this program. Recall, from ~\Cref{sec:comp-scope},
that whether one \unit{} is narrower than another is dependent on divisibility, not just size.


To eliminate such cases, \name{} only allows an invocation of
\code{group(level[n])} if the current \unit{} bound contains a \unit{} broader
than \code{level[n]}.  Once \code{group(level[n])} is invoked, it modifies the
current \unit{} bound in two ways. First, it removes all \unit{}s broader
than \code{level[n]} from it.  Second, it sets the broadest \unit{} within
the \code{group} to be \code{level[n]}.



\begin{wrapfigure}{r}{0.25\textwidth}
  \vspace{-2em}
  \begin{ccudaBox}
with group(thread[4]):
  match split(thread):
    case 2:
      ...
    case 1:
      ...
    case 1:
      ...
 \end{ccudaBox}
   \vspace{-2em}
  \caption{Example \code{split}.}
  \label{fig:small-split-example}
  \vspace{-0.5em}
\begin{ccudaBox}
with group(thread[4]):
 match split(thread):
    case 4:
      ...
    case 1:
      ...
\end{ccudaBox}
\vspace{-2em}
\caption{Illegal split exceeds the available \unit{}s.}
\label{fig:illegal-sum-split}
\vspace{-0.5em}
\begin{ccudaBox}
with group(thread[3]):
  match split(thread):
    case 1:
      ...
    case 2:
      ...
\end{ccudaBox}
\vspace{-2em}
\caption{Illegal split violates alignment.}
\label{fig:illegal-align-split}
\vspace{-1em}
\end{wrapfigure}

\vspace*{-0.25em}

\paragraph{\textbf{Split}} Unlike \code{group}, which is used for replication
into equally-sized, narrower \unit{}s, \code{split} is used for sharding
the current \unit{} unequally. For example, 
~\Cref{fig:small-split-example} shows a \code{split} from \code{thread[4]} into
one branch with \code{thread[2]} \unit{} and two with \code{thread[1]}. The three
arms of the \code{split} execute independently in parallel,\footnote{Despite the use of the
\code{match} syntax, \emph{all} branches of the \code{split} execute.} as a form
of unordered composition.
When \code{split(level)} is invoked, the \unit{}s of
its branches diverge. At the end of the \code{split}, they reconverge and continue execution
with the original \unit{}.

Use of \code{split} is necessary to write warp-specialized~\cite{singe} code, a
programming pattern used in high-performance kernels.  Another important use of
\code{split}---masking off threads---can be found in our running example.  Line
\refline{58} in ~\Cref{fig:global-mm} shows a \code{split} that requests the
first warp in the block, narrowing from \code{block[1]} into a single
\code{thread[32]}. This warp is later used
to execute a Tensor Core operation.

Because \code{split} corresponds to unordered composition,
it must provide each of its branches their requested \unit{}s
simultaneously.
\name{} thus checks that the sum
of the \unit{}s requested by all branches of the \code{split} can be satisfied.
For example, the program in ~\Cref{fig:illegal-sum-split} does not type check.
Finally, because \unit{}s enforce alignment, every branch of the \code{split}
must also be aligned; not all \code{split}s whose sizes are at most 
the available units are valid. ~\Cref{fig:illegal-align-split} shows an
example violating this constraint: the second branch
of the split is not aligned to 2.

Once \code{split} is invoked, for each branch that requests \code{n} units,
all \unit{}s broader than \code{level[n]} are removed from its \unit{}
bound, and the available units for \code{level} are set to \code{n}.
\vspace*{-0.25em}

\subsection{\Unit{}s on Data}
\label{sec:data-unit}

\Cref{sec:comp-scope} described how programmers can control different layers of
the hierarchy by changing their \unit{}s on code through \code{group} and
\code{split}. To ensure these operations remain meaningful, \name{}
must ensure that threads inside a \unit{} remain logically grouped,
even when they encounter \code{for},
\code{while}, and \code{if} statements.
As we saw in \Cref{sec:mem-hierarchy}, making this guarantee
requires \name{} to track how data varies across threads.



\subsubsection{Thread-Local Data}
In \name{}, each local variable has a \unit{},
which indicates the frequency at which its values change
in space. This frequency remains constant for the duration of a program, and it tells
\name{} that a \code{level[n]} variable is always indistinguishable to threads
within that \unit{}. For example, \code{blk\_row @ block[1]} on line \refline{19} in \Cref{fig:global-mm}
has the same value across all threads in a block.

Programmers specify the \unit{} that each variable lives at in its declaration.
A  variable \code{v} of type \code{int} is declared at \code{thread[1]} \unit{}
using the syntax \code{v : int @ thread[1]}.\footnote{If not explicitly
annotated, \name{} infers a variable's \unit{} to be the \unit{} of the code
where it was declared.} To enforce this frequency invariant, \name{}
restricts reads from and writes to variables based on their \unit{}s.
The rule can be summarized as follows: \name{} allows programs to ``read up''
from broader \unit{}s and ``write down'' to narrower ones.

\begin{wrapfigure}{r}{0.30\textwidth}
  \vspace{-2em}
\begin{ccudaBox}
flag : bool @ thread[1] = ...
with group(block[1]):
  if (flag)
    __syncthreads();
\end{ccudaBox}
\vspace{-2em}
\caption{Illegal read of \code{thread[1]} variable.}
\label{fig:illegal-read}
\begin{ccudaBox}
x : bool @ thread[1] = ...
y : bool @ block[1] = ...
with group(block[1]):
  y = x
  if(y):
    __syncthreads();
\end{ccudaBox}
\vspace{-2em}
\caption{Illegal write of \code{thread[1]} variable into a \code{block[1]} variable.}
\label{fig:illegal-write}
\vspace{-2em}
\end{wrapfigure}

\paragraph{\textbf{Read Up}} 
Variables can only be read if their \unit{} is at least as broad as the current
code \unit{}.  ~\Cref{fig:illegal-read} gives an example of an illegal read that
would violate this constraint.
While \code{\_\_syncthreads()} should always be safe in
\code{block[1]}, branching on
the variable \code{flag}---which may take different values across threads in a block---can cause only some of
those threads to reach the \code{\_\_syncthreads()}, violating its collective invariant.

\paragraph{\textbf{Write Down}}
Writes are dually constrained. Only values that live at broader \unit{}s can be
written into variables that live at narrower ones.
For example, a \code{block[1]}
variable can be used to write into a \code{thread[1]} variable, but
not vice versa.
An example of an illegal write is shown in Figure~\ref{fig:illegal-write},
where writing from a \code{thread[1]} variable into a \code{block[1]} one
would cause a deadlock.


Together, the ``read up'' and ``write down'' rules ensure that information only
flows from broader \unit{}s to narrower ones.  In ~\Cref{fig:global-mm},
we can see the ``read up'' rule in action on lines \refline{42} and \refline{43}. Meanwhile,
lines \refline{19} and \refline{20} are instances of the ``write down'' rule.

\subsubsection{Thread-Shared Data}
Unlike thread-local data, which is literally replicated across threads and is
backed by distinct physical storage, thread-shared data consists of pointers
into shared and global memory that are visible to some collection of threads. As
a result, the \unit{} such data inhabits can evolve as the program executes.

In \name{}, there are three mechanisms for obtaining thread-shared data. The
first is to directly allocate data residing in shared memory.\footnote{\name{}
assumes that all global memory allocations have been made before launching the
CUDA kernel, as is typically the case with CUDA programs.} The second and third
are to obtain offsets into existing data by using the \code{\partition{}} or
\code{\ownpartition{}} operations. These constructs mirror \code{group} and
\code{split}, and are used to temporarily narrow the \unit{} associated with a
pointer.

\name{} simplifies shared memory allocations by requiring them to have a static size.
The \code{@requires}
annotation specifies the amount of shared memory a function expects to have.
\name{} uses this information to ensure that a function's
allocations do not exceed its declared limit, and checks whether there is enough
shared memory available at its call sites. The function in ~\Cref{fig:global-mm} declares the
amount of shared memory it will require on \refline{2}. We use standard
techniques~\cite{hoffmann2011types} to statically bound memory usage with \name's type system.

Additionally, since shared memory is only visible to
threads within the same block, allocations to it are only permitted at
\code{block[1]}. An example of such an allocation is shown on lines
\refline{28-30} of ~\Cref{fig:global-mm}.
During compilation, \name{} will
automatically handle the necessary pointer arithmetic to assign each allocation
an appropriate offset within the shared memory space.


\paragraph{\textbf{Partition}} The \code{\partition} operation plays the same
role for memory that \code{group} plays for code. It is used to refine the
\unit{} of a pointer by computing offsets for each narrowed \unit{}.
Concretely, the \code{partition} operation takes a pointer variable \code{x}, an
indexing function \code{f}, and produces a new variable \code{y} at
a narrower \unit{} \code{level[n]} to be used inside the \code{partition}.
After the partition ends, the threads participating in the \code{\partition}
must synchronize to restore the original pointer \code{x} to its \unit{}.\footnote{
Strictly speaking, this synchronization is only required
at the next point the memory is to be \emph{re-used}. \name{} has a mechanism
for optimizing the placement of these synchronization points, which is describe in
~\Cref{sec:compilation}.
}


Inside the scope of the partition, the original variable \code{x} cannot be used,
and the pointer can only be accessed using \code{y}. Each use of \code{y}
applies the indexing function \code{f} to compute the true offset of the access.
For example, the use of \code{\partition} on line \refline{34}
in~\Cref{fig:global-mm} takes a pointer \code{C\_smem} at \code{block[1]}, gives
it a new name \code{C\_th}, and distributes it across \code{thread[1]}
perspectives by transforming each occurrence of \code{C\_th[i]} in the body into
\code{C\_smem[i*4]}.


At this point, it is necessary to consider how different indexing functions
affect the possibility of data races. If the indexing function is injective,
each narrower \unit{} receives a distinct offset into the underlying
array, and the resulting partition is free of data races. When the indexing
function is not injective, multiple threads may race on the same location,
introducing a potential data race. 

Preventing data races is not one of \name{}'s
goals. In \name{}, data races \emph{can} occur within a \code{partition};
however, since the data is eventually synchronized before it is reused, the
last writer wins. 
Out-of-bounds accesses
are considered undefined behavior. 
Prior work, in particular Descend
\cite{kopcke-descend-2024}, describes a type system for restricting data-races
and out-of-bounds accesses, and we believe a similar approach can be combined with \name{}'s.
\name{} instead focuses on the interaction between the compute and
memory hierarchies and on reasoning about them simultaneously to ensure that
operations are executed only with sufficient compute resources, a
guarantee that Descend cannot provide.




\paragraph{\textbf{Claim}} The \code{claim} operation lets programmers
narrow a pointer's \unit{} by giving it to only
one collection of threads with that narrower \unit{}.
For example, line \refline{57} of ~\Cref{fig:global-mm}
shows an example of a \code{claim}, where
the pointer \code{C\_smem @ block[1]} is only available to a single warp
in the block and narrowed to \code{Cs\_warp @ thread[32]} to call a Tensor Core operation.

A \code{claim} takes the original pointer
\code{x}, the target \unit{} to narrow it to, and a new name \code{y} to assign to
the narrowed memory. In all \code{split}s within the \code{claim}, the new
pointer \code{y} is accessible only within one branch; sibling
branches are not permitted to read or write from this memory. 
As with \code{partition}, the original variable \code{x} is
not accessible within the \code{claim}. 

\subsection{The \code{id()} Function}

As opposed to exposing users to special hardware variables like
\code{blockIdx.x}, and \code{threadIdx.x}, \name{} provides an \code{id()}
function instead.  The \code{id} function returns the relative index of a narrow
\unit{} within a broader \unit{}. That is, the interpretation of the \code{id}
function depends on both the \unit{} of the variable it is being written to and
the \unit{} of the code invoking it. In~\Cref{fig:global-mm}, we use a call to
\code{id} with \code{grid[1]} \unit{} on lines \refline{19} and \refline{20} to
locate the tile that each block is in charge of computing. On line \refline{32},
however, a write of \code{id} into a \code{thread[1]} variable at
\code{block[1]} \unit{} will return the relative ID of the thread within the
block, not in the grid.

\subsection{Collective Operations}\label{sec:collective}
There are two types of collective operations in \name{}.

\subsubsection*{Function Calls}

Each function carries a \emph{\unit{} signature}, which consists of its
top-level \unit{} bound, shared memory usage, and its arguments' \unit{}s.
An example of such a signature can be seen in lines \refline{3-9} in \Cref{fig:global-mm}.
At call sites, \name{} ensures that the callee's \unit{} signature can be satisfied
by the caller.  Since the \unit{} bound describes a minimum requirement,
functions can be called with a broader code \unit{} than necessary, accounting for alignment. 
When checking function arguments at call sites, \name{} distinguishes between primitive
data types and pointers. For primitive data types, like \code{int} or
\code{bool}, arguments can have broader \unit{}s than specified in the function
signature.  On the other hand, pointers that live at broader \unit{}s can only be passed
if the pointer is marked as \code{const}, indicating that it will only be used
for reading. Otherwise, we require that the pointer's \unit{} exactly match
that of the function's \unit{} signature.

\subsubsection*{Intrinsics}
The compiler provides a pre-defined set of collective intrinsics---like the Tensor Core
instruction discussed in~\Cref{sec:compute-hierarchy}---each of which declares
its \unit{} signature. Programmers may also add to this set using \code{unsafe}. From within
\code{unsafe} code, programmers can inline assembly instructions and wrap them
in a \name{} function, specifying its \unit{} signature.  \name{} checks these
call sites like those of any other function.

\vspace*{-0.5em}

\subsection{Asynchrony}

Asynchronous data movement works similarly to other memory operations.  Users
can mark storage as asynchronous with the \code{async} construct. As with
\code{partition} and \code{claim}, \code{async} hides the old variable and
exposes a fresh one that is only accessible within the \code{async} statement.
Inside, \name{} ensures the new variable is only used by async data-movement
intrinsics.

\begin{wrapfigure}{r}{0.35\textwidth}
  \vspace{-2em}
  \begin{ccudaBox}
  with async(old_name) as new_name:
      match split(thread):
        case 1:
          cp_async(smem, new_name, 16);
  \end{ccudaBox}
  \vspace{-2em}
  \caption{Example of an \code{cp.async} instruction in \name{}.}
  \label{fig:async-movement}
  \vspace{-1em}
  \end{wrapfigure}

\name{} includes two such intrinsics: bulk \cite{bulk-cite} and non-bulk \cite{non-bulk-cite} asynchronous data-movement instructions. We show an example of the former in ~\Cref{fig:async-movement}. As with
the data operations in \Cref{sec:data-unit}, \name{} inserts the necessary
synchronization before the program's next use of the original variable, ensuring
that the asynchronous transfer has completed.

%% file: figs/3-mm-1.tex
\begin{ccudaCode}
@prism("global")
# Top-level perspective bounds and shared memory usage.
@requires(grid[1], block[1], thread[32], smem=1280)
def mmaTF32NaiveKernel(A: ptr(const(float)) @ grid[1],
                       B: ptr(const(float)) @ grid[1],
                       C: ptr(float) @ grid[1],
                       M : int @ grid[1],
                       N : int @ grid[1],
                       K : int @ grid[1]):
             
  # Starts out with grid[1] perspective.
  with group(grid[1]):
    # @ grid[1] inferred from current perspective 
    # Each block computes an 16 x 8 tile
    num_blocks_n : const(int) = (N + 8 - 1) / 8 

    # id() function returns the block id
    # inferred from @ block[1].
    blk_row : const(int) @ block[1] = (id()/num_blocks_n)*16
    blk_col : const(int) @ block[1] = (id()
    
    # Give each block an offset into C
    offset = lambda x: blk_row * N + blk_col + x
    with partition(C, p=block[1], f=offset) as C_blk:
      with group(block[1]):
        # SHMEM declarations are only allowed 
        # with a block[1] perspective.
        A_smem : shared(float[16 * 8]) @ block[1]
        B_smem : shared(float[8 * 8])  @ block[1]
        C_smem : shared(float[16 * 8]) @ block[1]
        # Now, id() returns the thread id
        idx = lambda x: x * 4
        # To write to C_smem, drop to thread[1] perspective 
        with partition(C_smem,p=thread[1],f=idx) as C_th:
           for i in range(0, 4, 1):
              with group(thread[1]):
                C_th[i] = 0
\end{ccudaCode}

%% file: figs/3-mm-2.tex
\begin{minted}[
  linenos,
  autogobble,
  breaklines,
  baselinestretch=1.0,
  fontsize=\tiny,
  numbersep=3pt,
  frame=single,
  firstnumber=38,
]{python}
        for i in range(0, K_tiles, 1):
          
          # --- Sync point --- (backedge from for loop)
          for j in range(0, 4, 1):
            global_row : int @ thread[1] = blk_row + row 
            global_col: int @ thread[1] = i * 8 + col 
            with partition(A_smem, p=thread[1], f=..) as A_th:
              with group(thread[1]):
                A_thrd[0] = A[global_row * K + global_col]
            
          # --- Sync point --- (backedge from for loop)
          for j in range(0, 2, 1):
            # Similar to write into A_smem ...
            with partition(B_smem,p=thread[1],f=...) as B_th:
              with group(thread[1]):
                B_th[0] = B[global_row_b * N + global_col_b]
        
            # Give each warp an offset into C_smem.
            # --- Sync point --- (backedge from for loop)
            with claim(C_smem, p=thread[32]) as Cs_warp:
                match split(thread): # Masks off other threads
                  case 32:
                    # Call function that performs a
                    # Tensor Core instruction.
                    simple_mma(A_smem, B_smem, Cs_warp)
            # --- Sync point ---

        for j in range(0, 4, 1):
          flat_idx_c : int @ thread[1] = id() * 4 + j
          row_c : int @ thread[1] = flat_idx_c / MMA_K 
          col_c : int @ thread[1] = flat_idx_c % MMA_K 
          idx = lambda x: row_c * N + col_c + x
          with partition(C_blk, p=thread[1], f=idx) as C_th:
            with group(thread[1]):
              C_th[0]  = C_smem[row_c * MMA_N + col_c]

  return
\end{minted}

%% file: 4-core-calculus.tex
\newcommand{\defas}{\mathrel{::=}}
\newcommand{\alt}{\:|\:}
\newcommand{\size}[1]{\text{size}(#1)}

\newcommand{\hier}{\textit{h}}
\newcommand{\capa}{\pi}
\newcommand{\mem}{m}
\newcommand{\thread}{\code{Thread}\xspace}
\newcommand{\block}{\code{Block}\xspace}
\newcommand{\grid}{\code{Grid}\xspace}
\newcommand{\local}{\code{Local}\xspace}
\newcommand{\shared}{\code{Shared}\xspace}
\newcommand{\globalkind}{\code{Global}\xspace}
\newcommand{\threads}{T\xspace}
\newcommand{\blocks}{B\xspace}
\newcommand{\fun}[3]{\textbf{Fun}(#1, #2, #3)}
\newcommand{\splitcom}{\code{split}\xspace}
\newcommand{\destructcom}{\code{destruct}\xspace}
\newcommand{\replicom}{\code{group}\xspace}
\newcommand{\down}[1]{\code{lower}_{#1}\xspace}
\newcommand{\partitioncmd}[5]{\code{partition}_{#1}\code{ }#2\code{ into }#3\code{ by }#4\code{ in }#5}
\newcommand{\asyncp}[4]{\code{async\_partition}_{#1}\code{ }#2\code{ into }#3\code{ in }#4}
\newcommand{\asyncmemcpy}[2]{\code{async\_mempcy}(#1, #2)}
\newcommand{\memcpy}[2]{\code{memcpy}(#1, #2)}
\newcommand{\claim}[5]{\code{claim}_{#1}\code{ }#2\code{ into }#3\code{ at }#4\code{ in }#5}
\newcommand{\stepsto}{\leadsto}
\newcommand{\async}[1]{\code{async }#1}
\newcommand{\skipcom}{\code{skip}\xspace}
\newcommand{\alignto}{\textbf{ align to }}

\newcommand{\update}{\textbf{update}\xspace}
\newcommand{\get}{\textbf{get}\xspace}
\newcommand{\rename}{\textbf{rename}\xspace}

\vspace*{-0.5em}

\section{Formalization in \calcname}\label{sec:core}

Having introduced the full \name language, we now describe \calcname, a core
calculus that formalizes its most fundamental aspects by statically tracking \unit{}s on code and data. 
We use \calcname to argue that well-typed \name programs are not only type-safe, but will also
never execute operations for which they lack the correct \unit{}.

In this section, we describe 
\calcname's type system and operational semantics---in
particular how it manages compute and data \unit{}s---and 
build up to a formal proof of type-and-\unit{} safety.  
We instrument \calcname's operational semantics with runtime \unit{} enforcement:
its rules will get stuck if they encounter an operation for which they have the wrong perspective.  
This runtime enforcement means that our safety theorem guarantees that dynamically-realized 
\unit{}s match the ones inferred by the type system. 

\vspace*{-0.5em}

\enlargethispage{\baselineskip}

\subsection{\calcname Type System}

The core idea in \calcname is to track, at the type level, the program's \unit{} on code and data.
To achieve this, we borrow techniques from the literature on dependency tracking \cite{abadi-core-1999}.
In particular, the code \unit{}
is tracked on the typing judgment, 
which has the form $\Gamma \vdash^{\capa} e : \tau$ for expressions and 
$\Gamma \vdash^{\capa} s$ for statements.
The $\capa$ over the $\vdash$ is the code \unit{} on $e$ and $s$---and is comprised of a level and a size---the same structure as a \unit{} in \name{}.

The typing context also tracks the \unit{} at which each variable lives;
data can only be read from or written to a variable 
when its \unit{} is compatible with that of the code interacting with it.
This requirement is made manifest in the \rref{T-Var} rule,
found in ~\Cref{fig:calculus-types}.
Observe that the $\capa$ in the variable rule must match 
exactly between data and code; principles like ``read up'' and ``write down''
are instead encoded directly in the rules for reading and writing, like \rref{T-Arr-Access}, 
which views the array being read with broader \unit{} than the current code \unit{}.

~\Cref{fig:calculus-types} also shows other key rules, which fall into two main
categories: those for managing \unit{}s on code and those for
\unit{}s on data.

\begin{figure}
  \small

  \begin{center}
    \begin{tabularx}{1\textwidth}{
        >{\raggedright\arraybackslash}X
        >{\raggedleft\arraybackslash}X
      }
      \fbox{$\Gamma \vdash^{\capa} e : \tau$} & \textit{(Expression typing)}
    \end{tabularx}
  \end{center}
  $$
  \infer[\rlabel*{T-Var}]
  {\Gamma \vdash^{\capa} x : \tau}
  {x :^{\capa} \tau \in \Gamma}\qquad
  \infer[\rlabel*{T-Arr-Access}]
  {\Gamma \vdash^{\capa} e_1[e_2] : \tau}
  {\Gamma \vdash^{\capa'} e_1 : \tau[]^l & \Gamma \vdash^{\capa} e_2 : \code{int} & \capa \leq \capa'
  }
  $$
  \begin{center}
    \begin{tabularx}{1\textwidth}{
        >{\raggedright\arraybackslash}X
        >{\raggedleft\arraybackslash}X
      }
      \fbox{$\Gamma \vdash^{\capa} s$} & \textit{(Statement typing)}
    \end{tabularx}
  \end{center}
  $$
  \infer[\rlabel*{T-Split}]
  {\Gamma \vdash^{(\hier, n)} \splitcom (n_1, n_2) \{ s_1 \} \{ s_2 \}}
  {\Gamma \vdash^{(\hier, n_1)} s_1 & \Gamma \vdash^{(\hier, n_2)} s_2 &
  n_1, n_2 \alignto n}
  $$
  \text{where }
  $n_1, n_2 \alignto n \defas (n_1 + n_2 \leq n) \textit{ and }
  (n_1 \alt n) \textit{ and } (n_2 \alt n) \textit{ and } (n_2 \alt n_1 + n)$
  \vspace{1.5ex}
  $$
  \infer[\rlabel*{T-Group}]
  {\Gamma  \vdash^{q \cdot \capa} \replicom \code{ } q \code{ }s}
  {\Gamma  \vdash^{\capa} s}
  \qquad
  \infer[\rlabel*{T-Lower}]
  {\Gamma, x :^{\capa} \tau[]^l \vdash^{\capa} \down{} \code{ }x \code{ into  } y \code{ in } s}
  {\Gamma, y :^{\downarrow \capa} \tau[]^l \vdash^{\capa} s & l \neq \local}
  \vspace{2ex}
  $$
  $$
  \infer[\rlabel*{T-Partition}]
  {\Gamma, x :^{(\hier, n)} \tau[]^l \vdash^{(\hier, n)} \partitioncmd{}{x}{y}{c}{s}}
  {\Gamma, y :^{(\hier, n/c)} \tau[]^l \vdash^{(\hier, n)} s & c \alt n & l \neq \local}
  \qquad
  \infer[\rlabel*{T-Destruct}]
  {\Gamma \vdash^{\capa} \destructcom \code{ in } s}
  {\Gamma \vdash^{\downarrow \capa} s} 
  \vspace{2ex}
  $$
  $$
  \infer[\rlabel*{T-Claim}]
  {\Gamma, x :^{(\hier, n)} \tau[]^l \vdash^{(\hier, n)} \claim{}{x}{y}{n'}{s}}
  {\Gamma, y :^{(\hier, n')} \tau[]^l \vdash^{(\hier, n')} s & n' \leq n & l \neq \local}
  $$

  \normalsize
  \caption{Core typing rules of \calcname. The typing rules presented here are a simplified selection 
  of the full rules, which can be found in Appendix \ref{appendix:proofs-types}.}
  \label{fig:calculus-types}
  \vspace*{-1em}
\end{figure}

\subsubsection{Managing \Unit{}s on Code}


\calcname's \replicom statement directly corresponds to \name's, 
and is checked by the \rref{T-Group} rule. Given some 
statement $s$ that checks with \unit{} $\capa$,
the statement $\replicom\;\;q\;\;s$ will
check with $q \cdot \capa$,
enforcing \name{}'s divisibility requirement.


The \splitcom statement, meanwhile, is checked by the \rref{T-Split} rule,
and functions like a binary version of the n-ary \code{split} construct
in \name. It enforces the same divisibility requirements to ensure that
the \unit{}s on code and data remain properly aligned, and then checks
the two sub-statements $s_1$ and $s_2$ with the divided, narrower \unit{}s.

\enlargethispage{\baselineskip}

In \calcname, to better model the details of how \unit{}s shift down the GPU
hierarchy, we introduce a third construct called \code{destruct} that makes
explicit exactly where such shifts occur.  In the corresponding rule,
\rref{T-Destruct}, the $\downarrow$ operation on $\capa$s ``destructs'' the
\unit{} into many narrower \unit{}s at a lower level. This operation is defined
as $ \downarrow(\grid, 1) = (\block, \blocks)$ and $\downarrow(\block, 1) =
(\thread, \threads)$, where $\blocks$ and $\threads$ are parameters to a
particular instantiation of \calcname to describe the number of blocks per grid
and threads per block.\footnote{\calcname is abstracted over these $\blocks$ and
$\threads$ values, so instead of tracking \unit{} bounds the way that \name
does, it only tracks the top-level \unit{} described in \Cref{sec:comp-scope}.}

\subsubsection{Managing \Unit{}s on Data}

The mechanism for managing data \unit{} mirrors that of code \unit{},
with each operation for data corresponding to an operation for code.

The \code{partition} operation is analogous to \code{group}ing a code \unit{}. 
The typing rule for this operation, \rref{T-Partition}, 
requires that the data \unit{} on $x$, the variable to be partitioned,
is the same as the current \unit{} on code. 
After partitioning, a fresh variable $y$ is introduced with a new \unit{} $\capa / c$.
Within the \code{partition}, 
we disallow references to $x$
and continue checking
the body with the original $\capa$; the \code{partition} has no effect 
on the code \unit{}.

Unlike \code{partition}, which divides up a piece of data equally among narrower
\unit{}s, the \code{claim} operation views the claimed data with
exactly one narrower \unit{}. Accordingly, \calcname needs to ensure that only
one branch of a \code{split} operation, with the appropriate $\capa$, can refer
to the claimed variable. To ensure that this is the case, the \rref{T-Claim}
rule links the data \unit{} of the variable to the compute \unit{} of the
code claiming it by changing both at the same time. This represents a
minor difference from \name, which uses additional static analysis to
ensure that a claimed variable is only accessed in a single \code{split} branch.

Lastly, the \rref{T-Lower} rule mirrors the \rref{T-Destruct} rule;
it uses the $\downarrow$ operator to move a variable from one level 
of the hierarchy to another, distributing it equally among all the narrower 
\unit{}s at that level in the same manner as \rref{T-Partition}.

Note that these rules only apply to thread-shared memory, i.e., 
arrays that do not live in \local. The provenance of an array type 
is denoted by the superscript $l$ above it.

\subsection{\calcname Semantics}

\begin{figure}
  \small

  \begin{center}
    \begin{tabularx}{1\textwidth}{
        >{\raggedright\arraybackslash}X
        >{\raggedleft\arraybackslash}X
      }
      \fbox{$L, S, \Sigma, P \stepsto L', S', \Sigma', P'$} & \textit{(Machine judgment)}
    \end{tabularx}
  \end{center}
  $$
  \infer[\rlabel*{S-Program}]
  {L, S, \Sigma, P \stepsto L[t \mapsto \eta'], S[b \mapsto \sigma'], \Sigma', P[(t, b) \mapsto s']}
  {L(t), S(b), \Sigma, t, b, 0 \vdash^{(\grid, 1)} s \stepsto s' \dashv \eta', \sigma', \Sigma' & P(t, b) = s}
  \vspace*{2.75ex}
  $$
  \begin{center}
    \begin{tabularx}{1\textwidth}{
        >{\raggedright\arraybackslash}X
        >{\raggedleft\arraybackslash}X
      }
      \fbox{$\eta, \sigma, \Sigma, t, b, p \vdash^\capa s \stepsto s' \dashv \eta', \sigma', \Sigma'$} & \textit{(Thread judgment)}
    \end{tabularx}
  \end{center}

  $$
  \infer[\rlabel*{S-Split-Left}]
  {\eta, \sigma, \Sigma, t, b, p \vdash^{(\hier, n)} \splitcom(n_1, n_2) \{s_1\} \{s_2\} \stepsto \splitcom(n_1, n_2) \{s_1'\} \{s_2\} \dashv \eta', \sigma', \Sigma'}
  {p < n_1 & n_1, n_2 \alignto n & \eta, \sigma, \Sigma, t, b, p, \vdash^{(\hier, n_1)} s_1 \stepsto s_1' \dashv \eta', \sigma', \Sigma'}
  \vspace{2.75ex}
  $$
  $$
  \infer[\rlabel*{S-Split-Right}]
  {\eta, \sigma, \Sigma, t, b, p \vdash^{(\hier, n)} \splitcom(n_1, n_2) \{s_1\} \{s_2\} \stepsto \splitcom(n_1, n_2) \{s_1\} \{s_2'\} \dashv \eta', \sigma', \Sigma'}
  {p \geq n_1 & p < n_1 + n_2 & n_1, n_2 \alignto n & \eta, \sigma, \Sigma, t, b, p - n_1, \vdash^{(\hier, n_2)} s_2 \stepsto s_2' \dashv \eta', \sigma', \Sigma'}
  \vspace{2.75ex}
  $$
  $$
  \infer[\rlabel*{S-Group}]
  {\eta, \sigma, \Sigma, t, b, p \vdash^{(\hier, q \cdot n)} \replicom \code{ } q \code{ } s \stepsto \replicom \code{ } q \code{ } s'; \dashv \eta', \sigma', \Sigma'}
  {\eta, \sigma, \Sigma, t, b,  p \textbf{ mod } n \vdash^{(\hier, n)} s \stepsto s' \dashv \eta', \sigma', \Sigma'}
  \vspace{2.75ex}
  $$
  $$
  \infer[\rlabel*{S-Destruct-Block}]
  {\eta, \sigma, \Sigma, t, b, 0 \vdash^{(\block, 1)} \destructcom \code{ in } s \stepsto \destructcom \code{ in } s'
  \dashv \eta', \sigma', \Sigma'}
  {\eta, \sigma, \Sigma, t, b, t \textbf{ mod } \threads \vdash^{(\thread, \threads)} s \stepsto s'\dashv \eta', \sigma', \Sigma'}
  \vspace{2.75ex}
  $$
  $$
  \infer[\rlabel*{S-Alloc-Shared}]
  {\eta, \sigma, \Sigma, t, b, p \vdash^{(\block, 1)} x := \code{alloc } \shared \code{ } \tau \code{ } n \code{ in } s \stepsto s 
  \dashv \eta, \sigma[x \mapsto^\capa \langle x, n \rangle], \Sigma}
  {}
  $$

  \normalsize
\vspace*{-0.5em}
  \caption{Core semantic rules of \calcname. As with the typing rules, we present only a simplified selection 
  of the full rules, which can be found in Appendix \ref{appendix:proofs-semantics}.}
  \label{fig:calculus-semantics}
  \vspace*{-1em}
\end{figure}

Having explained the key rules of the type system, we can 
move on to discuss \calcname's operational semantics. 
To reflect the fact that a GPU program executes in parallel
across numerous threads, we model the semantics of \calcname in the style of
\citet{turon-logical-2013}, using a two-level small step judgment. We present
the key rules of this semantics in ~\Cref{fig:calculus-semantics}.

\enlargethispage{\baselineskip}

The top level (i.e., machine-level) judgment has just one rule: \rref{S-Program}.
This rule acts as a ``frame'' for the lower level (i.e., thread-level) judgment,
and steps a collection of thread-ID-indexed local memories ($L$), 
a collection of block-ID-indexed shared memories ($S$), a global memory ($\Sigma$), 
and a thread pool ($P$) to an updated collection of memories and updated thread pool. 
The thread pool maps thread and block IDs to code, intuitively representing 
the program being executed by each thread at the current moment. 
The \rref{S-Program} rule non-deterministically chooses a thread ID and block ID 
and steps it according to the thread-level judgment. This 
allows the semantics to model the full range of non-deterministic
behavior arising from the GPU's thread scheduler.\footnote{
  In reality, the GPU's warp scheduler issues instructions to threads in a warp
  in lockstep, but modeling every thread 
  as completely independent is both simpler and a conservative overestimate 
  of the nondeterministic behavior of the GPU. 
}

The thread-level judgment has the shape $\eta, \sigma, \Sigma, t, b, p \vdash^\capa s \stepsto s' \dashv \eta', \sigma', \Sigma'$, where 
\begin{itemize}
  \item $\eta$ denotes thread-local memory,
  \item $\sigma$ denotes shared memory,
  \item $\Sigma$ denotes global memory,
  \item $t$ denotes the thread's ID, 
  \item $b$ denotes the ID of the block in which the thread lives, and
  \item $p$ denotes the relative position of the thread within $\capa$ (the \unit{} ID).
\end{itemize}
Critically, notice that a $\capa$ also appears on the thread-level judgment 
just as it does on the typing judgment. This is because 
the thread semantics \emph{dynamically tracks and enforces \unit{}s}. 
The same way evaluation of a program ``gets stuck'' if a value does not have the right type, 
the semantics of \calcname also get stuck if code attempts to access data or invoke commands 
with the wrong \unit{}. As an example, observe the \rref{S-Alloc-Shared} rule in 
~\Cref{fig:calculus-semantics}, which requires a $(\block, 1)$ \unit{} and will 
fail to step if encountered with a different one. 
This runtime \unit{} is present in \calcname, but is erased by \name{} during
compilation; in \Cref{sec:core-theorem} we use it to prove that 
well-typed programs will
always execute with the same \unit{} that the type system viewed them with.

The semantic rules for \unit{}s involve manipulating $p$ to track 
which threads take which code paths when \unit{}s are \code{split} or \code{group}ed. 
Notice that in \rref{S-Program}, the thread-level judgment 
always begins with \unit{} $(\grid, 1)$: all the \unit{} management rules 
are congruences, narrowing the \unit{} of further evaluation as determined by the particular rule used. 

These rules take great care to ensure that $p$ always describes the relative position 
of a compute resource within its \unit{}; the payoff is that \calcname's 
semantics can later use this $p$ value to model the way that \name automatically 
adjusts indices into data when \code{partition}ing a data \unit{}. 
Beyond these key rules for managing \unit{} on code, we have modeled all other 
core features of \name, such as asynchronous operations and thread synchronization, in \calcname. 
To handle such features we equip the operational semantics with additional structure, 
including sets of semaphores \cite{dijkstra-over-0} for synchronization and 
a stack of effect handlers for modeling deferred asynchronous computations 
inspired by \citet{ahman-asynchronous-2021}. 
We have elided these details here for simplicity, but interested readers can find them
in their full complexity in Appendix \ref{appendix:proofs-semantics}.

\subsection{Soundness Theorem}\label{sec:core-theorem}
Together, the type system and operational semantics allow us to prove the following syntactic soundness
theorem, which says that \calcname programs are type safe and do not get stuck 
trying to execute operations for which they lack the required \unit{}:

\begin{theorem}{(Type-and-\Unit{} Safety).}
  For any program $s$ such that $\Gamma \vdash^\capa s$, either:
  \begin{enumerate}
    \item $s$ is \skipcom, or
    \item for any well-typed environments $\eta$, $\sigma$, and $\Sigma$,
      there is an $s'$, $\eta'$, $\sigma'$, and $\Sigma'$ such that \\ $\eta, \sigma, \Sigma, t, b, p \vdash^\capa s \stepsto^{\star} s' \dashv \eta', \sigma', \Sigma'$
      and $\Gamma' \vdash^\capa s'$, where $\Gamma'$ is an extension of $\;\Gamma$, and 
      $\eta'$, $\sigma'$, and $\Sigma'$ are well-typed with respect to $\Gamma'$.
  \end{enumerate}
\end{theorem}

\begin{proof}
  Via the usual progress and preservation lemmas, available in Appendix \ref{appendix:proofs-proofs}.
\end{proof}

It is worth noting that this soundness theorem guarantees a syntactic safety property,
not a liveness property:
it does not guarantee
that all threads sharing \unit{} $\capa$ that \emph{can} reach a program point typed with $\capa$ \emph{will} eventually do so.
Indeed, in the presence of nontermination,
liveness does not hold---some of the threads could \code{split} off and loop forever.
While we believe the liveness version of this theorem holds for a terminating fragment of \calcname, it
is not provable with syntactic methods; the proof would require semantic techniques
that are notoriously challenging and would be a research contribution \cite{turon-logical-2013, birkedal-concurrent-2012, farzan-proving-2016} 
in and of itself. We plan to tackle this proof for \calcname in future work.


%% file: 5-compilation.tex
\section{Implementation}\label{sec:compilation}

\name{} is implemented as an embedded language in Python. Once a program type
checks, \name{} lowers it to a CUDA file. All \unit{} information is
erased during this step, and the generated CUDA contains no run-time checks.
The file can then be compiled by \code{nvcc} \cite{nvidia-cuda-12-3-0}, 
NVIDIA's closed-source compiler, to produce an executable. Because \name{} operates at roughly the
same level of abstraction as CUDA, there is a one-to-one mapping between most language constructs and their CUDA counterparts. A notable change is the addition of three parameters to each device function:
the thread’s relative ID, the block’s ID, and an offset for shared memory allocations.

\paragraph{Inserting Synchronization}

\begin{wrapfigure}{r}{0.18\textwidth}
  \centering
  \vspace*{-2em}
  \includegraphics[width=0.13\textwidth]{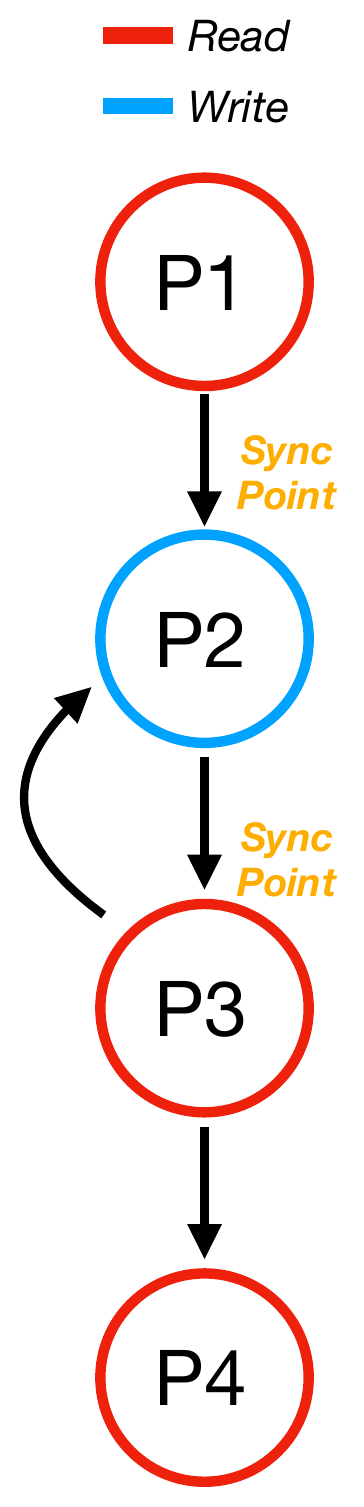} 
  \vspace*{-1em}
  \caption{An example data-control-flow graph with synchronization points.}
  \label{fig:sync-points}
  \vspace*{-2em}
\end{wrapfigure}

Most of \name{}’s implementation is straightforward, but
inserting synchronization points is more involved.
As described in ~\Cref{sec:data-unit}, once data has been \code{partition}ed, \name{} is
responsible for synchronizing the data after the \code{partition} ends.

To determine where this synchronization must occur, \name{}
constructs a \emph{data-control-flow graph} from the program. Nodes
correspond to \code{partition}s, and edges capture program-order precedence:
a parent \code{partition} must complete before its child begins. The graph can have
backedges introduced through loops. In this graph, each \code{partition} is categorized
as a read or a write \code{partition} by checking whether the  \code{partition}ed variable ever 
appears as an lvalue. 
Synchronization points are inserted according to the following scheme:

\begin{enumerate}
  \item If the \emph{parent} \code{partition} is a write, a synchronization point is
  inserted before the current one to ensure that it observes the most recent data; or
  \item If the \emph{current} \code{partition} is a write, a synchronization point is
  inserted before it to ensure that all preceding reads have
  completed.
\end{enumerate}

\Cref{fig:sync-points} shows an example graph with synchronization
points derived from these two conditions. The inferred synchronization points in
~\Cref{fig:global-mm} have also been marked on lines ~\refline{40}, \refline{48}, \refline{56}, and \refline{63}.

Using this information, \name{} emits \emph{wait} operations before \code{partition}s
begin and \emph{arrive} operations after they end, using CUDA's general \emph{split-barrier} \cite{mbarrier} primitive to implement them. For special cases, such as
synchronizing an entire block or warp, \name{} instead uses primitives
like \code{\_\_syncthreads()} or \code{\_\_syncwarp()}.

Synchronization for asynchronous data movement is handled in exactly the same
way and uses the same underlying graph. CUDA allows asynchronous loads
to be associated with a split barrier, so \name{} binds each asynchronous
transfer to the appropriate wait–arrive pair inferred from the graph. Certain
features, such as commit group-style synchronization \cite{commit-group,commit-group-bulk}, require additional
reasoning, and \name{} performs further static analysis to insert the
necessary synchronization.

It is worth noting that naively inserting synchronization immediately after each
\code{partition} would be correct but prohibitively slow. To avoid this, \name{}
applies two optimizations: a wait-motion pass
pushes waits downward toward the first use of the \code{partition}ed variable, and an
arrive-motion pass pulls arrives upward toward its last use.

%% file: 6-evaluation.tex
\section{Evaluation}\label{sec:eval}

Having explained the design of \name,
we now evaluate it in the context of three main questions:

\begin{description}[labelwidth=1em, leftmargin=4em]
    \item[\textbf{RQ1}] Can \name{} express a variety of composable CUDA programs?
    \item[\textbf{RQ2}] Can \name{} express programs that use advanced GPU
    features?
    \item[\textbf{RQ3}] Can \name{} match the performance of existing,
    speed-of-light CUDA code?
\end{description}


To perform this evaluation, we use two GPUs. The first is the NVIDIA H100 SXM5,
a server-grade chip that supports Tensor Core operations and a dedicated
hardware copy engine, the Tensor Memory Accelerator (TMA) \cite{bulk-cite}. Notably, the H100
introduces a new logical level called the warpgroup, and we show that our programming model can
accommodate it. Moreover, because the H100 has historically served
as the primary GPU for large-scale AI training, many CUDA kernels are
already highly optimized for this hardware and achieve near speed-of-light
performance, providing a rigorous baseline for comparison. To ensure our results
generalize beyond the H100, we also test programs on a second GPU, the
NVIDIA 4070 SuperTi---a consumer-grade chip.

As mentioned in ~\Cref{sec:compilation}, when \name{} typechecks a program, it
produces a CUDA file. We compile the CUDA file with \code{nvcc} version 12.3
~\cite{nvidia-cuda-12-3-0} with flags \code{-03 --use\_fast\_math}. We
initialize all inputs using a random number generator \cite{thonking} and report the average
runtime sampled over 10 iterations, following a warm-up phase of 5 iterations.
All aggregate results reported in this section are geometric means. All
programs in the evaluation have been reproduced in Appendix \ref{sec:code-appendix}.

\input{eval/6.1}

\input{eval/6.2}
\input{eval/6.3}

%% file: eval/6.1.tex
\subsection{RQ1: Can \name{} Express a Variety of CUDA Programs?}\label{sec:diversity}

We evaluate \name{}’s expressivity by writing programs that have fundamentally
different patterns of convergence, along with a library modeled after CUB  \cite{nvidia-cub-2025} that 
contains composable pieces.

\subsubsection{Programs with Different Convergence Behavior}

\paragraph{\textbf{Matrix Multiplication} (4070Ti)} We chose matrix multiplication as our
first benchmark for two main reasons. First, there exist several implementations that
achieve near-peak performance, providing a strong baseline. Second, it allows for a range of increasingly sophisticated implementations
that exercise different parts of the language, making it ideal for evaluating
expressiveness. 

We adapt the codebase from
~\citet{Boehm-SGEMM-CUDA} to implement \code{float} matrix multiplication, commonly referred to as \code{sgemm}. 
Concretely, we compute $C \leftarrow
\alpha A B + \beta  C$ where $A$, $B$ and $C$ are matrices and $\alpha$ and $\beta$ are scalars. As this is a
\code{float} benchmark, it does not need advanced GPU features
like asynchrony or Tensor Cores to achieve speed-of-light performance. We will discuss these in ~\Cref{sec:advance-features}.

We implement five variants of the \code{sgemm} benchmark: 
(1) a naive version
that follows the traditional single-program multiple-data pattern, written
from the \unit{} of a thread; 
(2) a version that exploits memory
coalescing \cite{coalesced}, still expressed at the thread level; 
(3) a version that builds on (2) by introducing 2D tiling and staging data in shared memory, which requires
shifting first to the \unit{} of a block and then to the \unit{} of a thread,
while also requiring block-level synchronization; 
(4) a version that applies 2D
tiling with vectorized loads, again written from the \unit{} of a block; and 
(5)
a version that combines the optimizations from (2) through (4) while adding an
additional level of tiling from the \unit{} of a warp.

We can express all variants cleanly, and the resulting programs are close to
their CUDA counterparts, which distinguish \unit{}s through disciplined style.
\name{}, on the other hand, enforces this discipline at compile time. We
evaluate the performance of these variants in ~\Cref{sec:eval-perf}.

\paragraph{\textbf{Single-Pass Parallel Prefix Scan with Decoupled Look-Back} (4070Ti)}

We also implement scan, a widely used parallel primitive, in \name{}. We focus
on the prefix-sum scan, which computes, for each position in an array, the sum
of all elements up to that position. Prefix sum sits in a different corner of
the GPU design space from matrix multiplication: it is memory intensive,
requires careful attention to the convergence behavior of threads, and
traditionally requires multiple passes over data.

We implement the single-pass parallel prefix scan with decoupled look-back,
introduced by
~\citet{MerrillGarland2016-SinglePassParallelPrefix}, an elegant algorithm that
does not require multiple passes over the input data, and involves several distinct points of convergence. Within each
block, work is decomposed into fine-grained thread-level and warp-level scans.
After producing the local result, blocks publish their partial prefix to global
memory. Finally, each block waits until enough information from
earlier blocks is available, at which point it accumulates the value
and completes its section of the scan.

We implement this full strategy in \name{}. Our implementation uses \code{unsafe}
to implement a global-memory spinlock that lets a block check when the
previous block's data is ready.

\input{eval/6.5}

%% file: eval/6.5.tex
\subsubsection{Case Study: Can \name{} help build composable functions?} \hfill\\

\noindent
While answering the other RQs, we found ourselves developing a small library
of functions---similar in spirit to CUB---that we would frequently call. 
In this section, we qualitatively study
how \name{} can help programmers design libraries that they can
compose with confidence. 

\begin{wrapfigure}{r}{0.55\textwidth}
    \centering
    \vspace{-2em}
\begin{cudaBox}
template<typename T, int BlockDimX, 
    int ItemsPerThread, BlockLoadAlgorithm Algorithm, 
    int BlockDimY = 1, int BlockDimZ = 1>
class BlockLoad: 

// --- Calling the load function by instantiating the class ---

using BlockLoad = cub::BlockLoad<int, 128, 4, BLOCK_LOAD_DIRECT>;
// Allocate shared memory for BlockLoad
__shared__ typename BlockLoad::TempStorage temp_storage;
int thread_data[4]; // Thread local data
BlockLoad(temp_storage).Load(d_data, thread_data);
\end{cudaBox}
\vspace*{-2em}
\caption{Using \code{BlockLoad} in CUB.}
\label{fig:block-load}
\vspace*{-1em}
\end{wrapfigure}

As mentioned in ~\Cref{sec:intro}, CUB occupies a unique design space in the GPU
library ecosystem. Unlike many other libraries such as cuBLAS \cite{cublas},
cuDNN \cite{cudnn}, and cuSPARSE \cite{cusparse}, which provide host-side
functions, CUB provides a device-side library
organized into different levels. It makes these levels apparent by prefixing
each of its functions with \code{Device}, \code{Block}, and \code{Thread}. The prefix sum already used variants of these functions, translated into \name{}. 
 
Let's turn our attention to a particular CUB
function---\code{BlockLoad}---and examine how it is equivalently expressed in
\name{}; we will see how \name{}'s type system reifies CUB's implicit assumptions.
In CUB, the load function is implemented as a class, as shown in Figure \ref{fig:block-load}.

CUB exposes a leaky abstraction, where information about the number of threads,
block sizes, and other details seeps through:

\begin{enumerate}
    \item The CUB documentation needs to specify the number of threads that the
    function can assume to be available, because within the function, each
    thread must locate itself in the computation and use its \code{threadIdx.x}
    accordingly.
    
    \item The “item per thread” design can serve two purposes.
    The first is performance: if loops have constant bounds, they can
    be unrolled. The second is
    correctness: the function relies on the assumption that all threads call the
    function with an equal number of values to load.
    
    \item The CUB documentation specifies that \code{thread\_data}
   can be data local to each thread.
    
    \item Finally, it says that if shared memory is
    being overwritten, a \code{\_\_syncthreads()} call must be made to ensure
    that all reads have completed.
\end{enumerate}

\begin{wrapfigure}{r}{0.40\textwidth}
    \centering
    \vspace{-3.5em}
    \begin{cudaBox}
        @prism("device")
        @requires(block[1], thread[32])
        def block_load(input : ptr(const(int)) @ block[1], 
                          output : ptr(int) @ thread[1], 
                          items_per_thread : int @ block[1]):
        \end{cudaBox}
\vspace*{-2em}
\caption{\code{block\_load} signature in \name{}.}
\label{fig:prism-load}
\vspace*{-1.5em}
\end{wrapfigure}

In \name{}, however, we are able to describe these requirements through 
the interface shown in Figure \ref{fig:prism-load}, reducing the need to
communicate numerous implementation details through documentation:

\begin{enumerate}
    \item We do not need to pass in the number of threads at all. Whenever
    \name{} calls a function, it track the relative thread ID, so each
    function can be written locally as if it were running alone, rather than
    having to determine where the thread resides on the grid.
    
    \item We do not need to make \code{item\_per\_thread} a template argument
    for \emph{correctness}. Its frequency is set at the function signature, so \name{}
    will never allow a function to be called with a value living at a narrower \unit{}. 
    
    \item In our interface, \code{thread\_data} is explicitly set to a
    thread-local value. Since it is not marked as \code{const}, \name{}
    conservatively assumes it may be written to, and enforces at compile time
    that only \code{thread[1]} values are passed in.
    
    \item Finally, using our synchronization pass outlined in ~\Cref{sec:compilation},
    a \code{\_\_syncthreads()} call will be inserted automatically if
    \code{input} is going to be used for writing.
\end{enumerate}

\begin{wrapfigure}{r}{0.42\textwidth}
    \centering
    \vspace{-2.5em}
\begin{cudaBox}
@prism("device")
@requires(thread[32])
# block_load calls into warp_load
def warp_load(input : ptr(const(int)) @ thread[32], 
                  output : ptr(int) @ thread[1], 
                  items_per_thread : int @ thread[32]):
\end{cudaBox}
\vspace*{-2em}
\caption{\code{warp\_load} signature in \name{}.}
\label{fig:prism-warp-load}
\vspace*{-0.5em}
\begin{cudaBox}
@prism("device")
@requires(thread[1])
# warp_load calls into thread_load
def thread_load(input : ptr(const(int)) @ thread[1], 
                    output : ptr(int) @ thread[1], 
                    items_per_thread : int @ thread[1]):
\end{cudaBox}
\vspace*{-2em}
\caption{\code{thread\_load} signature in \name{}.}
\label{fig:prism-thread-load}
\vspace*{-2.5em}
\end{wrapfigure}

Moreover, our version of CUB's function is built as a composition of two other functions, 
whose interfaces are shown in Figures \ref{fig:prism-warp-load} and \ref{fig:prism-thread-load}.
The \code{block\_load} function is implemented by a call to \code{warp\_load}, which in turn calls 
\code{thread\_load}. Had we mistakenly attempted to call \code{warp\_load} from inside
\code{thread\_load}, however, \name{} would reject this, rather than silently failing at runtime.

%% file: eval/6.2.tex
\subsection{RQ2: Can \name{} Express Programs that Use Advanced GPU Features?}\label{sec:advance-features}

To answer this question, we write a matrix multiplication for the \code{bf16} datatype on
the H100, also known as \code{hgemm}.
This benchmark is an acid test of our language, as \code{hgemm} pushes several language features to the extreme. To write an 
\code{hgemm} that can hit peak throughput on an H100, we need to write
a warp-specialized kernel that uses the TMA---an asynchronous hardware copy engine that can move tiles of data at a time---and
the warpgroup–level Tensor Core instructions, or \code{wgmma} \cite{wgmma}---new to the Hopper architecture. A high-performance kernel
for this matrix-multiplication overlaps computation with data movement by
pipelining loads.

The implementation in \name{} looks different from CUDA code,
particularly in how pipelining is expressed. Since \name{} uses named variables introduced by 
\code{partition}s or \code{claim}s to determine the
synchronization each region requires, when pipelining, we cannot
dynamically change the pipeline slot simply by maintaining an index
that wraps around based on the pipeline's length. Instead, each pipeline slots must
be given separate names so that \name{} can track them independently and overlap compute with data-movement. 
This leads to pipeline slots that
must be individually named and forces the load logic to be effectively
``unrolled''. This, in turn, forces all pipelines in \name{} to be statically sized. In practice,
these pipelines are statically sized anyway to ensure they fit in shared memory.

Notably, for this benchmark, in addition to \code{wgmma} and TMA, the program
needs to dynamically reallocate registers between producer and consumer
warpgroups, an instruction only available on Hopper. This redistribution is
a warpgroup-level collective operation, and \name{} can check it like any other. 
Moreover, getting \code{wgmma} to work did not require
introducing a new \unit{} into the language; \code{thread[128]} was sufficient.
We did, however, need to add a dedicated TMA-style asynchronous data-movement
construct, since \name{} must eventually insert the appropriate synchronization
for these transfers.

%% file: eval/6.3.tex
\begin{wrapfigure}{r}{0.45\textwidth}
    \centering
    \begin{minipage}{\linewidth}
    \centering
    \includegraphics[width=\linewidth]{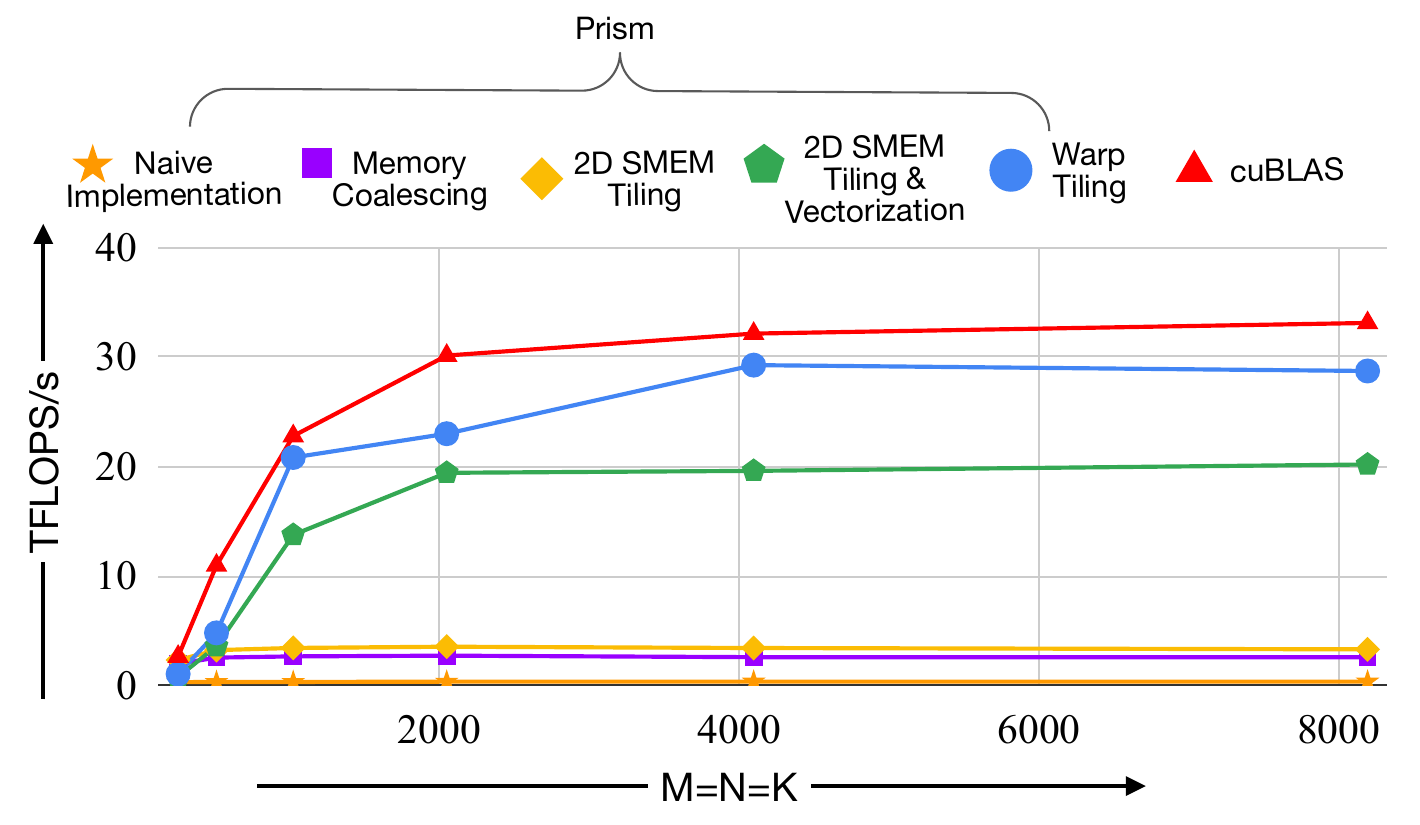}
    \caption{Performance of \code{sgemm} on square matrices as matrix dimension $M=N=K$ increases.
    \label{fig:sgemm}} 
    \end{minipage}
    \vspace*{-4em}
\end{wrapfigure}

\subsection{RQ3: Can \name{} Match the Performance of Existing, Speed-of-Light CUDA Code?}\label{sec:eval-perf}

In ~\Cref{sec:diversity} and ~\Cref{sec:advance-features}, we examined programs
that expressed the same computation in multiple ways, 
relied on multiple points of convergence, and used advanced GPU features. We now
discuss their performance.

\begin{wrapfigure}{r}{0.45\textwidth}
    \vspace{-1em}
    \centering
\begin{minipage}{\linewidth}
    \centering
    \includegraphics[width=\linewidth]{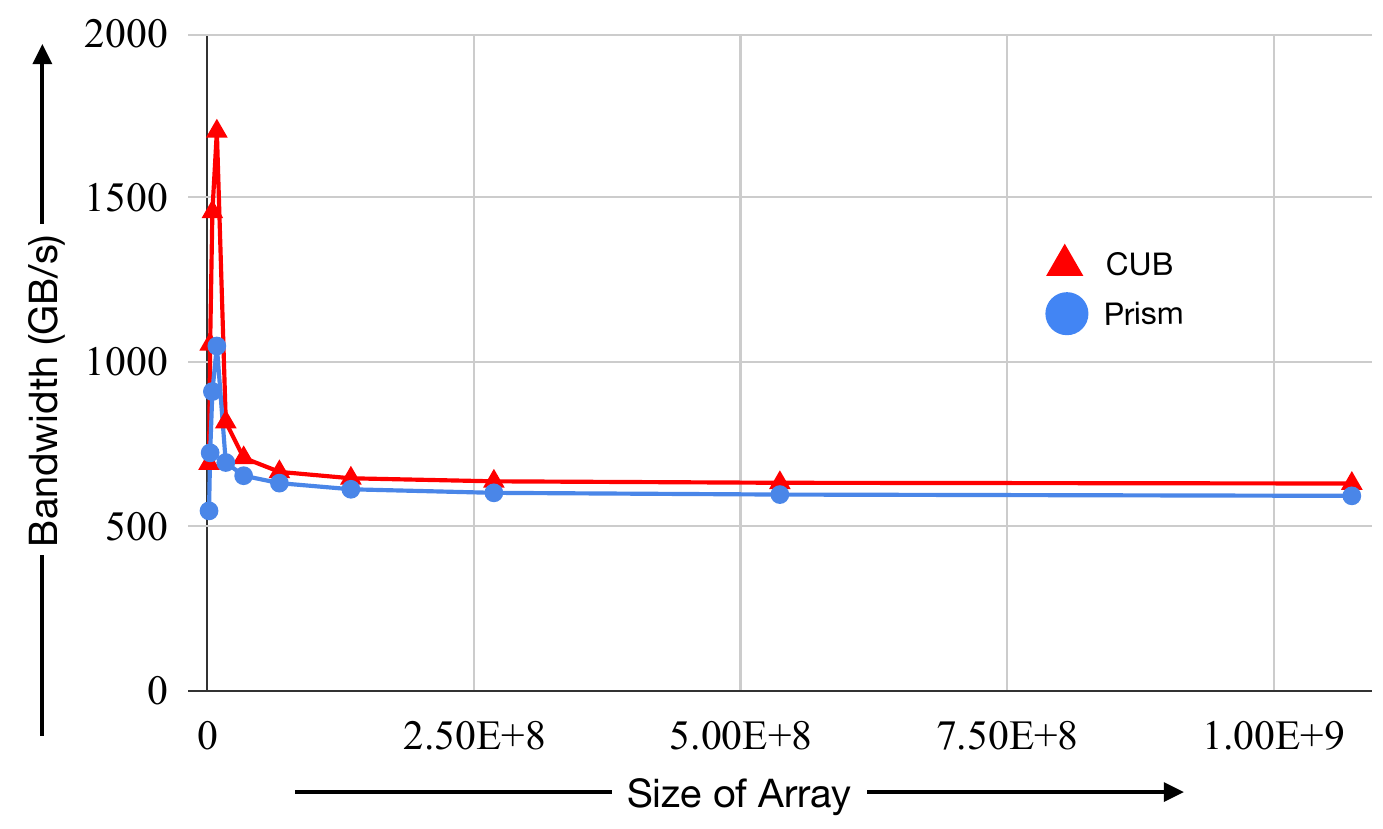}
    \vspace{-2em}
    \caption{Performance for prefix sum as input array size increases.
    \label{fig:prefix-sum}} 
    \end{minipage}
    
    \vspace{1ex}
    
    \begin{minipage}{\linewidth}
    \centering
    \includegraphics[width=\linewidth]{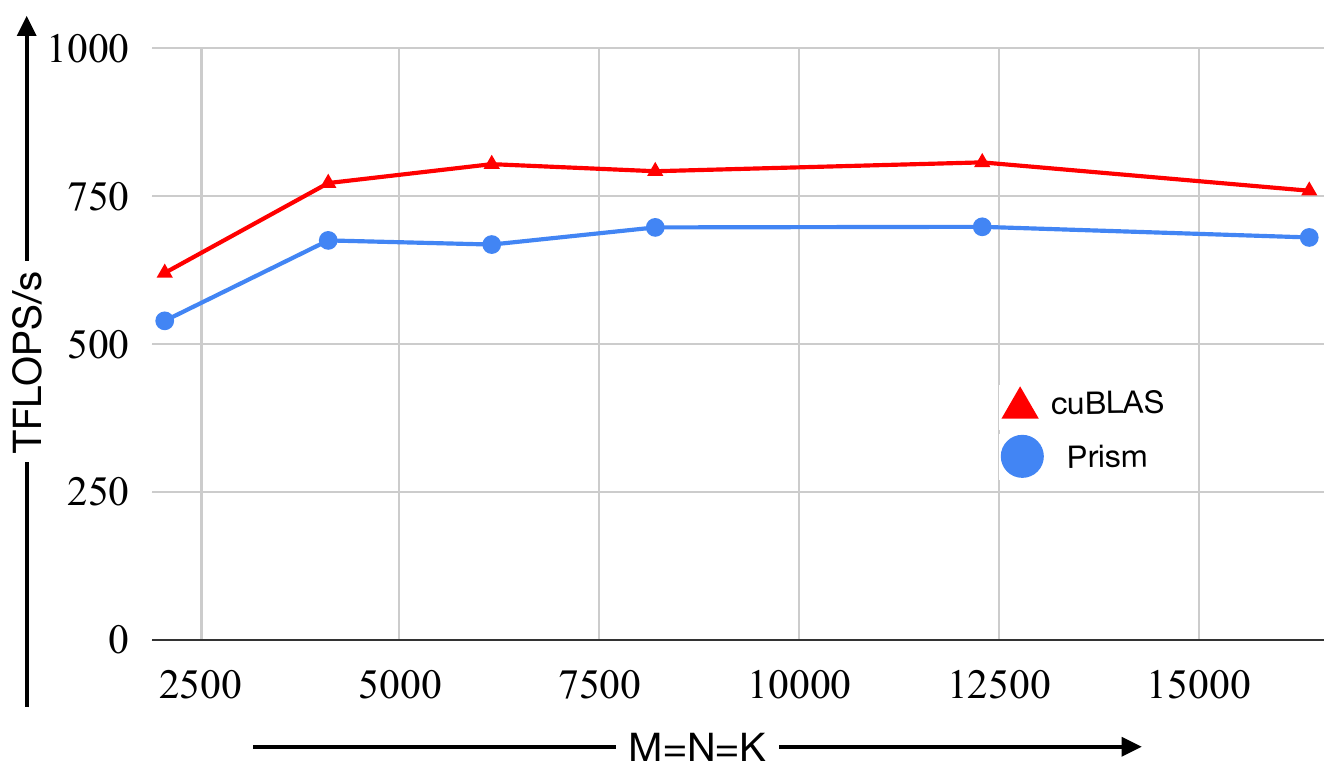}
    \vspace{-2em}
    \caption{Performance of \code{hgemm} on square matrices as matrix dimension $M=N=K$ increases.
    \label{fig:hgemm}}
    \end{minipage}
    \vspace{-2em}
    \end{wrapfigure}

The performance of the matrix multiplication benchmarks is shown in ~\Cref{fig:sgemm},
where we demonstrate that \name{} is competitive with cuBLAS \cite{cublas}.

For the prefix scan, we compare our performance to CUB in ~\Cref{fig:prefix-sum},
which shows that Prism is within 7\% of CUB's achieved bandwidth for arrays that do not fit in the L2 cache.

Finally, we evaluate our H100 implementation and show (in ~\Cref{fig:hgemm}) that it is competitive
with cuBLAS on square sizes, coming within 15\%. 
This performance difference arises due to one of \code{nvcc}'s 
optimization passes failing to trigger: the dynamic register allocation 
instruction is merely a hint---not a directive---to \code{nvcc}, 
and its optimization pass is sensitive to the way that pipeline structure is expressed. 
We emphasize that this
is an exacting benchmark, and coming close to cuBLAS's performance demonstrates 
\name{}'s ability to control low-level features.

%% file: 8-related-work.tex
\section{Related Work}\label{sec:related-work}

\name{} builds on a rich tradition of systems languages for GPU
programming and theoretical foundations of
parallel programming. 
\name{} differs from these systems in a crucial way: it treats collective operations, 
and therefore composability, as first-class concerns.

\paragraph{\textbf{Imperative Languages for GPUs}}

CUDA~\cite{nvidia-cuda-2025}, ROCm ~\cite{AMD-ROCm}, and OpenCL~\cite{Khronos-OpenCL} are imperative
programming languages that expose low-level access to GPU hardware. 
These languages offer no support for managing collective operations, statically or otherwise; 
instead, users must orchestrate computation on the machine by describing how individual threads execute code. 
The subtle failures permitted by this model are the motivation for \name{}.

Descend ~\cite{kopcke-descend-2024}, a new Rust-inspired \cite{rust} language, 
uses a type system to track aspects of the compute and memory hierarchy
and is thus the closest to our work. 
However, Descend's main concern is preventing data races, 
as opposed to ensuring that collective operations execute in valid contexts. 
As a result, Descend lacks support for many collective operations like Tensor Cores; 
in fact, because  Descend allows threads to read their own IDs and thus induce data-dependent control flow, 
adding such support is not possible with its current design. 
It is also worth noting that while Descend does formalize a type system, 
it does not attempt to prove any properties about it.

\paragraph{\textbf{Functional Languages for GPUs}}

Futhark ~\cite{futhark-2017}, Accelerate ~\cite{accelerate-grover} 
and Vertigo ~\cite{Elliott-Vertigo} are functional array languages with compilers targeting CUDA. 
These languages expose a high level interface, 
abstracting away details of the hardware entirely in exchange for stronger safety guarantees.  
\name{}, however, exposes low-level details of the hardware
and thus does not trade performance for safety. 

\paragraph{\textbf{Tile-Based Kernel DSLs}}
More recently, tile-based GPU languages---Triton ~\cite{triton-tillet-2019}, 
Pallas ~\cite{JAX_Pallas_GPU}, Tilus ~\cite{ding2025tilustilelevelgpgpuprogramming},  and
Helion ~\cite{PyTorch_Helion2025}---offer a middle ground between high-level abstraction 
and low-level control. However, these languages sidestep the question of composability
entirely because they restrict programmers to a single layer of the hierarchy:
a block (Triton, Tilus, Helion), or a warpgroup (Pallas). Gluon ~\cite{triton-intro-gluon}, meanwhile,  is a new low-level, tile-based language, but it does not check and cannot enforce that collective operations  are executed at the correct layer of the hierarchy.

\paragraph{\textbf{Task-Based and Scheduling Languages}}

Languages like Cypress ~\cite{cypress-yadav-2025}, Halide ~\cite{halide-jrk-2013}, Fireiron ~\cite{fireiron-hagedorn-2020} and RISE/ELEVATE ~\cite{bastian-2020} 
expose a scheduling language that lets users modify an existing reference program through external commands. 
Because these languages operate by transforming a fixed source of truth, they expose a fundamentally different programming model than \name{}.
    
\paragraph{\textbf{Libraries Built on CUDA}}

Many GPU libraries, such as CUTLASS ~\cite{Thakkar_CUTLASS_2023}, CUB ~\cite{nvidia-cub-2025}, and ThunderKittens ~\cite{spector2025thunderkittens}, 
expose device functions that operate at various levels of the compute hierarchy. 
These functions are typically organized via C++ namespaces to reflect their intended usage (e.g., warp-level, block-level). 
However, this organization is purely conventional: it encodes hierarchy through naming rather than enforcing it statically. 
As a result, correct use requires careful discipline from both the library implementer (to uphold naming and usage invariants) and the user (to correctly interpret them). 
Any mismatch or subtle misunderstanding between the intended and actual use of these functions goes unchecked by the compiler.

\paragraph{\textbf{Theoretical Foundations}} 
The design of \calcname is heavily inspired by existing 
work on dependency tracking \cite{abadi-core-1999}.
Dependency tracking calculi allow type systems to track how data and code depend on 
each other, and have commonly been used 
to implement secure information flow analyses \cite{denning-denning-info}. 
In \calcname, data that lives at a narrow \unit{} is unable to flow into data living at a broader \unit{},
and we use dependency tracking to capture this restriction in \calcname's type system.

Prior theoretical work has tackled reasoning about thread divergence. In particular, \citet{muller-modeling-2021} build a sophisticated quantitative program logic 
for this use case. \citet{singhania2018static}, meanwhile, uses static analysis techniques to predict thread divergence to unlock optimizations. 

%% file: 7-limitations.tex
\section{Conclusion, Limitations, and Future Work}\label{sec:limitations}

We have presented \name, a new, low-level GPU language that statically
guarantees safe usage of compute resources by construction, without sacrificing
low-level control. \name{} introduces a new mental model for writing GPU code,
which we are excited to make more expressive and ergonomic.

Particularly, we plan to improve the experience of writing pipelines in \name{}.
As discussed in Section \ref{sec:advance-features}, \name currently requires
pipeline slots to be explicitly named; we can make this more ergonomic this by
adding language support for generating pipeline-style code.

\name{} can be extended to more architectures; in particular, it can accommodate
newer GPUs—such as Blackwell—that support coarser-grained Tensor Core
operations. More broadly, we hope to generalize \name{}'s model to other
hierarchical compute environments, including distributed systems.

We also believe that Prism is capable of guaranteeing data-race freedom, but
additional work is required to support this both formally and in practice. In
particular, we think that if \name{} restricted users to \code{partition}s with
injective indexing functions, data-race freedom would follow naturally. We are also
interested in exploring how the design principles of
Descend~\cite{kopcke-descend-2024}, which build on Rust’s ownership
model~\cite{rust}, could be applied to \name{}.

On the theoretical side, we plan to explore a terminating fragment of
\calcname{} and prove the liveness property discussed in Section
\ref{sec:core-theorem}: that all threads sharing \unit{} $\capa$ eventually
reach the parts of a program viewed at that $\capa$. This amounts to showing
that threads sharing the same \unit{} execute the same code and observe the same
data, which we hope to prove using logical relations following
\citet{turon-logical-2013} and
\citet{spies-transfinite-2021,spies-transfinite-2021-1}.

We believe that \name{} is a promising low-level substrate that enables
confident composition and can serve as a foundation for building higher-level
libraries and abstractions.

%% file: 10-ack.tex
\section{Acknowledgements}

We would like to thank Ege Kabasakaloglu for writing the single-pass parallel
prefix scan with decoupled look-back in \name{}, Amanda Liu for early
conversations about \name{}’s formalization, and Adam Chlipala for helping us
crystalize the \code{partition} construct. Paul Biberstein, Christophe Gyurgyik,
Olivia Hsu, Fredrik Kjolstad, Scott Kovach, Benjamin Pierce, Alexander Root,
Adrian Sampson, Shiv Sundram, Stephanie Weirich, Rohan Yadav, and Bobby Yan
provided helpful feedback on many drafts of the paper.

%% file: appendix-proof.tex
\theoremstyle{definition}

\section{Complete \calcname Type System, Semantics, and Syntactic Soundness Proofs}\label{appendix:proofs}

\subsection{Basic Definitions}\label{appendix:proofs-defs}

\begin{align*}
  \text{Hierarchy Levels } \hier &\defas \grid \alt \block \alt \thread \\
  \text{Memory Kinds } l &\defas \local \alt \shared \alt \globalkind \\
  \text{\Unit{}s } \capa &: \hier \times \mathbb{N} \\
  \text{Base Types } b &\defas \code{bool} \alt \code{int} \alt \code{float} \\
  \text{Types } \tau &\defas b
  \alt b[]^l \alt \fun{\Gamma}{\capa}{\mem} \alt \async{\tau}\\
  \text{Contexts } \Gamma &\defas \cdot \alt \Gamma, x :^{\capa} \tau \\
  \text{Shared Memory Remaining } \mem &: \mathbb{N} \\
\end{align*}

\Unit{}s $\capa$ are part of an algebra parameterized over some constant values $\threads$ (the number of threads per block) and $\blocks$ (the number of blocks per grid). With these values, we have two isomorphisms:
$$
(\block, 1) \cong (\thread, \threads)
$$
and
$$
(\grid, 1) \cong (\block, \blocks)
$$

The \replicom and \splitcom operations of \name allows us to move along this isomorphism,
from left to right.
For clarity, in \calcname we split these operations into three:
a \splitcom operation that can split \unit{}s into multiple narrower ones,
and a \destructcom operation that directly moves us along the isomorphism, and
a \replicom that divides our current \unit{} into equal sized parts.

\Unit{}s $\capa$ are also lexicographically ordered in the obvious way.
$\hier$s are ordered such that $\thread \leq \block \leq \grid$,
and $(\hier_1, n_1) \leq (\hier_2, n_2)$ iff $n_1 \alt n_2$ and $\hier_1 \leq \hier_2$.

We define scalar multiplication $i \times h$ of natural numbers with $\hier$s:
$i \times (h, n) = (h, in)$.

We also define division of \unit{}s and hierarchy levels of type
$\capa \times \capa \rightarrow \mathbb{N}$.
$\grid / \block = \blocks$ and $\block / \thread = \threads$.
We lift this to \unit{}s like so: $(h_1, n_1) / (h_2, n_2) = ((h_1 / h_2) \cdot n_1) / n_2$.

Lastly we define a partial $\downarrow$ operator on $\hier$s such that $\downarrow \block = \thread$ and $\downarrow \grid = \block$.
Note that $\downarrow \thread$ is undefined.
This operator lifts to $\capa$s whose second component is 1 and encodes the leftward component of the isomorphism above:
$\downarrow (\block, 1) = (\thread, \threads)$ and
$\downarrow (\grid, 1) = (\block, \blocks)$.

Note also that the presentation of these rules in the main body of the paper elide the
$\mem$ portion, which tracks the maximum amount of memory a given computation is allowed to use.
In the full system presented here, both the typing rules and the operational semantics
carry an additional piece of information tracking allocated memory.

\subsection{Complete \calcname Typing Rules}\label{appendix:proofs-types}

\subsubsection{Expressions}

\small

$$
\infer[\textbf{T-Var}]
{\Gamma \vdash^{\capa} x : \tau}
{x :^{\capa} \tau \in \Gamma} \qquad
\infer[\textbf{T-Int}]
{\Gamma \vdash^{\capa} n : \code{int}}
{} \qquad
\infer[\textbf{T-Float}]
{\Gamma \vdash^{\capa} f : \code{float}}
{}
$$

$$
\infer[\textbf{T-Bool}]
{\Gamma \vdash^{\capa} b : \code{bool}}
{} \qquad
\infer[\textbf{T-Partition-Id}]
{\Gamma \vdash^{\capa} \code{partition\_id} : \code{int}}
{\capa < (\grid, 1)}
$$

$$
\infer[\textbf{T-Arr-Access}]
{\Gamma \vdash^{\capa} e_1[e_2] : \tau}
{\Gamma \vdash^{\capa'} e_1 : \tau[]^l & \Gamma \vdash^{\capa} e_2 : \code{int} & l = \globalkind \code{ or } l = \local & \capa \leq \capa'}
$$

$$
\infer[\textbf{T-Arr-Access-Shared}]
{\Gamma \vdash^{\capa} e_1[e_2] : \tau}
{\Gamma \vdash^{\capa'} e_1 : \tau[]^{\shared} & \Gamma \vdash^{\capa} e_2 : \code{int} & \capa \leq (\block, 1) & \capa \leq \capa'}
$$

$$
\infer[\textbf{T-Bop}]
{\Gamma \vdash^{\capa} e_1 \code{ bop } e_2 : \code{int}}
{\Gamma  \vdash^{\capa} e_1 : \code{int} & \Gamma \vdash^{\capa} e_2 : \code{int}} \qquad
\infer[\textbf{T-Cmp}]
{\Gamma  \vdash^{\capa} e_1 \code{ cmp } e_2 : \code{bool}}
{\Gamma  \vdash^{\capa} e_1 : \code{int} & \Gamma \vdash^{\capa} e_2 : \code{int}}
$$

\normalsize

\subsubsection{Statements}

\small

$$
\infer[\textbf{T-Function-Call}]
{\Gamma  \vdash^{\capa}_\mem f(e_1, \dots, e_n)}
{f :^\capa \fun{x_i : \tau_i}{\capa}{\mem'} \in \Gamma & \Gamma \vdash^{\capa} e_i : \tau_i &
\mem' \leq \mem}
$$

$$
\infer[\textbf{T-Split}]
{\Gamma \vdash^{(\hier, n)}_{\mem} \splitcom (n_1, n_2) \{ s_1 \} \{ s_2 \}}
{\Gamma \vdash^{(\hier, n_1)}_\mem s_1 & \Gamma \vdash^{(\hier, n_2)}_\mem s_2 & n_1, n_2 \alignto n} \qquad
\infer[\textbf{T-Destruct}]
{\Gamma \vdash^{\capa}_{\mem} \destructcom \code{ in } s}
{\Gamma \vdash^{\downarrow \capa}_\mem s}
$$

$$
\infer[\textbf{T-Group}]
{\Gamma  \vdash^{(\hier, q \cdot n)}_\mem \replicom \code{ } q \code{ }s}
{\Gamma  \vdash^{(\hier, n)}_\mem s} \qquad
\infer[\textbf{T-Assn}]
{\Gamma \vdash^{\capa}_\mem x = e}
{x :^{\capa'} \tau \in \Gamma & \Gamma \vdash^{\capa} e : \tau & \capa' \leq \capa
}
$$

$$
\infer[\textbf{T-Sync-Init}]
{\Gamma \vdash^{\capa}_\mem \code{init}_\psi}
{} \qquad
\infer[\textbf{T-Sync-Dec}]
{\Gamma \vdash^{\capa}_\mem \code{dec}_\psi}
{}
$$

$$
\infer[\textbf{T-Sync-Wait}]
{\Gamma \vdash^{\capa}_\mem \code{wait}_\psi}
{} \qquad
\infer[\textbf{T-Skip}]
{\Gamma \vdash^{\capa}_\mem \skipcom}
{} \qquad
\infer[\textbf{T-Free}]
{\Gamma \vdash^{\capa}_{\mem} \code{free } n}
{n \leq \mem}
$$

$$
\infer[\textbf{T-Decl}]
{\Gamma \vdash^{\capa}_\mem x : \tau\code{ @ }\capa' = e \code{ in } s}
{\Gamma \vdash^{\capa} e : \tau & \Gamma, x :^{\capa'} \tau \vdash^{\capa}_\mem s & \capa' \leq \capa & \tau \text{ not an array type } }
$$

$$
\infer[\textbf{T-Arr-Assn}]
{\Gamma \vdash^{\capa}_\mem e_1[e_2] = e_3}
{\Gamma \vdash^{\capa'} e_1 : \tau[]^l & \Gamma \vdash^{\capa} e_2 : \code{int} & \Gamma \vdash^{\capa} e_3 : \tau & l = \globalkind \code{ or } l = \local & \capa' \leq \capa 
}
$$

$$
\infer[\textbf{T-Arr-Assn-Shared}]
{\Gamma \vdash^{\capa}_\mem e_1[e_2] = e_3}
{\Gamma \vdash^{\capa'} e_1 : \tau[]^{\shared} & \Gamma \vdash^{\capa} e_2 : \code{int} & \Gamma \vdash^{\capa} e_3 : \tau & \capa \leq (\block, 1) & \capa' \leq \capa 
}
$$

$$
\infer[\textbf{T-If}]
{\Gamma \vdash^{\capa}_{\textbf{max}(\mem_1, \mem_2)} \code{if } e \code{ then } s_1 \code{ else } s_2}
{\Gamma \vdash^{\capa} e : \code{bool} & \Gamma \vdash^{\capa}_{\mem_1} s_1 & \Gamma \vdash^{\capa}_{\mem_2} s_2}
$$

$$
\infer[\textbf{T-While}]
{\Gamma \vdash^{\capa}_m \code{while } e \code{ do } s}
{\Gamma \vdash^{\capa} e : \code{bool} & \Gamma \vdash^{\capa}_m s}
\qquad
\infer[\textbf{T-Seq}]
{\Gamma \vdash^{\capa}_{\textbf{max}(\mem_1, \mem_2)} s_1; s_2}
{\Gamma \vdash^{\capa}_{\mem_1} s_1 & \Gamma \vdash^{\capa}_{\mem_2} s_2}
$$

$$
\infer[\textbf{T-Alloc}]
{\Gamma  \vdash^{\capa}_{\mem + n\cdot\code{size}(\tau)} x := \code{alloc } l \code{ } \tau \code{ } n \code{ in } s}
{\Gamma, x :^{\capa} \tau[]^l  \vdash^{\capa}_{\mem} s & l = \globalkind \code{ or } l = \local & }
$$

$$
\infer[\textbf{T-Alloc-Shared}]
{\Gamma  \vdash^{(\block, 1)}_{\mem + n\cdot\code{size}(\tau)} x := \code{alloc } \shared \code{ } \tau \code{ } n \code{ in } s}
{\Gamma, x :^{(\block, 1)} \tau[]^\shared  \vdash^{\capa}_{\mem} s & }
$$

$$
\infer[\textbf{T-Partition}]
{\Gamma, x :^{(\hier, n)} \tau[]^l \vdash^{(\hier, n)}_\mem \partitioncmd{\psi}{x}{y}{c}{s}}
{\Gamma, y :^{(\hier, n/c)} \tau[]^l \vdash^{(\hier, n)}_{\mem} s & c \alt n & l \neq \local}
$$

$$
\infer[\textbf{T-Claim}]
{\Gamma, x :^{(\hier, n)} \tau[]^l \vdash^{(\hier, n)}_\mem \claim{\psi}{x}{y}{n'}{s}}
{\Gamma, y :^{(\hier, n')} \tau[]^l \vdash^{(\hier, n')}_{\mem} s & n' \leq n & l \neq \local}
$$

$$
\infer[\textbf{T-Lower}]
{\Gamma, x :^{\capa} \tau[]^l \vdash^{\capa}_\mem \down{\psi} \code{ }x \code{ into  } y \code{ in } s}
{\Gamma, y :^{\downarrow \capa} \tau[]^l \vdash^{\capa}_{\mem} s & l \neq \local}
$$

$$
\infer[\textbf{T-Async-Partition}]
{\Gamma, x :^{(\thread, 1)} \tau[]^l \vdash^{(\thread, 1)}_\mem \asyncp{\phi}{x}{y}{s}}
{\Gamma, y :^{(\thread, 1)} \async{\tau[]^l} \vdash^{(\thread, 1)}_{\mem} s}
$$

$$
\infer[\textbf{T-Async-Memcpy}]
{\Gamma, x :^{(\thread, 1)} \async{\tau[]^l}, y :^{(\thread, 1)} \tau[]^{l'} \vdash^{(\thread, 1)}_\mem \asyncmemcpy{x}{y}}
{}
$$

$$
\infer[\textbf{T-Memcpy}]
{\Gamma, x :^\capa \tau[]^l, y :^\capa  \tau[]^{l'} \vdash^\capa_\mem \memcpy{x}{y}}
{}
$$

\normalsize

\subsection{Complete \calcname Semantics}\label{appendix:proofs-semantics}

\subsubsection{Definitions}

\begin{align*}
  \text{Global Memory } \Sigma &\defas \cdot \alt \Sigma, n \mapsto^{\capa} v \\
  \text{Shared Memory } \sigma &\defas \cdot \alt \sigma, n \mapsto^{\capa} v \\
  \text{Local Memory } \eta &\defas \cdot \alt \eta, n \mapsto^{\capa} v \\
  \\
  \text{Block Memory Map } S &\defas \forall n \in \blocks, n \mapsto \sigma \\
  \text{Thread Memory Map } L &\defas \forall n \in \threads, n \mapsto \eta \\
  \\
  \text{Synchronization Map } \Psi &: \psi \rightarrow \mathbb{N} \rightarrow \mathbb{N}\\
  \text{Deferred Computations Map } \Phi &: \phi \rightarrow \{s\}
\end{align*}

In real GPUs, thread IDs are only unique within their block. However,
in this calculus for simplicity we assume thread IDs are global. One can
convert back and forth between this abstracted notion of a thread ID and a block-unique ID via addition modulo $\threads$. That is, $t_{\text{real}} = t_{\text{simplified}} \textbf{ mod }  \threads$ and $t_{\text{simplified}} = t_{\text{real}} + b \cdot \threads$.

By convention the names for local and shared and global memory do not conflict, as on the GPU they will be separate pointer spaces. Additionally, we freely interchange between using names for variables and integer locations.

In the main body of the paper, for simplicity we elide the synchronization map and the
deferred computation map from the operational semantics, as our theorems do not make
any guarantees about non-interference. However, as they are part of the full semantics,
we include them here for completeness.
By convention the synchronization and deferred computation maps are a total functions,
initialized to map to $\lambda \_. 0$ for $\psi$s and
$\lambda \_. \{\}$ for $\phi$s not explicitly initialize.

The shape of the judgment for a single thread is $\eta, \sigma, \Sigma, t, b, p, \Psi \vdash^\capa_\mem s \stepsto s' \dashv_{\mem'} \eta', \sigma', \Sigma', \Psi'$.
The $t$ here is the thread ID, the $b$ is the block ID, and the $p$ is the \unit{} ID. The last of these three is modified and managed by the rules for \splitcom, \replicom and \destructcom, and tracks the relative position of the thread within a group. This semantics is in a small step style.

The shape of the judgment for expressions is $\eta, \sigma, \Sigma \vdash^\capa_{\capa'} e \Downarrow v$. The two $\capa$s represent the ambient compute context (i.e., the context in which resources are being read), while $\capa'$, represents the target compute context (i.e, the compute context of the variable into which the result of the expression is going to be written. This is relevant for computing the value of \code{partition\_id}, which divides the two contexts.
  As a shorthand, we can divide a \unit{} by a scalar value like so: $(\hier, n)/c = (\hier, n) / (\hier, c)$.

  The overall evaluation of a program is expressed as
  $$L, S, \Sigma, P, \Psi, \Phi \stepsto L', S', \Sigma', P', \Psi', \Phi'.$$ In this judgment $P$ serves as a \textit{thread pool}, mapping pairs of thread and block IDs (which don't change) to statements and memory (which can be updated by stepping). One can think of $P$ as tracking which program is running on each thread. This steps according to the following rule:

  $$
  \infer[\textbf{S-Program}]
  {L, S, \Sigma, P, \Psi, \Phi \stepsto L[t \mapsto \eta'], S[b \mapsto \sigma'], \Sigma', P[(t, b) \mapsto (s', \mem')], \Psi', \Phi'}
  {L(t), S(b), \Sigma, t, b, 0, \Psi, \Phi \vdash^{(\grid, 1)}_{\mem} s \stepsto s' \dashv_{\mem'} \eta', \sigma', \Sigma', \Psi', \Phi' & P(t, b) = (s, \mem)}
  $$

  For simplicity of notation, we define an \update operation that searches the three environments for the one that contains the variable being used (by convention, there is no conflict between the environments, as in reality they exist in three separate address spaces). We also define a similar \get operation that retrieves a variable from memory, and a \rename operation that remaps a variable with the same value but under a different name.

  \begin{align*}
    \update(\eta, \sigma, \Sigma, x, v) &= (\eta[x \mapsto^\capa v], \sigma, \Sigma) \textit{ when } x \in^\capa \eta \\
    \update(\eta, \sigma, \Sigma, x, v) &= (\eta, \sigma[x \mapsto^\capa v], \Sigma) \textit{ when } x \in^\capa \sigma \\
    \update(\eta, \sigma, \Sigma, x, v) &= (\eta, \sigma, \Sigma[x \mapsto^\capa v]) \textit{ when } x \in^\capa \Sigma
  \end{align*}
  \begin{align*}
    \get(\eta, \sigma, \Sigma, x) &= \eta(x) \textit{ when } x \in^\capa \eta \\
    \get(\eta, \sigma, \Sigma, x) &= \sigma(x) \textit{ when } x \in^\capa \sigma \\
    \get(\eta, \sigma, \Sigma, x) &= \Sigma(x) \textit{ when } x \in^\capa \Sigma
  \end{align*}
  \begin{align*}
    \rename(\eta, \sigma, \Sigma, x, y, \capa') &= (\eta[y \mapsto^{\capa'} \eta(x)], \sigma, \Sigma) \textit{ when } x \in^\capa \eta \\
    \rename(\eta, \sigma, \Sigma, x, y, \capa') &= (\eta, \sigma[y \mapsto^{\capa'} \sigma(x)], \Sigma) \textit{ when } x \in^\capa \sigma \\
    \rename(\eta, \sigma, \Sigma, x, y, \capa') &= (\eta, \sigma, \Sigma[y \mapsto^{\capa'} \Sigma(x)]) \textit{ when } x \in^\capa \Sigma \\
  \end{align*}

  \subsubsection{\Unit{} Management Rules}

  \small

  $$
  \infer[\textbf{S-Split-Left}]
  {\eta, \sigma, \Sigma, t, b, p, \Psi, \Phi \vdash^{(\hier, n)}_{\mem} \splitcom(n_1, n_2) \{s_1\} \{s_2\} \stepsto \splitcom(n_1, n_2) \{s_1'\} \{s_2\} \dashv_{\mem'} \eta', \sigma', \Sigma', \Psi', \Phi'}
  {p < n_1 & n_1, n_2 \alignto n & \eta, \sigma, \Sigma, t, b, p, \Psi, \Phi \vdash^{(\hier, n_1)}_{\mem} s_1 \stepsto
  s_1' \dashv_{\mem'} \eta', \sigma', \Sigma', \Psi', \Phi'}
  $$

  $$
  \infer[\textbf{S-Split-Left-Done}]
  {\eta, \sigma, \Sigma, t, b, p, \Psi, \Phi \vdash^{(\hier, n)}_{\mem} \splitcom(n_1, n_2) \{\skipcom\} \{s_2\} \stepsto \skipcom\dashv_{\mem} \eta, \sigma, \Sigma, \Psi, \Phi}
  {p < n_1 & n_1, n_2 \alignto n}
  $$

  $$
  \infer[\textbf{S-Split-Right}]
  {\eta, \sigma, \Sigma, t, b, p, \Psi, \Phi \vdash^{(\hier, n)}_{\mem} \splitcom(n_1, n_2) \{s_1\} \{s_2\} \stepsto \splitcom(n_1, n_2) \{s_1\} \{s_2'\} \dashv_{\mem'} \eta', \sigma', \Sigma', \Psi', \Phi'}
  {p \geq n_1 & p < n_1 + n_2 & n_1, n_2 \alignto n & \eta, \sigma, \Sigma, t, b, p - n_1, \Psi, \Phi \vdash^{(\hier, n_2)}_{\mem} s_2 \stepsto s_2' \dashv_{\mem'} \eta', \sigma', \Sigma', \Psi', \Phi'}
  $$

  $$
  \infer[\textbf{S-Split-Right-Done}]
  {\eta, \sigma, \Sigma, t, b, p, \Psi, \Phi \vdash^{(\hier, n)}_{\mem} \splitcom(n_1, n_2) \{s_1\} \{\skipcom\} \stepsto \skipcom\dashv_{\mem} \eta, \sigma, \Sigma, \Psi, \Phi}
  {p \geq n_1 & p < n_1 + n_2 & n_1, n_2 \alignto n}
  $$

  $$
  \infer[\textbf{S-Split-None}]
  {\eta, \sigma, \Sigma, t, b, p, \Psi, \Phi \vdash^{(\hier, n)}_{\mem} \splitcom(n_1, n_2) \{s_1\} \{s_2\} \stepsto
  \skipcom \dashv_{\mem} \eta, \sigma, \Sigma, \Psi, \Phi}
  {p \geq n_1 + n_2 & n_1, n_2 \alignto n}
  $$

  $$
  \infer[\textbf{S-Destruct-Block}]
  {\eta, \sigma, \Sigma, t, b, 0, \Psi, \Phi \vdash^{(\block, 1)}_{\mem} \destructcom \code{ in } s \stepsto \destructcom \code{ in } s'
  \dashv_{\mem'} \eta', \sigma', \Sigma', \Psi', \Phi'}
  {\eta, \sigma, \Sigma, t, b, t \textbf{ mod } \threads, \Psi, \Phi \vdash^{(\thread, \threads)}_{\mem} s \stepsto s'\dashv_{\mem'} \eta', \sigma', \Sigma', \Psi', \Phi'}
  $$

  $$
  \infer[\textbf{S-Destruct-Grid}]
  {\eta, \sigma, \Sigma, t, b, 0, \Psi, \Phi \vdash^{(\grid, 1)}_{\mem} \destructcom \code{ in } s \stepsto \destructcom \code{ in } s'
  \dashv_{\mem'} \eta', \sigma', \Sigma', \Psi', \Phi'}
  {\eta, \sigma, \Sigma, t, b, b \textbf{ mod } \blocks, \Psi, \Phi \vdash^{(\block, \blocks)}_{\mem} s \stepsto s'\dashv_{\mem'} \eta', \sigma', \Sigma', \Psi', \Phi'}
  $$

  $$
  \infer[\textbf{S-Destruct-Done}]
  {\eta, \sigma, \Sigma, t, b, 0, \Psi, \Phi \vdash^\capa_{\mem} \destructcom \code{ in skip}  \stepsto \skipcom
  \dashv_{\mem} \eta, \sigma, \Sigma, \Psi, \Phi}
  {}
  $$

  $$
  \infer[\textbf{S-Group}]
  {\eta, \sigma, \Sigma, t, b, p, \Psi, \Phi \vdash^{(\hier, q \cdot n)}_{\mem} \replicom \code{ } q \code{ } s \stepsto \replicom \code{ } q \code{ } s'; \dashv_{\mem'} \eta', \sigma', \Sigma', \Psi', \Phi'}
  {\eta, \sigma, \Sigma, t, b, p \textbf{ mod } n, \Psi, \Phi \vdash^{(\hier, n)}_{\mem} s \stepsto s' \dashv_{\mem'} \eta', \sigma', \Sigma', \Psi', \Phi'}
  $$

  $$
  \infer[\textbf{S-Group-Done}]
  {\eta, \sigma, \Sigma, t, b, p, \Psi, \Phi \vdash^\capa_{\mem} \replicom \code{ } q \code{ skip} \stepsto \skipcom; \dashv_{\mem} \eta, \sigma, \Sigma, \Psi, \Phi}
  {}
  $$

  \normalsize

  \subsubsection{Thread Synchronization}

  We define a \textbf{size} operation on \unit{}s to compute the size of a \unit (the number of individual compute resources sharing it). 
  The operation is defined thusly:

  \begin{align*}
    \textbf{size}(\thread, n) &= n \\
    \textbf{size}(\block, n) &= n \cdot \threads \\
    \textbf{size}(\grid, n) &= n \cdot \blocks \cdot \threads \\
  \end{align*}

  \small

  $$
  \infer[\textbf{S-Sync-Wait-Done}]
  {\eta, \sigma, \Sigma, t, b, p, \Psi, \Phi \vdash^\capa_{\mem} \code{wait}_\psi \stepsto \skipcom \dashv_{\mem} \eta, \sigma, \Sigma, \Psi, \Phi}
  {\Psi(\psi)(p) = 0}
  $$

  $$
  \infer[\textbf{S-Sync-Wait-Spin}]
  {\eta, \sigma, \Sigma, t, b, p, \Psi, \Phi \vdash^\capa_{\mem} \code{wait}_\psi \stepsto \code{wait}_\psi \dashv_{\mem} \eta, \sigma, \Sigma, \Psi, \Phi}
  {\Psi(\psi)(p) \neq 0}
  $$

  $$
  \infer[\textbf{S-Sync-Dec}]
  {\eta, \sigma, \Sigma, t, b, p, \Psi, \Phi \vdash^\capa_{\mem} \code{dec}_\psi \stepsto \skipcom \dashv_{\mem} \eta, \sigma, \Sigma, \Psi', \Phi}
  {\Psi' = \Psi(\psi)[p \mapsto \Psi(\psi)(p) - 1]}
  $$

  $$
  \infer[\textbf{S-Sync-Init-Zero}]
  {\eta, \sigma, \Sigma, t, b, p, \Psi, \Phi \vdash^\capa_{\mem} \code{init}_\psi \stepsto \skipcom \dashv_{\mem} \eta, \sigma, \Sigma, \Psi', \Phi}
  {\Psi(\psi)(p) = 0 & \Psi' = \Psi(\psi)[p \mapsto \textbf{size}(\capa)]}
  $$

  $$
  \infer[\textbf{S-Sync-Init-Nonzero}]
  {\eta, \sigma, \Sigma, t, b, p, \Psi, \Phi \vdash^\capa_{\mem} \code{init}_\psi \stepsto \skipcom \dashv_{\mem} \eta, \sigma, \Sigma, \Psi, \Phi}
  {\Psi(\psi)(p) \neq 0}
  $$

  \normalsize

  \subsubsection{Asynchrony}

  \small

  $$
  \infer[\textbf{S-Async-Partition-Congr}]
  {\deduce{\asyncp{\phi}{x}{y}{s'} \dashv_{\mem} \eta', \sigma', \Sigma', \Psi', \Phi'}{\eta, \sigma, \Sigma, t, b, p, \Psi, \Phi \vdash^{(\thread, 1)}_{\mem'} \asyncp{\phi}{x}{y}{s} \stepsto}}
  {\rename(\eta, \sigma, \Sigma, x, y, (\thread, 1)), t, b, p, \Psi, \Phi \vdash^{(\thread, 1)}_{\mem} s \stepsto s' \dashv_{\mem'} \eta', \sigma', \Sigma', \Psi', \Phi'}
  $$

  $$
  \infer[\textbf{S-Async-Partition-Unwind}]
  {\deduce{(\asyncp{\phi}{x}{y}{s}) \dashv_{\mem} \eta, \sigma, \Sigma, \Psi, \Phi'}{\eta, \sigma, \Sigma, t, b, p, \Psi, \Phi \vdash^{(\thread, 1)}_{\mem} \asyncp{\phi}{x}{y}{\skipcom} \stepsto}}
  {\Phi = \Phi'[\phi \mapsto \Phi'(\phi) \cup \{s\}]}
  $$

  $$
  \infer[\textbf{S-Async-Partition-Done}]
  {\deduce{\skipcom \dashv_{\mem} \eta, \sigma, \Sigma, \Psi, \Phi}{\eta, \sigma, \Sigma, t, b, p, \Psi, \Phi \vdash^{(\thread, 1)}_{\mem} \asyncp{\phi}{x}{y}{\skipcom} \stepsto}}
  {\Phi(\phi) = \emptyset}
  $$

  $$
  \infer[\textbf{S-Async-Memcpy}]
  {\deduce{\skipcom \dashv_{\mem} \eta, \sigma, \Sigma, \Psi, \Phi[\phi \mapsto \Phi(\phi) \cup \{ \memcpy{x}{y} \}]}{\eta, \sigma, \Sigma, t, b, p, \Psi, \Phi \vdash^{(\thread,1)}_{\mem} \asyncmemcpy{x}{y} \stepsto}}
  {}
  $$

  $$
  \infer[\textbf{S-Memcpy}]
  {\eta, \sigma, \Sigma, t, b, p, \Psi, \Phi \vdash^\capa_{\mem} \memcpy{x}{y} \stepsto \skipcom \dashv_{\mem} \eta', \sigma', \Sigma', \Psi, \Phi}
  {(\eta', \sigma', \Sigma') = \update(\eta, \sigma, \Sigma, x, \get(\eta, \sigma, \Sigma, y))}
  $$

  \normalsize

  \subsubsection{Variables and Memory}

  \small

  $$
  \infer[\textbf{S-Decl}]
  {\eta, \sigma, \Sigma, t, b, p, \Psi, \Phi \vdash^\capa_{\mem} x : \tau \code{ @ } \capa' := e \code{ in } s \stepsto s \dashv_{\mem} \eta[x \mapsto^{\capa'} v], \sigma, \Sigma, \Psi, \Phi}
  {\eta, \sigma, \Sigma \vdash^\capa_{\capa'} e \Downarrow v & \capa' \leq \capa}
  $$

  $$
  \infer[\textbf{S-Free}]
  {\eta, \sigma, \Sigma, t, b, p, \Psi, \Phi \vdash^\capa_{\mem} \code{free } n \dashv_{\mem-n} \eta, \sigma, \Sigma, \Psi, \Phi}
  {}
  $$

  $$
  \infer[\textbf{S-Alloc-Local}]
  {\deduce{s;  \code{ free } (n \cdot \text{size}(\tau)) \dashv_{\mem + n \cdot \text{size}(\tau)} \eta[x \mapsto^\capa \langle x, n \rangle], \sigma, \Sigma, \Psi, \Phi}{\eta, \sigma, \Sigma, t, b, p, \Psi, \Phi \vdash^\capa_{\mem} x := \code{alloc } \local \code{ } \tau \code{ } n \code{ in } s \stepsto}}
  {}
  $$

  $$
  \infer[\textbf{S-Alloc-Shared}]
  {\deduce{s;  \code{ free } (n \cdot \text{size}(\tau)) \dashv_{\mem + n \cdot \text{size}(\tau)} \eta, \sigma[x \mapsto^\capa \langle x, n \rangle], \Sigma, \Psi, \Phi}{\eta, \sigma, \Sigma, t, b, p, \Psi, \Phi \vdash^\capa_{\mem} x := \code{alloc } \shared \code{ } \tau \code{ } n \code{ in } s \stepsto}}
  {\capa = (\block, 1)}
  $$

  $$
  \infer[\textbf{S-Alloc-Global}]
  {\deduce{s;  \code{ free } (n \cdot \text{size}(\tau)) \dashv_{\mem + n \cdot \text{size}(\tau)} \eta, \sigma, \Sigma[x \mapsto^\capa \langle x, n \rangle], \Psi, \Phi}{\eta, \sigma, \Sigma, t, b, p, \Psi, \Phi \vdash^\capa_{\mem} x := \code{alloc } \globalkind \code{ } \tau \code{ } n \code{ in } s \stepsto}}
  {}
  $$

  $$
  \infer[\textbf{S-Assn}]
  {\eta, \sigma, \Sigma, t, b, p, \Psi, \Phi, \Phi \vdash^\capa_{\mem} x = e \stepsto \skipcom \dashv_{\mem} \eta', \sigma', \Sigma', \Psi, \Phi}
  {x \in^{\capa'} \eta, \sigma, \Sigma & \eta, \sigma, \Sigma \vdash^\capa_{\capa'} e \Downarrow v & (\eta', \sigma', \Sigma') = \update(\eta, \sigma, \Sigma, x, v) & \capa' \leq \capa}
  $$

  $$
  \infer[\textbf{S-Arr-Assn}]
  {\eta, \sigma, \Sigma, t, b, p, \Psi, \Phi \vdash^\capa_{\mem} e_1[e_2] = e_3  \stepsto \skipcom \dashv_{\mem} \eta', \sigma', \Sigma', \Psi, \Phi}
  {\deduce{\deduce{\eta, \sigma, \Sigma \vdash^\capa_{\capa'} e_3 \Downarrow v}{\eta, \sigma, \Sigma, t, b, p \vdash^\capa_{\capa} e_2 \Downarrow i}}{\eta, \sigma, \Sigma, t, b, p \vdash^\capa_{\capa} e_1 \Downarrow \langle l, n \rangle} &
  \deduce{\deduce{(\eta', \sigma', \Sigma') = \update(\eta, \sigma, \Sigma, l + i, v)}{x \in^{\capa'} \eta, \sigma, \Sigma}}{i < n & \capa' \leq \capa}}
  $$

  $$
  \infer[\textbf{S-Partition}]
  {\eta, \sigma, \Sigma, t, b, p, \Psi, \Phi \vdash^\capa_{\mem} \partition{\psi}{x}{y}{c}{s} \stepsto \code{init}_\psi;  s'; \code{dec}_\psi; \code{wait}_\psi \dashv_{\mem} \eta', \sigma', \Sigma', \Psi, \Phi}
  {s' = s[(y + c \cdot p)/ y] & (\eta', \sigma', \Sigma') = \rename(\eta, \sigma, \Sigma, x, y, \capa/c)}
  $$

  $$
  \infer[\textbf{S-Claim}]
  {\deduce{\code{init}_\psi;  \splitcom(n_1, n_2)\{s'\}\{\skipcom\}; \code{dec}_\psi; \code{wait}_\psi \dashv_{\mem} \eta', \sigma', \Sigma', \Psi, \Phi}{\eta, \sigma, \Sigma, t, b, p, \Psi, \Phi \vdash^{(\hier, n_1 + n_2)}_{\mem} \claim{\psi}{x}{y}{n_1}{s} \stepsto}}
  {(\eta', \sigma', \Sigma') = \rename(\eta, \sigma, \Sigma, x, y, (\hier, n_1))}
  $$

  $$
  \infer[\textbf{S-Lower}]
  {\eta, \sigma, \Sigma, t, b, p, \Psi, \Phi \vdash^\capa_{\mem} \down{\psi} \textbf{ } x \code{ into  } y \code{ in } s \stepsto \code{init}_\psi; s; \code{dec}_\psi; \code{wait}_\psi \dashv_{\mem} \eta', \sigma', \Sigma', \Psi, \Phi}
  {(\eta', \sigma', \Sigma') = \rename(\eta, \sigma, \Sigma, x, y, \downarrow\capa)}
  $$

  \normalsize

  \subsubsection{Control Flow}

  \small

  $$
  \infer[\textbf{S-If-True}]
  {\eta, \sigma, \Sigma, t, b, p, \Psi, \Phi \vdash^\capa_{\mem} \code{if } e \code{ then } s_1 \code{ else } s_2 \stepsto s_1 \dashv_{\mem} \eta, \sigma, \Sigma, \Psi, \Phi}
  {\eta, \sigma, \Sigma \vdash^\capa_{\capa} e \Downarrow \code{true}}
  $$

  $$
  \infer[\textbf{S-If-False}]
  {\eta, \sigma, \Sigma, t, b, p, \Psi, \Phi \vdash^\capa_{\mem} \code{if } e \code{ then } s_1 \code{ else } s_2 \stepsto s_2 \dashv_{\mem} \eta, \sigma, \Sigma, \Psi, \Phi}
  {\eta, \sigma, \Sigma\vdash^\capa_{\capa} e \Downarrow \code{false}}
  $$

  $$
  \infer[\textbf{S-While}]
  {\eta, \sigma, \Sigma, t, b, p, \Psi, \Phi \vdash^\capa_{\mem} \code{while } e \code{ do } s \stepsto \code{if } e \code{ then } (s; \code{ while } e \code{ do } s) \code{ else } \skipcom \dashv_{\mem} \eta, \sigma, \Sigma, \Psi, \Phi}
  {}
  $$

  $$
  \infer[\textbf{S-Seq-First}]
  {\eta, \sigma, \Sigma, t, b, p, \Psi, \Phi \vdash^\capa_{\mem} s_1; s_2 \stepsto s_1'; s_2 \dashv^{\capa'}_{\mem'} \eta', \sigma', \Sigma', \Psi', \Phi'}
  {\eta, \sigma, \Sigma, t, b, p, \Psi, \Phi \vdash^\capa_{\mem} s_1 \stepsto s_1' \dashv^{\capa'}_{\mem'} \eta', \sigma', \Sigma', \Psi', \Phi'}
  $$

  $$
  \infer[\textbf{S-Seq-Done}]
  {\eta, \sigma, \Sigma, t, b, p, \Psi, \Phi \vdash^\capa_{\mem} \skipcom; s_2 \stepsto s_2 \dashv_{\mem} \eta, \sigma, \Sigma, \Psi, \Phi}
  {}
  $$

  $$
  \infer[\textbf{S-Function-Call}]
  {\eta, \sigma, \Sigma, t, b, p, \Psi, \Phi \vdash^\capa_{\mem} f(e_1, \dots, e_n) \stepsto s \dashv_{\mem} \eta[x_i \mapsto v_i], \sigma, \Sigma, \Psi, \Phi}
  {\Sigma(f)= \{[x_1 : \tau_1, \dots, x_n : \tau_n], s\}  &\sigma, \Sigma \vdash^\capa_{\capa} e_i \Downarrow v_i & \mem' \leq \mem}
  $$

  \normalsize

  \subsubsection{Expressions}

  \small

  $$
  \infer[\textbf{E-Partition-Id}]
  {\eta, \sigma, \Sigma \vdash^\capa_{\capa'} \code{partition\_id} \Downarrow \capa/\capa' - 1}
  {\capa < (\grid, 1) & \capa' \leq \capa}
  $$

  $$
  \infer[\textbf{E-Var}]
  {\eta, \sigma, \Sigma \vdash^\capa_{\capa'} x \Downarrow \get(\eta, \sigma, \Sigma, x)}
  {}
  $$

  $$
  \infer[\textbf{E-Arr-Access}]
  {\eta, \sigma, \Sigma \vdash^\capa_{\capa'} e_1[e_2] \Downarrow \get(\eta, \sigma, \Sigma, l + i)}
  {\eta, \sigma, \Sigma \vdash^\capa_{\capa'} e_1 \Downarrow \langle l, n \rangle & \eta, \sigma, \Sigma \vdash^\capa_{\capa'} e_2 \Downarrow i & i < n & \capa' \leq \capa}
  $$

  $$
  \infer[\textbf{E-Int}]
  {\eta, \sigma, \Sigma \vdash^\capa_{\capa'} \vdash n \Downarrow n}
  {} \qquad
  \infer[\textbf{E-Bool}]
  {\eta, \sigma, \Sigma \vdash^\capa_{\capa'} \vdash b \Downarrow b}
  {}
  $$

  $$
  \infer[\textbf{E-Bop}]
  {\eta, \sigma, \Sigma \vdash^\capa_{\capa'} \vdash e_1 \code{ bop } e_2 \Downarrow v}
  {\eta, \sigma, \Sigma \vdash^\capa_{\capa'} \vdash e_1 \Downarrow v_1 & \eta, \sigma, \Sigma \vdash^\capa_{\capa'} e_2 \Downarrow v_2 & v = v_1 \code{ bop } v_2}
  $$

  $$
  \infer[\textbf{E-Cmp}]
  {\eta, \sigma, \Sigma \vdash^\capa_{\capa'} \vdash e_1 \code{ cmp } e_2 \Downarrow v}
  {\eta, \sigma, \Sigma \vdash^\capa_{\capa'} \vdash e_1 \Downarrow v_1 & \eta, \sigma, \Sigma \vdash^\capa_{\capa'} e_2 \Downarrow v_2 & v = v_1 \code{ cmp } v_2}
  $$

  \normalsize

  \subsection{Theorems and Proofs}\label{appendix:proofs-proofs}

  Note that in this section we assume no out of bounds array accesses.
  In general \name (and by extension \calcname) makes no guarantees about array out of bounds.

  \subsubsection{More definitions}

  As a premise to our type safety theorems, we need to assume we have a well-typed environment, written $\Gamma \vdash \eta, \sigma, \Sigma$. We define what this means inductively

  \small

  $$
  \infer[\textbf{V-Int}]
  {\eta, \sigma, \Sigma \vdash n : \code{int}}{} \qquad
  \infer[\textbf{V-Bool}]
  {\eta, \sigma, \Sigma \vdash b : \code{bool}}{} \qquad
  \infer[\textbf{V-Float}]
  {\eta, \sigma, \Sigma \vdash f : \code{float}}{} \qquad
  $$

  $$
  \infer[\textbf{V-Array}]
  {\eta, \sigma, \Sigma \vdash \langle x, n \rangle : \tau[]^{l}}{\forall i < n, \eta, \sigma, \Sigma \vdash \get(\eta, \sigma, \Sigma, x + i) : \tau} \qquad
  \infer[\textbf{V-Function}]
  {\eta, \sigma, \Sigma \vdash \{x_i:\tau_i, s\} : \fun{x_i:\tau_i}{\capa}{\mem}}{\Gamma, x_i:^\capa\tau_i \vdash^\capa_\mem s & \Gamma \vdash \cdot, \cdot, \textbf{fns }\Sigma}
  $$

  $$
  \infer[\textbf{G-Empty}]
  {\cdot  \vdash \eta, \sigma, \Sigma}
  {} \qquad
  \infer[\textbf{G-Int}]
  {\Gamma, x :^\capa \code{int} \vdash \eta, \sigma, \Sigma}
  {\eta(x) =^\capa v & \eta, \sigma, \Sigma \vdash v : \code{int}}
  $$

  $$
  \infer[\textbf{G-Bool}]
  {\Gamma, x :^\capa \code{bool} \vdash \eta, \sigma, \Sigma}
  {\eta(x) =^\capa v & \eta, \sigma, \Sigma \vdash v : \code{bool}} \qquad
  \infer[\textbf{G-Float}]
  {\Gamma, x :^\capa \code{float} \vdash \eta, \sigma, \Sigma}
  {\eta(x) =^\capa v & \eta, \sigma, \Sigma \vdash v : \code{float}} \qquad
  $$

  $$
  \infer[\textbf{G-Local}]
  {\Gamma, x :^\capa \tau[]^{\local} \vdash \eta, \sigma, \Sigma}
  {\eta(x) =^\capa v & \eta, \sigma, \Sigma \vdash v : \tau[]} \qquad
  \infer[\textbf{G-Shared}]
  {\Gamma, x :^\capa \tau[]^{\shared} \vdash \eta, \sigma, \Sigma}
  {\sigma(x) =^\capa v & \eta, \sigma, \Sigma \vdash v : \tau[]}
  $$

  $$
  \infer[\textbf{G-Global}]
  {\Gamma, x :^\capa \tau[]^{\globalkind} \vdash \eta, \sigma, \Sigma}
  {\Sigma(x) =^\capa v & \eta, \sigma, \Sigma \vdash v : \tau[]} \qquad
  \infer[\textbf{G-Function}]
  {\Gamma, x :^\capa \fun{\Gamma'}{\capa}{\mem} \vdash \eta, \sigma, \Sigma}
  {\Sigma(x) =^\capa v & \eta, \sigma, \Sigma \vdash v : \fun{\Gamma'}{\capa}{\mem}}
  $$

  \normalsize

  We can prove a couple simple lemmas about well-typed environments under operations like \textbf{rename}, \textbf{update}, and \textbf{get}.

  \begin{lemma}{(Well-typed \get)}
    If $\Gamma \vdash \eta, \sigma, \Sigma$ and $x :^\capa \tau \in \Gamma$ then $\eta, \sigma, \Sigma \vdash \get(\eta, \sigma, \Sigma, x) : \tau$.
  \end{lemma}

  \begin{lemma}{(Well-typed \rename)}
    If $\Gamma, x:^\capa \tau \vdash \eta, \sigma, \Sigma$ then $\Gamma, x:^\capa \tau, y:^{\capa'} \tau \vdash \rename(\eta, \sigma, \Sigma, x, y, \capa')$.
  \end{lemma}

  \begin{lemma}{(Well-typed \update)}
    If $\Gamma \vdash \eta, \sigma, \Sigma$ and $\eta, \sigma, \Sigma \vdash v : \tau$ then $\Gamma, x:^\capa \tau \vdash \update(\eta, \sigma, \Sigma, x, v)$.
  \end{lemma}

  We also define a well-formedness precondition on $p$ with respect to $\capa$:
  \begin{align*}
    (\hier, n) \vdash p &\defas p < n
  \end{align*}

  We also define well-formedness for the async stack:

  \begin{align*}
    \Gamma \vdash \Phi &\defas \forall \phi, s \in \Phi(\phi), \Gamma \vdash^{(\thread, 1)}_m s
  \end{align*}

  \subsection{Proofs}

  \begin{lemma}{(Expression Safety)}
    If $\Gamma \vdash^\capa e : \tau$ and $\Gamma \vdash \eta, \sigma, \Sigma$ and $\capa \vdash p$ and $\capa' \leq \capa$,
    then there is some $v$ such that $\eta, \sigma, \Sigma \vdash^\capa_{\capa'} e \Downarrow v$ and $v : \tau$.
  \end{lemma}
  \begin{proof}
    This proof proceeds by induction on the typing relation for expressions. Despite the fact that this property implies termination, we do not need a logical relation to prove it because the expression language is very simple.

    The \textbf{T-Int}, \textbf{T-Float},  and \textbf{T-Bool} cases are trivial, using the rules \textbf{E-Int}, \textbf{E-Bool} and \textbf{E-Float} to compute values.
    In the case for \textbf{T-Partition-Id},
    the $\capa$ premises of the typing rules and our assumption that $\capa' \leq \capa$
    match the premises of the evaluation rule, so this rule is simple as well.

    The cases for \textbf{T-Bop} and \textbf{T-Cmp} follow directly
    from the inductive hypotheses, assuming a valid and correctly implemented set of binary operators and comparators.

    The only interesting cases are \textbf{T-Arr-Access} and
    \textbf{T-Arr-Access-Shared}.

    In both cases our inductive hypotheses and inversion give us that
    $\capa' \leq \capa$, and $e_1$ evaluates to a $\langle x, n \rangle$,
    and that all the values between $x$ and $x+n$ in the appropriate environment are typed at $\tau$. We also know that $e_2$ evaluates to an integer $i$.
    We assume that all array accesses are in bounds, so $i < l$, which is sufficient to use the \textbf{E-Arr-Access} rule to complete this case, and the proof.
  \end{proof}

  \begin{lemma}{(Expression Determinism)}
    If $\eta, \sigma, \Sigma, t, b, p \vdash^\capa_{\capa'} e \Downarrow v_1$ and
    $\eta, \sigma, \Sigma, t, b, p \vdash^\capa_{\capa'} e \Downarrow v_2$ then $v_1 = v_2$.
  \end{lemma}
  \begin{proof}
    Straightforward by induction on the semantic derivation.
  \end{proof}

  \begin{lemma}{(Expression Well-Typedness)}
    If $\Gamma \vdash^\capa_{\capa'} e : \tau$ and $\Gamma \vdash \eta, \sigma, \Sigma$ and $\capa \vdash p$, and $\eta, \sigma, \Sigma, t, b, p \vdash^\capa_{\capa'} e \Downarrow v$, then $v : \tau$.
  \end{lemma}
  \begin{proof}
    By our expression safety lemma our well-typed expression must evaluate to a well-typed value $v'$. By our determinism lemma $v'$ must be the same as $v$, so $v$ is well-typed.
  \end{proof}

  \begin{lemma}{(Statement Progress)}
    If $\Gamma \vdash \eta, \sigma, \Sigma$ and $\Gamma \vdash^\capa_{\mem} s$ and $\capa \vdash p$, then either $s$ is \skipcom or there is some $s'$ such that $\eta, \sigma, \Sigma, t, b, p, \Psi, \Phi \vdash^\capa_{\mem'} s \stepsto s' \dashv_{\mem''} \eta', \sigma', \Sigma, \Psi', \Phi'$
  \end{lemma}
  \begin{proof}
    This proceeds by induction on the typing derivation.

    \subsubsection*{\Unit{} Management Rules}

    \begin{itemize}
      \item Case \textbf{T-Split}:

        In this case we have by our assumption that $\capa \vdash p$ that $p < n$. We also have
        that $n_1, n_2 \alignto n$, so $n_1 + n_2 \leq n$.
        There are three cases to consider, then: when $p < n_1$, when $p \geq n_1$ and $p < n_1 + n_2$,
        and when $p \geq n_1 + n_2$.

        In the first case, we have by our inductive hypothesis that $s$ is either \skipcom or that it can step in an $(\hier, n_1)$ context.
        In the former case we can use the \textbf{S-Split-Left-Done} rule and in the latter we can use the \textbf{S-Split-Left rule}.

        The second case is almost symmetric. The only additional work we have to do is to argue that $(\hier, n_2) \vdash p - n_1$, or equivalently that $p - n_1 < n_2$. This, however, is immediate from our assumption that $p < n_1 + n_2$.

        In the last case, we just use the \textbf{S-Split-None} rule to step to \skipcom.

      \item Case \textbf{T-Destruct}

        By our inductive hypothesis, we know that $s$ can step at $\downarrow\capa$ if $\downarrow\capa \vdash p$. The $\downarrow$ operation is only defined at $(\block, 1)$ or $(\grid, 1)$, so we only need to consider the cases where $\capa$ is one of those.

        In the former case $p$ becomes $t \textbf{ mod } \threads$ while $\downarrow\capa$ is $(\thread, \threads)$. $t \textbf{ mod } \threads < \threads$ for any $t$ so this satisfies the requirement that $\capa \vdash p$, which lets us use our inductive hypothesis: $s$ is either \skipcom or can step.
        If it can step, we can use this to satisfy the premise of \textbf{S-Destruct-Block} to step in this case. If it is \skipcom, then we use the rule \textbf{S-Destruct-Done} to step instead.

        The latter case is the same, except using the fact that $b\textbf{ mod } \blocks < \blocks$ and the $\textbf{S-Destruct-Grid}$ rule.

      \item Case \textbf{T-Group}

        In this case we have by assumption that $(\hier, q \cdot n) \vdash p$, i.e.,
        that $p < q \cdot n$.

        In this case we have our IH that
        if $(\hier, n) \vdash p'$ for some $p'$, then $s$ is either \textbf{skip}
        or steps with $(\hier, n)$ \unit{} with $p'$ as our \unit{} ID.

        We choose $p'$ to be $p \textbf{ mod } n$. This is always $< n$, so $(\hier, n) \vdash p'$.
        This lets us use our IH to get that $s$ is either \skipcom (in which case we can use the \textbf{S-Group-Done} rule to step) or itself steps, which lets us use the \textbf{S-Group} rule to step.
    \end{itemize}

    \subsubsection*{Thread Synchronization Rules}
    \begin{itemize}
      \item Case \textbf{T-Sync-Wait}

        $\Psi(\psi)(p)$ is either zero or it is not. In the former case we use the \textbf{S-Sync-Wait-Done} rule and in the latter we use \textbf{S-Sync-Wait-Spin}.

      \item Case \textbf{T-Sync-Dec}

        We use the \textbf{S-Sync-Dec} rule to step.

      \item Case \textbf{T-Sync-Init}

        We use the \textbf{S-Sync-Init-Zero} or \textbf{S-Sync-Init-Nonzero} rules depending on whether $\Psi(\psi)(p)$ is zero or not.

    \end{itemize}

    \subsubsection*{Asynchrony Rules}
    \begin{itemize}
      \item Case \textbf{T-Async-Partition}

        In this case, we have via our IH that if $\Gamma, y :^{(\thread, 1)} \async{\tau[]^l} \vdash \eta', \sigma', \Sigma'$, then we can either step $s$ with $(\thread, 1)$ \unit{} under environments $\eta'$, $\sigma'$, and $\Sigma'$, or $s$ is \skipcom.

        We have by assumption that $\Gamma, x :^{(\thread, 1)} \async{\tau[]^l} \vdash \eta, \sigma, \Sigma$. By our environment renaming lemma, this gives us what we need to use our IH with $(\eta', \sigma', \Sigma')$ as $\rename(\eta, \sigma, \Sigma, x, y, (\thread, 1))$.

        Thus $s$ either steps or is \skipcom. In the former case we can step with \textbf{S-Async-Partition-Congr}, and in the latter we can use either \textbf{S-Async-Partition-Unwind} or \textbf{S-Async-Partition-Done} depending on whether $\Phi(\phi)$ is empty or not.

      \item Case \textbf{T-Async-Memcpy}

        Immediate via use of the \textbf{S-Async-Memcpy} rule.

      \item Case \textbf{T-Memcpy}

        Immediate via use of the \textbf{S-Memcpy} rule.
    \end{itemize}

    \subsubsection*{Memory Rules}
    \begin{itemize}
      \item Case \textbf{T-Decl}

        Via our lemma about expression type safety and our hypothesis that $e$ is well-typed, we obtain the premises necessary to use the \textbf{S-Decl} rule to step.

      \item Case \textbf{T-Arr-Assn}

        Each of $e_1$, $e_2$ and $e_3$ must evaluate to a well-typed value by the expression type safety lemma. In particular, both $e_1$ evaluates to some $\langle l, n \rangle$ and $e_2$ evaluates to some $i$. We assume all array accesses are in bounds, so this is sufficient to use the \textbf{S-Arr-Assn} rule to step.

      \item Case \textbf{T-Arr-Assn-Shared}

        Same as previous case.

      \item Case \textbf{T-Free}

        Trivial via the \textbf{S-Free} rule.

      \item Case \textbf{T-Partition}

        Trivial via the \textbf{S-Partition} rule.
      \item Case \textbf{T-Claim}

        Trivial via the \textbf{S-Claim} rule.

      \item Case \textbf{T-Lower}

        Trivial via the \textbf{S-Lower} rule.
      \item Case \textbf{T-Alloc}

        We assume that $l$ is not \shared, so we can use the \textbf{S-Alloc-Local} or \textbf{S-Alloc-Global} rule, depending on whether $l$ is \local or \globalkind.

      \item Case \textbf{T-Alloc-Shared}

        Trivial via the \textbf{S-Alloc-Shared} rule.
    \end{itemize}

    \subsubsection*{Control Rules}
    \begin{itemize}
      \item Case \textbf{T-Skip}

        Trivial

      \item Case \textbf{T-While}

        Trivial, all while loops step via the \textbf{S-While} rule

      \item Case \textbf{T-If}

        By our proof of expression type safety, the expression $e$ steps to either the boolean value true or false. We can thus use either the \textbf{S-If-True} or \textbf{S-If-False} rules to step.

      \item Case \textbf{T-Seq}

        In this case we know by our IH that $s_1$ is either \skipcom or can step. In the former case we use the \code{S-Seq-Done} rule and in the latter we use the \code{S-Seq-First} rule.

      \item Case \textbf{T-Function-Call}

        In this case we know by our expression safety lemma that each of the arguments will evaluate to a well-typed value. We also have by assumption that $f$ has a function type, which by inversion on the \textbf{V-Function} rule tells us that it is a closure type. Additionally our assumption that $\Gamma \vdash \eta, \sigma, \Sigma$ tell us that $\Sigma$ contains $f$ at the same type that $\Gamma$ does. These premises are sufficient to use the \textbf{S-Function-Call} rule.
    \end{itemize}
  \end{proof}

  \begin{lemma}{(Statement Preservation)}
    If $\Gamma \vdash \eta, \sigma, \Sigma$ and $\Gamma \vdash^\capa_{\mem} s$ and
    $\eta, \sigma, \Sigma, t, b, p, \Psi, \Phi \vdash^\capa_{\mem'} s \stepsto s' \dashv_{\mem''} \eta', \sigma', \Sigma, \Psi', \Phi'$ and $\capa \vdash p$ and $\mem \geq \mem'$ and $\Gamma \vdash \Phi$, then there is some $\Gamma'$ such that $\Gamma \subseteq \Gamma'$ and $\Gamma' \vdash^\capa_{\mem} s'$ and $\Gamma' \vdash \eta', \sigma', \Sigma'$ and $\mem \geq \mem''$ and $\Gamma' \vdash \Phi'$.
  \end{lemma}
  \begin{proof}

    We proceed by induction on the derivation of $\Gamma \vdash^\capa_{\mem} s$.

    \subsubsection*{\Unit{} Management Rules}

    \begin{itemize}
      \item Case \textbf{T-Split}

        In this case we have by assumption that $(\hier, n) \vdash p$, $\Gamma \vdash \eta, \sigma, \Sigma$, and $\eta, \sigma, \Sigma, t, b, p, \Psi, \Phi \vdash^\capa_{\mem'} \splitcom(n_1, n_2)\{s_1\}\{s_2\} \stepsto s' \dashv_{\mem''} \eta', \sigma', \Sigma, \Psi', \Phi'$.
        Our inductive hypotheses give us that if for any $p$, if $(\hier, n_1) \vdash p$ and $s_1$ steps with \unit{} ID $p$, or if $(\hier, n_2) \vdash p$ and $s_2$ steps with \unit{} ID $p$ then their results are well typed.

        By inversion on our semantic derivation, we are in one of 5 cases.

        In the \textbf{S-Split-Left} case $p < n_1$ and $s_1$ steps to $s_1'$. This is sufficient to tell us that $s_1'$ is well-typed and the output environments of that relation $\eta', \sigma',$ and $\Sigma'$ are all well-typed by $\Gamma' \supseteq \Gamma$, and that the memory is properly bounded by the typing rules.

        We can thus use the \textbf{T-Split-Left} rule to conclude that the result of this case is well-typed.

        The \textbf{S-Split-Left-Done} rule is trivial via the $\textbf{T-Skip}$ rule.

        The \textbf{Right} cases are symmetric, with the observation that when $p \geq n_1$ and $p < n_1 + n_2$ then $p - n_1 < n_2$.

        The last \textbf{S-Split-None} rule is trivial via the $\textbf{T-Skip}$ rule.

      \item Case \textbf{T-Destruct}

        In this case we have by assumption that $\downarrow\capa$ is defined, so $\capa$ is either $(\block, 1)$ or $(\grid, 1)$. We also assume that
        $\eta, \sigma, \Sigma, t, b, p, \Psi, \Phi \vdash^\capa_{\mem'} \destructcom \code{ in } s \stepsto s' \dashv_{\mem''} \eta', \sigma', \Sigma, \Psi', \Phi'$. We also have by our inductive hypothesis that for any $p$ such that $\downarrow \capa \vdash p$ and $s''$ such that $s$ steps to $s''$ at $p$, then that step preserved well-typedness.

        By inversion on the step relation, we are in one of three cases.

        If the rule used \textbf{S-Destruct-Block}, then we know that $s$ steps to $s''$ and $p$ is $t \textbf{ mod } T$. $(\thread, \threads) \vdash t  \textbf{ mod } T$ for any $t$, so we can use our inductive hypothesis to conclude that the $s''$ stepped to by $s$ is well-typed, as are its environments and memory usage. The \textbf{T-Destruct} rule then gives us our desired goal.

        The \textbf{S-Destruct-Grid} case proceeds similarly, while the \textbf{S-Destruct-Done} case is trivial.

      \item Case \textbf{T-Group}

        In this case we have by assumption that \\
        $\eta, \sigma, \Sigma, t, b, p, \Psi, \Phi \vdash^{(\hier, q\cdot n)}_{\mem'} \replicom\code{ }q\code{ }s \stepsto s' \dashv_{\mem''} \eta', \sigma', \Sigma, \Psi', \Phi'$.
        We have by our inductive hypothesis that $s$ steps to $s''$ at some \unit{} ID $p'$ and $(h, n) \vdash p'$, then $s''$ is well typed,
        as are the other outputs of that step.

        By inversion on the step relation, we are in one of two cases.
        The \textbf{S-Group-Done} case is trivial, so we shall focus on the \textbf{S-Group} case.
        In this case we have that $s$ steps to $s''$ at \unit{} ID $p \textbf{ mod } n$.
        It is always the case that $(h, n) \vdash p \textbf{ mod } n$ for any $p$,
        so we can use our inductive hypothesis to conclude that $s''$ is well-typed,
        as are its output environments and memory usage.
        From there, it is a simple application of the \textbf{T-Group} rule to conclude that $\replicom\code{ }q\code{ }s'$ is well-typed, and to finish the case.
    \end{itemize}

    \subsubsection*{Thread Synchronization Rules}

    These rules are all trivial: with one exception all thread synchronization primitives step to skip without changing environment or memory, and are thus obviously well-typed. \textbf{S-Sync-Wait-Spin} does not produce skip, but it steps to the same statement as we already assumed typechecks in the premise of the lemma, so is straightforward nonetheless.

    If we wanted to say something about deadlock freedom we'd have more work here, but we aren't doing that, so these rules are easy.

    \subsubsection*{Asynchrony Rules}
    \begin{itemize}
      \item Case \textbf{T-Async-Partition}
        In this case we have that $s$ is well-typed in a context where $x$ has been renamed into $y$, with $(\thread, 1)$ \unit{}. We also have that $\Gamma, x :^{(\thread, 1)} \tau[]^l \vdash \eta, \sigma, \Sigma$ and $\Gamma, x :^{(\thread, 1)} \tau[]^l \Phi$.

        By inversion, we are in one of three cases.

        In the \textbf{S-Async-Partition-Done} case, we are done.

        In the \textbf{S-Async-Partition-Congr case}, our inductive hypothesis gives us that there is some $\Gamma'$ such that $\Gamma, y :^{(\thread, 1)} \async{\tau[]^l}  \subseteq \Gamma'$ and $\Gamma' \vdash \Phi'$ and $\Gamma' \vdash \rename(\eta, \sigma, \Sigma, x, y, (\thread, 1))$ via our well-typed renaming lemma.  This lets us use the \textbf{T-Async-Partition} rule to check this case, with a choice of $\Gamma'$ as $\Gamma', x :^{(\thread, 1)} \tau[]^l $ .

        In the \textbf{S-Async-Partition-Unwind case}, our assumption that $\Phi$ is well-typed tells us that $\Gamma, y :^{(\thread, 1)} \vdash^\{(\thread, 1)_m s$. Thus, we can use the \textbf{T-Async-Partition} rule to type this case.

        \item Case \textbf{T-Async-Memcpy}

          In this case the statement and environment typing are trivial, we need only to show that the async stack remains well typed.

          In this case we have that $x$ and $y$ have the same type at $(\thread, 1)$. This is sufficient for us to check $x = y$ at $(\thread, 1)$, meaning that adding that instruction to the stack maintains its well-typedness.

        \item Case \textbf{T-Memcpy}

          Immediate via use of the \textbf{S-Memcpy} rule. We just need to show that
          the environment remains well typed, which we know via our lemmas about \update and \get.
      \end{itemize}

      \subsubsection*{Memory Rules}
      \begin{itemize}
        \item Case \textbf{T-Decl}

          By inversion, we are using \textbf{S-Decl} rule for evaluation. Our well-typed expression lemma gives us that $v$ is well-typed, so it follows from our assumptions and our lemmas about extending environments that the extended $\eta$ and $s$ are well-typed by $\Gamma, x:^\capa \tau$.

        \item Case \textbf{T-Free}

          Trivial.

        \item Case \textbf{T-Alloc}

          By inversion we are either in the \textbf{S-Alloc-Local} or \textbf{S-Alloc-Global} rules. In either case, we assume that $\Gamma, x:^\capa \tau[]^l$ checks $s$, meaning we can use our extended environment lemmas and the \textbf{T-Seq} and \textbf{T-Free} rules to check these cases.

        \item Case \textbf{T-Alloc-Shared}

          Same as previous case.

        \item Case \textbf{T-Partition}

          In this case we have by assumption that $\Gamma, y :^{(h, n/c)} \tau[]^l s$ and $l$ is not local and $c$ divides $n$. By inversion on our step relation we
          must be in the \textbf{S-Partition} case, so we have
          $\eta, \sigma, \Sigma, t, b, p, \Psi, \Phi \vdash^\capa_{\mem} \partition{\psi}{x}{y}{c}{s} \stepsto \code{init}_\psi;  s'; \code{dec}_\psi; \code{wait}_\psi \dashv_{\mem} \eta', \sigma', \Sigma', \Psi, \Phi$ where $s' = s[(y + c \cdot p)/ y]$ and $(\eta', \sigma', \Sigma') = \rename(\eta, \sigma, \Sigma, x, y, \capa/c)$. Via our well-typed renaming lemma contexts we know that we can check the renamed environments in the extended environment $\Gamma, y :^{\capa/c} \tau[]^l, x :^{\capa} \tau[]^l$, and this context is also sufficient to check $s'$ (via a substitution-preserves-typing lemma that is obvious). Using the \textbf{T-Skip} rule this is exactly what we need to show to complete this case, as the thread sync primitives check trivially via their typing rules.

        \item Case \textbf{T-Claim}

          Essentially the same as \textbf{T-Partition}.

        \item Case \textbf{T-Lower}

          Essentially the same as \textbf{T-Partition}.
      \end{itemize}

      \subsubsection*{Control Rules}
      \begin{itemize}
        \item Case \textbf{T-If}

          We have by assumption that $\code{if } e \code{ then } s_1 \code{ else } s_2$ steps to some $s'$, and by inversion we know that either $e$ evaluates to \code{true} and $s'$ is $s_1$, or $e$ evaluates to \code{false} and  $s'$ is $s_2$.

          In either case, our inductive hypotheses is sufficient to tell us that these are well-typed. In particular, in both cases our IHs tell us that the amount of memory used by stepping each branch of the \code{if} is less than the amount of memory computed by the type system for each branch. Because the whole \code{if} expression checks using the greater of the memory usage of $m_1$ or $m_2$ (i.e., the memory usage on each branch), the resulting usage for the whole conditional is also bounded by the type system.

        \item Case \textbf{T-Skip}

          Trivial: \skipcom does not step

        \item  Case \textbf{T-Seq}

          In this case we have that $s_1$ and $s_2$ are both well-typed (with $m_1$ memory and $m_2$ memory respectively), and $m = \textbf{max}(m_1, m_2)$. We also have by inversion that $s_1$ either steps to \skipcom or $s_1'$, and our inductive hypothesis tells us that $s_1'$ is well-typed.

          In the former case we can use the \textbf{S-Seq-Done} rule to trivially finish the case. In the latter, our IH allows us to finish the case, since $m_1$ is always $\leq \textbf{max}(m_1, m_2)$

        \item  Case \textbf{T-While}

          By inversion, we have that
          \begin{align*}
            &\eta, \sigma, \Sigma, t, b, p, \Psi, \Phi \vdash^\capa_{\mem} \code{while } e \code{ do } s \\
            &\stepsto \code{if } e \code{ then } (s; \code{ while } e \code{ do } s) \code{ else } \skipcom \dashv_{\mem} \eta, \sigma, \Sigma, \Psi, \Phi.
          \end{align*}

          We also have by assumption that $e$ and $s$ are well-typed. With this information, through use of the \textbf{T-If}, \textbf{T-Seq}, \textbf{T-While}, and \textbf{T-Skip} rules, we can conclude that the result of this rule is also well-typed.

        \item Case \textbf{T-Function-Call}

          By inversion we have that
          \begin{align*}
            &\eta, \sigma, \Sigma, t, b, p, \Psi, \Phi \vdash^\capa_{\mem} f(e_1, \dots, e_n) \\
            &\stepsto \code{call } s \code { with } (x_i :^{\capa_i} \tau_i \mapsto v_i)@\{\eta, \sigma, \Sigma, \mem'\}  \dashv_{\mem} \eta, \sigma, \Sigma, \Psi, \Phi,
          \end{align*}
          and also that $\Sigma(f)= \{[x_1 : \tau_1, \dots, x_n : \tau_n], s\}$, and $\eta, \sigma, \Sigma, t, b, p \vdash^\capa e_i \Downarrow v_i$, and $m' \leq m$.

          We also have from the premises of our case that $f :^\capa \fun{x_1 :^{\capa_1} \tau_1, \dots, x_n :^{\capa_n} \tau_n}{\capa}{\mem'} \in \Gamma$, and $\Gamma \vdash^{\capa} e_i : \tau_i$, and $\mem' \leq \mem$.

          Our lemma for expression well-typedness tells us that that each $v_i$ is a well-typed value, and our assumption that $\Gamma \vdash \eta, \sigma, \Sigma$ tells us that $\Sigma(f)$ is a well-typed function and thus that $\textbf{fns } \Gamma, x_i :^{\capa_i} \tau_i \vdash^\capa_{m'} s$. We can take the union of this with $\Gamma$ to produce $\Gamma, x_i :^{\capa_i} \tau_i \vdash^\capa_{m'} s$, which clearly checks $s$ and the output environments.
      \end{itemize}
    \end{proof}

%% file: eval/code-appendix.tex
\section{\name{} Programs}\label{sec:code-appendix}

\subsection{Complete TF32 Matrix Multiplication}\label{sec:tf32-func}
\input{eval/3-example}

\subsection{Different Matrix Multiplications Variants}
\input{eval/mm-variants}

\subsection{Single-Pass Parallel Prefix Scan with Decoupled Look-Back}
\input{eval/decoupled-block}

\subsection{CUB Functions}

\input{eval/cub-functions}

\subsection{H100 Matrix Multiplication}
\input{eval/h100-func}

%% file: eval/3-example.tex
\begin{ccudaCode}    
@prism("device")
@requires(thread[32])
def simple_mma(a: ptr(const(float)) @ thread[32], 
                b: ptr(const(float)) @ thread[32],
                c: ptr(float) @ thread[32]):
    
    with group(thread[32]):
    stride : int @ thread[32] = 8
    
    tid : int @ thread[1] = id()
    
    regA: uint32_t[4] @ thread[1]
    row_a : int @ thread[1] = tid / 4
    col_a : int @ thread[1] = tid 
    offsetA : int @ thread[1] = row_a * stride + col_a
    
    regA[0] = __float_as_uint(a[offsetA + 0])
    regA[1] = __float_as_uint(a[offsetA + 8 * stride]); 
    regA[2] = __float_as_uint(a[offsetA + 4]); 
    regA[3] = __float_as_uint(a[offsetA + 8 * stride + 4]); 
    
    regB: uint32_t[2] @ thread[1]
    row_b : int @ thread[1] = tid 
    col_b : int @ thread[1] = tid / 4
    offsetB : int @ thread[1] = row_b * stride + col_b
    
    regB[0] = __float_as_uint(b[offsetB + 0])
    regB[1] = __float_as_uint(b[offsetB + 4 * stride])
    
    regC: uint32_t[4] @ thread[1]
    regC[0] = 0
    regC[1] = 0
    regC[2] = 0
    regC[3] = 0
    
    # Now do MMA!
    intrinsic.mma(
        regA[0], regA[1], regA[2], regA[3],  
        regB[0], regB[1],
        regC[0], regC[1], regC[2], regC[3],
        out=[regC[0], regC[1], regC[2], regC[3]]
    )
    
    row_c : uint32_t @ thread[1] = tid / 4
    col_c : int @ thread[1] = tid 
    offsetC : int @ thread[1] = row_c * stride + 2 * col_c
    
    with partition(c, offset=offsetC, dimension=thread[1]) as c_thread:
        with group(thread[1]):
        c_thread[0] += __uint_as_float(regC[0])
        c_thread[1] += __uint_as_float(regC[1])
        c_thread[8 * stride] += __uint_as_float(regC[2])
        c_thread[8 * stride + 1] += __uint_as_float(regC[3])
    
    return
\end{ccudaCode}

\begin{ccudaCode}
@prism("global")    
@requires(grid[1], block[1], thread[32], smem=1280)
def mmaTF32NaiveKernel(A: ptr(const(float)) @ grid[1],
                B: ptr(const(float)) @ grid[1],
                C: ptr(float) @ grid[1],
                M : int @ grid[1],
                N : int @ grid[1],
                K : int @ grid[1]):
    with group(grid[1]):
    MMA_N  : constexpr(int) = 8
    MMA_K  : constexpr(int) = 8
    MMA_M  : constexpr(int) = 16
    K_tiles : const(int) = (K + MMA_K - 1) / MMA_K
    
    num_blocks_n : const(int) @ grid[1] = (N + MMA_N - 1) / MMA_N
        
    block_row : const(int) @ block[1] = id() / num_blocks_n 
    block_col : const(int) @ block[1] = id() 
    
    warp_row: const(int) @ block[1] = block_row * MMA_M 
    warp_col : const(int) @ block[1] = block_col * MMA_N
    
    with claim(C, scope=block[1], offset= warp_row * N + warp_col) as C_blk:
        with group(block[1]):
        A_smem : shared(float[16 * 8]) @ block[1]
        B_smem : shared(float[8 * 8])  @ block[1]
        C_smem : shared(float[16 * 8]) @ block[1]
    
        idx : int @ thread[1] = id() * 4
        with partition(C_smem, dimension=thread[1], offset=idx) as C_thrd:
            for i in range(0, 4, 1):
                with group(thread[1]):
                C_thrd[i] = 0
         
        for i in range(0, K_tiles, 1):
            a_idx : int @ thread[1] = id() * 4
            # Sync point
            for j in range(0, 4, 1):
            flat_idx : int @ thread[1] = a_idx + j
            row : int @ thread[1] = flat_idx / MMA_K 
            col : int @ thread[1] = flat_idx 
            global_row : int @ thread[1] = warp_row + row 
            global_col: int @ thread[1] = i * MMA_K + col 
            with partition(A_smem, dimension=thread[1], offset= row * MMA_K + col) as A_smem_thrd:
                with group(thread[1]):
                A_smem_thrd[0] = A[global_row * K + global_col]
                
            b_idx : int @ thread[1] = id() * 2
            # Sync point
            for j in range(0, 2, 1):
            flat_idx_b : int @ thread[1] = b_idx + j
            row_b : int @ thread[1] = flat_idx_b / MMA_K 
            col_b : int @ thread[1] = flat_idx_b 
            global_row_b : int @ thread[1] = i * MMA_K + row_b 
            global_col_b : int @ thread[1] = warp_col + col_b
            with partition(B_smem, dimension=thread[1], offset= row_b * MMA_K + col_b) as B_smem_thrd:
                with group(thread[1]):
                B_smem_thrd[0] = B[global_row_b * N + global_col_b]
            
            # Sync point
            with claim(C_smem, scope=thread[32], offset=0) as C_smem_warp:
                match split(thread):
                case 32:
                    simple_mma(A_smem, B_smem, C_smem_warp)
        
        for j in range(0, 4, 1):
            flat_idx_c : int @ thread[1] = id() * 4 + j
            row_c : int @ thread[1] = flat_idx_c / MMA_K 
            col_c : int @ thread[1] = flat_idx_c 
            with partition(C_blk, dimension=thread[1], offset= row_c * N + col_c) as C_thrd:
            with group(thread[1]):
                C_thrd[0]  = C_smem[row_c * MMA_N + col_c]
return
\end{ccudaCode}

%% file: eval/mm-variants.tex
\begin{ccudaBox}
@prism("global")
@requires(grid[1], block[1], thread[1], smem=49000)
def my_sgemm_kernel_1( M : int @ grid[1],
                       N : int @ grid[1],
                       K : int @ grid[1],
                       alpha : float @ grid[1],
                       A: ptr(const(float)) @ grid[1],
                       B: ptr(const(float)) @ grid[1],
                       beta : float @ grid[1],
                       C: ptr(float) @ grid[1],
                       block_size : int @ grid[1]):
    with group(grid[1]):
        bid : int @ block[1] = id()
        with partition(C, dimension=block[1], offset=0) as C_blk:
            with group(block[1]):
                tid : int @ thread[1] = bid * block_size + id()
                with partition(C_blk, dimension=thread[1], offset=0) as C_thrd:
                    x : const(int) @ thread[1] = tid / N
                    y : const(int) @ thread[1] = tid 
                    with group(thread[1]):
                        tmp : float @ thread[1] = 0
                        for i in range(0, K, 1):
                            tmp += A[x * K + i] * B[i * N + y]
                        C_thrd[x * N + y] = alpha * tmp + beta * C[x * N + y]
    return
\end{ccudaBox}

\begin{ccudaBox}
@prism("global")
@requires(grid[1], block[1], thread[1], smem=49000)
def my_sgemm_kernel_2( M : int @ grid[1],
                        N : int @ grid[1],
                        K : int @ grid[1],
                        alpha : float @ grid[1],
                        A: ptr(const(float)) @ grid[1],
                        B: ptr(const(float)) @ grid[1],
                        beta : float @ grid[1],
                        C: ptr(float) @ grid[1],
                        tile_size : int @ grid[1]):
    with group(grid[1]):
        global_tid: int @ thread[1] = id()
        total_elements : const(int) @ grid[1] = M * N
        num_tiles_per_row :const(int) @ grid[1] = (N + tile_size - 1) / tile_size
        tile_id : const(int) @ thread[1] = global_tid / (tile_size * tile_size)
        tile_row : const(int) @ thread[1] = tile_id / num_tiles_per_row
        tile_col : const(int) @ thread[1] = tile_id 
        local_id : const(int) @ thread[1] = global_tid 
        local_row : const(int) @ thread[1] = local_id / tile_size 
        local_col : const(int) @ thread[1] = local_id 
        cRow : const(int) @ thread[1] = tile_row * tile_size + local_row
        cCol : const(int) @ thread[1] = tile_col * tile_size + local_col
        with partition(C, dimension=thread[1], offset=cRow * N + cCol) as C_thrd:
            with group(thread[1]):
                if global_tid < total_elements: 
                    if cRow < M and cCol < N:
                        tmp : float @ thread[1] = 0  
                        for i in range(0, K, 1):
                            tmp += A[cRow * K + i] * B[i * N + cCol]
                        C_thrd[0] = alpha * tmp + beta * C_thrd[0]

    return
\end{ccudaBox}

\begin{ccudaBox}    
@prism("global")
@attr("__launch_bounds__((BM * BN) / (TM * TN), 1)")
@template([("BM", c_int),
            ("BN", c_int),
            ("BK", c_int),
            ("TM", c_int),
            ("TN", c_int)])
@requires(grid[1], block[1], thread[1], smem=49000)
def my_sgemm_kernel_3( M : int @ grid[1],
                        N : int @ grid[1],
                        K : int @ grid[1],
                        alpha : float @ grid[1],
                        A: ptr(const(float)) @ grid[1],
                        B: ptr(const(float)) @ grid[1],
                        beta : float @ grid[1],
                        C: ptr(float) @ grid[1],
                        num_blocks_N : int @ grid[1]):
    with group(grid[1]):
        cRow : const(int) @ block[1] = id() / num_blocks_N
        cCol : const(int) @ block[1] = id() 
        totalResultsBlocktile: const(uint) @ grid[1] = BM * BN
        numThreadsBlocktile: const(uint) @ grid[1] = totalResultsBlocktile / (TM * TN)
        with partition(C, dimension=block[1], offset = cRow * BM * N + cCol * BN) as C_blk:
            with partition(A, dimension=block[1], offset = cRow * BM * K) as A_blk:
                with partition(B, dimension=block[1], offset = cCol * BN) as B_blk:
                    with group(block[1]):
                        As : shared(float[128 * 8]) @ block[1]  
                        Bs : shared(float[128 * 8]) @ block[1]  
                        threadRow : const(int) @ thread[1] = id() / (BN / TN) 
                        threadCol : const(int) @ thread[1] = id() 
    
                        innerRowA : const(int) @ thread[1] = id() / BK 
                        innerColA : const(int) @ thread[1] = id() 
                        strideA : const(int) @ thread[1] = numThreadsBlocktile / BK 
    
                        innerRowB : const(int) @ thread[1] = id() / BN 
                        innerColB : const(int) @ thread[1] = id() 
                        strideB : const(int) @ thread[1] = numThreadsBlocktile / BN
    
                        threadResults : float[64] @ thread[1] = { 0 }
                        regM : float[8] @ thread[1] = { 0 } # Need to go add the {}
                        regN : float[8] @ thread[1] = { 0 }
    
                        i : int @ thread[1] = 0
    
                        for bkIdx in range(0, K, BK):
    
                            with partition(As, dimension=thread[1], offset = 0) as As_thrd:
                                with partition(Bs, dimension=thread[1], offset = 0) as Bs_thrd:
                                    with group(thread[1]):
                                        for loadOffset in range(0, BM, strideA):
                                            As_thrd[(innerRowA + loadOffset) * BK + innerColA] = \
                                                A_blk[ i * BK + (innerRowA + loadOffset) * K + innerColA]
    
                                        for loadOffset in range(0, BK, strideB):
                                            Bs_thrd[(innerRowB + loadOffset) * BN + innerColB] = \ 
                                                B_blk[(i * BK * N) + (innerRowB + loadOffset) * N + innerColB]
    
                                        i += 1
    
                                        for dotIdx in range(0, BK, 1):
                                            for i in range(0, TN, 1):
                                                regM[i] = As[(threadRow * TM + i) * BK + dotIdx]
                                            for i in range(0, TN, 1):
                                                regN[i] = Bs[dotIdx * BN + threadCol * TN + i]
                                            for resIdxM in range(0, TM, 1):
                                                for resIdxN in range(0, TN, 1):
                                                        threadResults[resIdxM * TN + resIdxN] += regM[resIdxM] * regN[resIdxN]
    
                        for resIdxM in range(0, TM, 1):   
                            for resIdxN in range(0, TN, 1):   
                                with partition(C_blk, dimension=thread[1], offset = (threadRow * TM + resIdxM) * N + threadCol * TN + resIdxN) as C_thrd:
                                    with group(thread[1]):
                                        C_thrd[0] = alpha * threadResults[resIdxM * TN + resIdxN] + beta * C_blk[(threadRow * TM + resIdxM) * N + threadCol * TN + resIdxN]
                            
    return
\end{ccudaBox}

\begin{ccudaBox}
@prism("global")
@attr("__launch_bounds__((BM * BN) / (TM * TN), 1)")
@template([("BM", c_int),
           ("BN", c_int),
           ("BK", c_int),
           ("TM", c_int),
           ("TN", c_int)])
@requires(grid[1], block[1], thread[1], smem=49000)
def my_sgemm_kernel_4( M : int @ grid[1],
                       N : int @ grid[1],
                       K : int @ grid[1],
                       alpha : float @ grid[1],
                       A: ptr(float) @ grid[1],
                       B: ptr(float) @ grid[1],
                       beta : float @ grid[1],
                       C: ptr(float) @ grid[1],
                       num_blocks_N : int @ grid[1]):
    with group(grid[1]):
        cRow : const(int) @ block[1] = id() / num_blocks_N
        cCol : const(int) @ block[1] = id() 
        totalResultsBlocktile: const(uint) @ grid[1] = BM * BN
        numThreadsBlocktile: const(uint) @ grid[1] = totalResultsBlocktile / (TM * TN)
        with partition(C, dimension=block[1], offset = cRow * BM * N + cCol * BN) as C_blk:
            with partition(A, dimension=block[1], offset = cRow * BM * K) as A_blk:
                with partition(B, dimension=block[1], offset = cCol * BN) as B_blk:
                    with group(block[1]):
                        As : shared(float[128 * 8]) @ block[1]  
                        Bs : shared(float[128 * 8]) @ block[1]  
                        threadRow : const(int) @ thread[1] = id() / (BN / TN) 
                        threadCol : const(int) @ thread[1] = id() 

                        innerRowA : const(int) @ thread[1] = id() / (BK / 4)
                        innerColA : const(int) @ thread[1] = id() 
                        strideA : const(int) @ thread[1] = numThreadsBlocktile / BK 

                        innerRowB : const(int) @ thread[1] = id() / (BN / 4)
                        innerColB : const(int) @ thread[1] = id() 
                        strideB : const(int) @ thread[1] = numThreadsBlocktile / BN

                        threadResults : float[64] @ thread[1] = { 0 }
                        regM : float[8] @ thread[1] = { 0 } # Need to go add the {}
                        regN : float[8] @ thread[1] = { 0 }

                        i : int @ thread[1] = 0

                        for bkIdx in range(0, K, BK):

                            with partition(As, dimension=thread[1], offset = 0) as As_thrd:
                                with partition(Bs, dimension=thread[1], offset = 0) as Bs_thrd:
                                    with group(thread[1]):

                                        tmp : float4 @ thread[1] = float4_cast(A_blk[i * BK  + innerRowA * K + innerColA * 4])
                                        As_thrd[(innerColA * 4 + 0) * BM + innerRowA] = tmp.x
                                        As_thrd[(innerColA * 4 + 1) * BM + innerRowA] = tmp.y
                                        As_thrd[(innerColA * 4 + 2) * BM + innerRowA] = tmp.z
                                        As_thrd[(innerColA * 4 + 3) * BM + innerRowA] = tmp.w

                                        i += 1

                                        for dotIdx in range(0, BK, 1):
                                            for i in range(0, TN, 1):
                                                regM[i] = As[dotIdx * BM + threadRow * TM + i]
                                            for i in range(0, TN, 1):
                                                regN[i] = Bs[dotIdx * BN + threadCol * TN + i]
                                            for resIdxM in range(0, TM, 1):
                                                for resIdxN in range(0, TN, 1):
                                                    threadResults[resIdxM * TN + resIdxN] += regM[resIdxM] * regN[resIdxN]

                        for resIdxM in range(0, TM, 1):   
                            for resIdxN in range(0, TN, 4):   
                                with partition(C_blk, dimension=thread[1], offset = (threadRow * TM + resIdxM) * N + threadCol * TN + resIdxN) as C_thrd:
                                    with group(thread[1]):
                                        temp : float4 =  float4_cast(C_thrd[0])
                                        temp.x = alpha * threadResults[resIdxM * TN + resIdxN] + beta * temp.x
                                        temp.y = alpha * threadResults[resIdxM * TN + resIdxN + 1] + beta * temp.y
                                        temp.z = alpha * threadResults[resIdxM * TN + resIdxN + 2] + beta * temp.z
                                        temp.w = alpha * threadResults[resIdxM * TN + resIdxN + 3] + beta * temp.w
                                        float4_cast[C_thrd[0]] = temp
                        
    return
\end{ccudaBox}

\begin{ccudaBox}
@prism("global")
@attr("__launch_bounds__(NUM_THREADS)")
@template([("BM", c_int),
    ("BN", c_int),
    ("BK", c_int),
    ("WM", c_int),
    ("WN", c_int),
    ("WNITER", c_int),
    ("TM", c_int),
    ("TN", c_int),
    ("NUM_THREADS", c_int),
    ])    
@requires(grid[1], block[1], warp[1], smem=49000)
def my_sgemm_kernel_5( M : int @ grid[1],
                        N : int @ grid[1],
                        K : int @ grid[1],
                        alpha : float @ grid[1],
                        A: ptr(float) @ grid[1],
                        B: ptr(float) @ grid[1],
                        beta : float @ grid[1],
                        C: ptr(float) @ grid[1],
                        num_blocks_N : int @ grid[1]):
    
    with group(grid[1]):
        cRow : const(int) @ block[1] = id() / num_blocks_N
        cCol : const(int) @ block[1] = id() 
    
        WMITER : constexpr(uint) @ grid[1] = (WM * WN) / (WARPSIZE * TM * TN * WNITER)
        WSUBM : constexpr(uint) @ grid[1] = WM / WMITER
        WSUBN : constexpr(uint) @ grid[1] = WN / WNITER
    
        rowStrideA : constexpr(int) @ grid[1] = (NUM_THREADS * 4) / BK 
        rowStrideB : constexpr(int) @ grid[1] = NUM_THREADS / (BN / 4) 
        with partition(C, dimension=block[1], offset = cRow * BM * N + cCol * BN) as C_blk:
            with partition(A, dimension=block[1], offset = cRow * BM * K) as A_blk:
                with partition(B, dimension=block[1], offset = cCol * BN) as B_blk:
                    with group(block[1]):
                        warpIdx : const(uint) @ thread[32] = id()
                        warpCol : const(uint) @ thread[32] =  warpIdx 
                        warpRow : const(uint) @ thread[32] =  warpIdx / (BN / WN)
                        As : shared(float[128 * 16]) @ block[1]  
                        Bs : shared(float[128 * 16]) @ block[1]  
    
                        innerRowA : const(int) @ thread[1] = id() / (BK / 4)
                        innerColA : const(int) @ thread[1] = id() 
    
                        innerRowB : const(int) @ thread[1] = id() / (BN / 4)
                        innerColB : const(int) @ thread[1] = id() 
    
                        threadResults : float[128] @ thread[1] = { 0 }
                        regM : float[8] @ thread[1] = { 0 } # Need to go add the {}
                        regN : float[16] @ thread[1] = { 0 }
                        j : int @ block[1] = 0
                        for bkIdx in range(0, K, BK):
                            with partition(A_blk, dimension=block[1], offset = j * BK) as A_offset:
                                with partition(B_blk, dimension=block[1], offset = j * BK * N) as B_offset: 
                                    loadFromGmem_our(N, K, A_offset, B_offset, As, Bs, template=[BM, BN, BK, rowStrideA, rowStrideB])
                                    processFromSmem_our(regM, regN, threadResults, As, Bs, warpRow, warpCol, template=[BM, BN, BK, WM, WN, WMITER, WNITER, WSUBM, WSUBN, TM, TN])
                                    j+=1
                                    pass
    
                        with partition(C_blk, dimension=warp[1], offset = warpRow * WM * N + warpCol * WN) as C_warp:
                            with group(warp[1]):
    
                                threadIdxInWarp : int @ thread[1] = id()
                                threadColInWarp : const(uint) @ thread[1] = threadIdxInWarp 
                                threadRowInWarp : const(uint) @ thread[1] = threadIdxInWarp / (WSUBN / TN)
    
                                for wSubRowIdx in range(0, WMITER, 1):
                                    for wSubColIdx in range(0, WNITER, 1):
                                        with partition(C_warp, dimension=thread[1], offset = (wSubRowIdx * WSUBM) * N + wSubColIdx * WSUBN) as C_interim:
                                            for resIdxM in range(0, TM, 1):
                                                for resIdxN in range(0, TN, 4):
                                                    with group(thread[1]):
                                                        tmp : float4 @ thread[1] = float4_cast(C_interim[(threadRowInWarp * TM + resIdxM) * N + threadColInWarp * TN + resIdxN])
                                                        i : const(int) @ thread[1] = (wSubRowIdx * TM + resIdxM) * (WNITER * TN) + wSubColIdx * TN + resIdxN
                                                        tmp.x = alpha * threadResults[i + 0] + beta * tmp.x
                                                        tmp.y = alpha * threadResults[i + 1] + beta * tmp.y
                                                        tmp.z = alpha * threadResults[i + 2] + beta * tmp.z
                                                        tmp.w = alpha * threadResults[i + 3] + beta * tmp.w
                                                        float4_cast[C_interim[(threadRowInWarp * TM + resIdxM) * N  +  threadColInWarp * TN + resIdxN]] = tmp
                                
    return 
\end{ccudaBox}

%% file: eval/decoupled-block.tex
\begin{ccudaCode}    
@prism("global")
@requires(grid[1], thread[1])
def add_m_thread(x: ptr(const(c_int)) @ grid[1], 
                    y: ptr(c_int) @ grid[1],
                    m: uint32_t @ grid[1]):
    with group(grid[1]):
        tid : uint32_t @ thread[1] = id()
        with partition(x, offset=tid, dimension=thread[1]) as x_thread:
            with partition(y, offset=tid, dimension=thread[1]) as y_thread:
                with group(thread[1]):
                    y_thread[0] = x_thread[0] + m
    return

# WORK_PER_BLOCK = (8192 * 2)
# WARPS_PER_BLOCK = 16
# WARPS_PASS_CNT = (8192 * 2) / (16 * 32) = 32
# WARP_ENDS = (8192 * 2) / 32 = 512
    
@prism("device")
@requires(thread[1])
def shfl_up_sync_wrapper(val : uint32_t @ thread[1], delta : uint32_t @ thread[1]):
    pass
    
@prism("device")
@requires(block[1], warp[16])
def scan_warps(x : ptr(uint32_t) @ block[1], x_shmem : ptr(uint32_t) @ block[1]):
    with group(block[1]): 
        for i in range(0, (8192 * 2) / (16 * 32)):
            warp_id : uint32_t @ warp[1] = id()
            with partition(x, offset = (warp_id + i * 16) * 32, dimension = warp[1]) as x_warp:
                with partition(x_shmem, offset = (warp_id + i * 16) * 32, dimension = warp[1]) as x_shmem_warp:
                    with group(warp[1]):
                        thread_id : uint32_t @ thread[1] = id()
                        with partition(x_warp, offset = thread_id, dimension = thread[1]) as x_thread:
                            with partition(x_shmem_warp, offset = thread_id, dimension = thread[1]) as x_shmem_thread:
                                with group(thread[1]):
                                    val : uint32_t @ thread[1] = x_thread[0]
                                        
                                    stride : uint32_t @ thread[1] = 1
                                    while stride < 32:
                                        received : uint32_t @ thread[1] = shfl_up_sync_wrapper(val, stride)
                                        if thread_id >= stride:
                                            val = received + val
                                        stride *= 2
                                    x_shmem_thread[0] = val
    
@prism("device")
@requires(block[1], warp[16])
def get_warp_ends_inclusive(x : ptr(uint32_t) @ block[1], warp_ends : ptr(uint32_t) @ block[1]):
    with group(block[1]): 
        warp_id : uint32_t @ warp[1] = id()
        with partition(x, offset = (warp_id) * 32 * 32, dimension = warp[1]) as x_warp:
            with partition(warp_ends, offset = (warp_id) * 32, dimension = warp[1]) as warp_ends_warp:
                with group(warp[1]):
                    thread_id : uint32_t @ thread[1] = id()
                    with partition(x_warp, offset = thread_id * 32, dimension = thread[1]) as x_thread:
                        with partition(warp_ends_warp, offset = thread_id, dimension = thread[1]) as warp_ends_thread:
                            with group(thread[1]):
                                warp_ends_thread[0] = x_thread[31]

@prism("device")
@requires(block[1], warp[16])
def update_with_new_data_and_writeback(x : ptr(uint32_t) @ block[1], x_shmem : ptr(uint32_t) @ block[1], warp_ends : ptr(uint32_t) @ block[1], prefix : uint32_t @ block[1]):
    with group(block[1]):
        for i in range(0, (8192 * 2) / (16 * 32)):
            warp_id : uint32_t @ warp[1] = id()
            with partition(x, offset = (warp_id + i * 16) * 32, dimension = warp[1]) as x_warp:
                with partition(x_shmem, offset = (warp_id + i * 16) * 32, dimension = warp[1]) as x_shmem_warp:
                    with group(warp[1]):
                        thread_id : uint32_t @ thread[1] = id()
                        with partition(x_warp, offset = thread_id, dimension = thread[1]) as x_thread:
                            with partition(x_shmem_warp, offset = thread_id, dimension = thread[1]) as x_shmem_thread:
                                with group(thread[1]):
                                    if warp_id + i * 16 == 0:
                                        x_thread[0] = prefix + x_shmem_thread[0]
                                    else:
                                        x_thread[0] = prefix + warp_ends[warp_id + i * 16 - 1] + x_shmem_thread[0]
    
@prism("device")
@requires(block[1], warp[16])
def naive_scan_warp_ends(x : ptr(uint32_t) @ block[1], x_buf : ptr(uint32_t)  @ block[1]):
    with group(block[1]):
        stride : uint32_t @ group[1] = 1
        while stride < (8192 * 2 / 32):
    
            warp_id : uint32_t @ warp[1] = id()
            with partition(x, offset = warp_id * 32, dimension = warp[1]) as x_warp:
                with partition(x_buf, offset = warp_id * 32, dimension = warp[1]) as x_buf_warp:
                    with group(warp[1]):
                        thread_id : uint32_t @ thread[1] = id()
                        with partition(x_warp, offset = thread_id, dimension = thread[1]) as x_thread:
                            with partition(x_buf_warp, offset = thread_id, dimension = thread[1]) as x_buf_thread:
                                with group(thread[1]):
                                    if thread_id + warp_id * 32  >= stride:
                                        x_buf_thread[0] = x_thread[0] + x_thread[-stride]
                                    else:
                                        x_buf_thread[0] = x_thread[0]
    
            tmp : ptr(uint32_t) @ block[1] = x
            x = x_buf
            x_buf = tmp
    
            stride *= 2
    return x

# Implemented in raw CUDA: only exposing the necessary interfaces
@prism("device")
@requires(block[1], thread[1])
def write_block_aggregate(aggregate : uint32_t @ block[1], block_idx : uint32_t @ block[1], block_infos : ptr(int) @ block[1]):
    pass
    
@prism("device")
@requires(block[1], thread[1])
def write_block_prefix(prefix : uint32_t @ block[1], block_idx : uint32_t @ block[1], block_infos : ptr(int) @ block[1]):
    pass

@prism("device")
@requires(block[1], thread[1])
def lookback_to_prev_blocks(block_idx : uint32_t @ block[1], block_infos : ptr(int) @ block[1]):
    pass
    
@prism("device")
@requires(grid[1])
def get_block_idx_from_global_counter(block_idx_counter : ptr(uint32_t) @ grid[1]):
    pass
    
@prism("global")
@requires(grid[1], block[1], warp[16], smem=69632)
def scan_kernel(n : uint32_t @ grid[1], x : ptr(uint32_t) @ grid[1], block_idx_counter : ptr(uint32_t) @ grid[1], block_infos : ptr(int) @ grid[1]):
    block_idx : uint32_t @ grid[1] = get_block_idx_from_global_counter(block_idx_counter)
    
    with partition(x, offset = (8192 * 2) * block_idx, dimension = block[1]) as x_block:
        with group(block[1]):
            x_shmem: shared(uint32_t[8192 * 2]) @ block[1]
                
            scan_warps(x_block, x_shmem)
    
            warp_ends : shared(uint32_t[512]) @ block[1]
            warp_ends_buf : shared(uint32_t[512]) @ block[1]
                
            get_warp_ends_inclusive(x_shmem, warp_ends)
            warp_ends_res : ptr(uint32_t) @ block[1] = naive_scan_warp_ends(warp_ends, warp_ends_buf)
    
            aggregate : uint32_t @ block[1] = warp_ends_res[512 - 1]
            write_block_aggregate(aggregate, block_idx, block_infos)
    
            prefix : uint32_t @ block[1] = lookback_to_prev_blocks(block_idx, block_infos)
            write_block_prefix(aggregate + prefix, block_idx, block_infos)
                
            update_with_new_data_and_writeback(x_block, x_shmem, warp_ends_res, prefix)
\end{ccudaCode}

%% file: eval/cub-functions
\begin{ccudaBox}
@prism("device")
@requires(block[1], thread[32])
def block_load(input : ptr(const(int)) @ block[1],
                output : ptr(int) @ thread[1],
                items_per_thread : int @ block[1]):
    with group(block[1]):
        tid : int @ thread[32]  = id()
        with partition(input, dimension=thread[32], offset = tid * items_per_thread) as input_thrd:
            with group(thread[32]):
                warp_load(input_thrd, output, items_per_thread)
                   
    return  
\end{ccudaBox}

\begin{ccudaBox}
@prism("device")
@requires(thread[32])
def warp_load(input : ptr(const(int)) @ thread[32],
                output : ptr(int) @ thread[1],
                items_per_thread : int @ thread[32]):
    with group(thread[32]):
        tid : int @ thread[1]  = id()
        with partition(input, dimension=thread[1], offset = tid * items_per_thread) as input_thrd:
            with group(thread[1]):
                thread_load(input_thrd, output, items_per_thread)
                       
    return  
\end{ccudaBox}

\begin{ccudaBox}
@prism("device")
@requires(thread[1])
def thread_load(input : ptr(const(int)) @ thread[1],
                output : ptr(int) @ thread[1],
                items_per_thread : int @ thread[1]):
    with group(thread[1]):
        for i in range(0, items_per_thread, 1):
            output[i] = input[i]
    return   
\end{ccudaBox}

\begin{ccudaBox}
@prism("device")
@requires(block[1], thread[32])
def block_store(input : ptr(const(int)) @ thread[1],
                output : ptr(int) @ block[1],
                items_per_thread : int @ block[1]):
    with group(block[1]):
        tid : int @ thread[32]  = id()
        with partition(output, dimension=thread[32], offset = tid * items_per_thread) as output_thrd:
            with group(thread[32]):
                warp_store(input, output_thrd, items_per_thread)
    return   
\end{ccudaBox}

\begin{ccudaBox}
@prism("device")
@requires(thread[32])
def warp_store(input : ptr(const(int)) @ thread[1],
                output : ptr(int) @ thread[32],
                items_per_thread : int @ thread[32]):
    with group(thread[32]):
        tid : int @ thread[1]  = id()
        with partition(output, dimension=thread[1], offset = tid * items_per_thread) as output_thrd:
            with group(thread[1]):
                thread_store(input, output_thrd, items_per_thread)
    return  
\end{ccudaBox}

\begin{ccudaBox}
@prism("device")
@requires(thread[1])
def thread_store(input : ptr(const(int)) @ thread[1],
               output : ptr(int) @ thread[1],
               items_per_thread : int @ thread[1]):
    with group(thread[1]):
        for i in range(0, items_per_thread, 1):
            output[i] = input[i]
    return   
\end{ccudaBox}

%% file: eval/h100-func.tex
\begin{minted}[
    linenos,
    autogobble,
    breaklines,
    breakanywhere,
    baselinestretch=1.0,
    fontsize=\tiny,
    numbersep=3pt,
    frame=single
  ]{python}
@prism("global")
@attr(" __launch_bounds__(128*3) ")
@requires(grid[1], block[1], warp[12], smem=227000)
def my_h100_match_no_tail(
            M: int @ grid[1],
            N: int @ grid[1],
            K: int @ grid[1],
            C: ptr(bf16) @ grid[1],
            tensorMapA: const(CUtensorMap) @ grid[1],
            tensorMapB: const(CUtensorMap) @ grid[1]):
        
    with group(grid[1]):
        BM : constexpr(int) @ grid[1] = 128
        BN : constexpr(int) @ grid[1] = 256
        BK : constexpr(int) @ grid[1] = 64
        NUM_THREADS : constexpr(int) @ grid[1] = 128*3
        QSIZE : constexpr(int) @ grid[1] = 4
        NUM_SM : constexpr(int) @ grid[1] = 128
    
        WGMMA_M : constexpr(int) @ grid[1] = 64
        WGMMA_K : constexpr(int) @ grid[1] = 16
        WGMMA_N : constexpr(int) @ grid[1] = BN
        num_consumers : constexpr(int) @ grid[1] = (NUM_THREADS / 128) - 1
        B_WG_M : constexpr(int) @ grid[1] = BM / num_consumers
    
        TM : constexpr(int) @ grid[1] = 16
        TN : constexpr(int) @ grid[1] = 8
        schedule_block : int @ block[1] = id()
            
        with partition(C, offset=0, dimension=block[1]) as block_C:                 
            with group(block[1]):
    
                # Allocate SA
                #  sA[QSIZE][BK*BM]
                sA_slot0 : shared(bf16[64* 128], align=128) @ block[1]
                sA_slot1 : shared(bf16[64* 128]) @ block[1]
                sA_slot2 : shared(bf16[64* 128]) @ block[1]
                sA_slot3 : shared(bf16[64* 128]) @ block[1]
    
                # Allocate SB
                #  sB[QSIZE][BK*BN]
                sB_slot0 : shared(bf16[64* 256], align=128) @ block[1]
                sB_slot1 : shared(bf16[64* 256]) @ block[1]
                sB_slot2 : shared(bf16[64* 256]) @ block[1]
                sB_slot3 : shared(bf16[64* 256]) @ block[1]

                num_blocks_k : const(int) @ block[1] = K / BK
                wg_idx : int @ warp[4] = id()
                blk_thrd_id : int @ thread[1] = id()

                tid : int @ thread[1] = id() % 128
                is_producer : bool @ warp[4] = wg_idx == 0 
    
                num_block_m : int @ block[1] = 0
                num_block_n : int @ block[1] = 0

                with group(warp[4]):
                    if (is_producer):
                        pass
                    else:
                        wg_idx = wg_idx - 1
    
                schedule_it : int @ block[1] = 0
                total_blocks_m : int @ block[1] = (((M + BM) - 1) / BM)
                total_blocks_n : int @ block[1] = (((N + BN) - 1) / BN)
                unsafe("assert(CEIL_DIV(M, BM)%TM == 0 && total_blocks_n%TN == 0);")
    
                while(true):
                    num : int @ block[1] = ((schedule_it * NUM_SM) + schedule_block)
                    if (num >= (total_blocks_m * total_blocks_n)):
                        break
                    cur_tile : int @ block[1] = (num / (TM * TN))
                    cur_tile_pos : int @ block[1] = (num % (TM * TN))
                    num_block_m = (TM * (cur_tile / (total_blocks_n / TN)))
                    num_block_n = (TN * (cur_tile % (total_blocks_n / TN)))
                    num_block_m += (cur_tile_pos / TN)
                    num_block_n += (cur_tile_pos % TN)
                    schedule_it += 1
                    d : c_float[1][16][8] @ thread[1]
    
                    with group(thread[1]):
                        d_ptr : ptr(float) @ thread[1]
                        # TODO: ANNOYONG TO NOT HAVE TO SET EACH INDIVIDUAL TO 0.
                        unsafe("d_ptr = (float *)d;")
                        memset_wrapper(d_ptr, 0, 1 * 16 * 8 * 4)
    
                    # Now remove these registers...
                    with group(warp[4]):
                        if (is_producer):
                            num_regs : constexpr(int) @ warp[4] =  24 if (num_consumers <= 2) else 32
                            warpgroup_reg_dealloc(template=[num_regs])
                            pass
                        else:
                            num_regs_p : constexpr(int) @ warp[4] =  256 if num_consumers == 1 else (240 if num_consumers == 2 else 160)
                            warpgroup_reg_alloc(template=[num_regs_p])
                            pass
                    
                    num_blocks_k_main : const(int) @ block[1] = (num_blocks_k / 4) * 4
    
                    for block_k_iter in range(0, num_blocks_k_main, 4):
    
                        with claim(sA_slot0, scope=thread[1], offset=0) as sA_new:
                            with claim(sA_slot1, scope=thread[1], offset=0) as sA_new1:
                                with claim(sA_slot2, scope=thread[1], offset=0) as sA_new2:
                                    with claim(sB_slot0, scope=thread[1], offset=0) as sB_new:
                                        with claim(sB_slot1, scope=thread[1], offset=0) as sB_new1:
                                            with claim(sB_slot2, scope=thread[1], offset=0) as sB_new2:
                                                with claim(sA_slot3, scope=thread[1], offset=0) as sA_new3:
                                                    with claim(sB_slot3, scope=thread[1], offset=0) as sB_new3:
                                                        match warp:
                                                            case 8:
                                                                pass
                                                            case 4:
                                                                match thread:
                                                                    case 1:
                                                                        expect_bytes(full_0, (BK*BN+BK*BM)*2)
                                                                        # TODO: BECOMES AN INTRINSIC CALL!
                                                                        ta_void_ptr : ptr(const(void)) @ thread[1]
                                                                        unsafe("ta_void_ptr = &tensorMapA;")
                                                                        load_async(sA_new, ta_void_ptr, full_0, (block_k_iter+0)*BK, num_block_m*BM)
                                                                        tb_void_ptr : ptr(const(void)) @ thread[1]
                                                                        unsafe("tb_void_ptr = &tensorMapB;")
                                                                        load_async(sB_new, tb_void_ptr, full_0, (block_k_iter+0)*BK, num_block_n*BN)

                                                                        expect_bytes(full_1, (BK*BN+BK*BM)*2)
                                                                        load_async(sA_new1, ta_void_ptr, full_1, (block_k_iter+1)*BK, num_block_m*BM)
                                                                        load_async(sB_new1, tb_void_ptr, full_1, (block_k_iter+1)*BK, num_block_n*BN)
    
                                                                        expect_bytes(full_2, (BK*BN+BK*BM)*2)
                                                                        load_async(sA_new2, ta_void_ptr, full_2, (block_k_iter+2)*BK, num_block_m*BM)
                                                                        load_async(sB_new2, tb_void_ptr, full_2, (block_k_iter+2)*BK, num_block_n*BN)
    
                                                                        expect_bytes(full_3, (BK*BN+BK*BM)*2)
                                                                        load_async(sA_new3, ta_void_ptr, full_3, (block_k_iter+3)*BK, num_block_m*BM)
                                                                        load_async(sB_new3, tb_void_ptr, full_3, (block_k_iter+3)*BK, num_block_n*BN)

                        with claim(sA_slot0, scope=warp[8], offset=0) as wgmma_sA:
                            with claim(sA_slot1, scope=warp[8], offset=0) as sA_1_producer:
                                with claim(sA_slot2, scope=warp[8], offset=0) as sA_2_producer:
                                    with claim(sB_slot0, scope=warp[8], offset=0) as wgmma_sB:
                                        with claim(sB_slot1, scope=warp[8], offset=0) as sB_1_producer:
                                            with claim(sB_slot2, scope=warp[8], offset=0) as sB_2_producer:
                                                with claim(sA_slot3, scope=warp[8], offset=0) as sA_3_producer:
                                                    with claim(sB_slot3, scope=warp[8], offset=0) as sB_3_producer:
                                                        match warp:
                                                            case 8: 
                                                                with group(warp[4]):
                                                                    with wgmma_async():
                                                                        for m_it in range(0, B_WG_M/WGMMA_M, 1):
                                                                            index_a : int @ warp[4] = m_it * BM
                                                                            index_b : int @ warp[4] = m_it * BN
                                                                            unsafe("#pragma unroll")
                                                                            for k_it in range(0, 64/WGMMA_K, 1):
                                                                                intrinsic.wgmma256(wgmma_sA[64*(m_it + wg_idx*B_WG_M/WGMMA_M)*WGMMA_M + index_a + k_it*WGMMA_K], wgmma_sB[(index_b + (k_it * WGMMA_K))], 1, 1, 1, 0, 0, out=[d[m_it]])
                                                          
                                                                    with wgmma_async():
                                                                        for m_it in range(0, B_WG_M/WGMMA_M, 1):
                                                                            index_a2 : int @ warp[4] = m_it * BM
                                                                            # TOTALLY OKAY! I just don't want to add support for this the compiler right now
                                                                            unsafe("#pragma unroll")
                                                                            for k_it in range(0, 64/WGMMA_K, 1):
                                                                                intrinsic.wgmma256(sA_1_producer[64*(m_it + wg_idx*B_WG_M/WGMMA_M)*WGMMA_M + index_a2 + k_it*WGMMA_K], sB_1_producer[(index_b2 + (k_it * WGMMA_K))], 1, 1, 1, 0, 0, out=[d[m_it]])

                                                                    with wgmma_async():
                                                                        for m_it in range(0, B_WG_M/WGMMA_M, 1):
                                                                            index_a3 : int @ warp[4] = m_it * BM
                                                                            index_b3 : int @ warp[4] = m_it * BN
                                                                            unsafe("#pragma unroll")
                                                                            for k_it in range(0, 64/WGMMA_K, 1):
                                                                                intrinsic.wgmma256(sA_2_producer[64*(m_it + wg_idx*B_WG_M/WGMMA_M)*WGMMA_M + index_a3 + k_it*WGMMA_K], sB_2_producer[(index_b3 + (k_it * WGMMA_K))], 1, 1, 1, 0, 0, out=[d[m_it]])

                                                                    with wgmma_async():
                                                                        for m_it in range(0, B_WG_M/WGMMA_M, 1):
                                                                            index_a4 : int @ warp[4] = m_it * BM
                                                                            index_b4 : int @ warp[4] = m_it * BN
                                                                            unsafe("#pragma unroll")
                                                                            for k_it in range(0, 64/WGMMA_K, 1):
                                                                                intrinsic.wgmma256(sA_3_producer[64*(m_it + wg_idx*B_WG_M/WGMMA_M)*WGMMA_M + index_a4 + k_it*WGMMA_K], sB_3_producer[(index_b4 + (k_it * WGMMA_K))], 1, 1, 1, 0, 0, out=[d[m_it]])
                                                            case 4:
                                                                pass

                    with claim(block_C, scope=warp[8], offset=num_block_n*BN*M + num_block_m*BM) as c_consumer:                 
                        match warp:
                            case 8:
                                with partition(c_consumer, dimension=thread[1], offset=0) as c_thrd:                 
                                    with group(warp[4]):
                                        warp_id: int @ warp[1] = id()
                                        with group(warp[1]):
                                            lane : int @ thread[1] = id()
                                            row : int @ thread[1] = warp_id * 16  + lane / 4
                                            unsafe("#pragma unroll")
                                            for m_it in range(0,  B_WG_M/WGMMA_M, 1):
                                                yo : int @ warp[1]  = m_it*WGMMA_M + wg_idx*B_WG_M
                                                with group(thread[1]):
                                                    if (row + yo + num_block_m*BM >= M):
                                                        continue
                                                    for w in range(0, WGMMA_N, 16):
                                                        if (w < w < N-num_block_n*BN):
                                                            col : int @ thread[1] =  w + 2*(tid % 4); 
                                                            c_thrd[col * M + row + yo] = d[m_it][w/16][0]
                                                            c_thrd[(col + 1) * M + (row) + yo] = d[m_it][w/16][1]
                                                            c_thrd[(col) * M + (row + 8) + yo] = d[m_it][w/16][2]
                                                            c_thrd[(col + 1) * M + (row + 8) + yo] = d[m_it][w/16][3]
    
                                                            c_thrd[(col + 8) * M + (row + 0) + yo] = d[m_it][w/16][4]
                                                            c_thrd[(col + 9) * M + (row + 0) + yo] = d[m_it][w/16][5]
                                                            c_thrd[(col + 8) * M + (row + 8) + yo] = d[m_it][w/16][6]
                                                            c_thrd[(col + 9) * M + (row + 8) + yo] = d[m_it][w/16][7]
                            case 4:
                                pass
     
    return
\end{minted}